\documentclass[twocolumn,prb,amsmath,amssymb]{revtex4-2}
\usepackage{amsthm}
\usepackage{booktabs}
\usepackage{array}[=2016-10-06]
\usepackage{graphicx}
\usepackage[colorlinks=true,linkcolor=blue,citecolor=blue,urlcolor=blue]{hyperref}

\newcommand{\vk}{\mathbf{k}}
\newcommand{\vd}{\boldsymbol{\delta}}

\newtheoremstyle{phys}
  {\medskipamount}
  {\medskipamount}
  {\upshape}
  {0pt}
  {\bfseries}
  {.}
  {0.5em}
  {\thmname{#1}\thmnumber{ #2}\thmnote{ (#3)}}
\theoremstyle{phys}
\newtheorem{theorem}{Theorem}
\newtheorem{lemma}{Lemma}
\newtheorem{corollary}{Corollary}
\newtheorem{proposition}{Proposition}
\newtheorem{result}{Result}

\newtheorem{remark}{Remark}
\newtheorem{prediction}{Prediction}

\begin{document}

\title{Two superconducting pairing models with a Kitaev spin-liquid glue:\\
microscopic kernel and parity alternation on odd-fold bond stars}

\author{Youping Dai}
\email{daiyouping@cuhk.edu.cn}
\affiliation{School of Science and Engineering, The Chinese University of Hong Kong, Shenzhen, China}

\begin{abstract}
We construct two superconducting pairing models on the doped Kitaev
honeycomb lattice, whose microscopic kernels are derived rather than
assumed. Both models are driven by a bond-local glue that is neutral
between the even- and odd-parity pairing channels; the resulting
chirality--parity linkage is therefore set by the bond-star geometry,
not by the microscopic glue. 

(i) Model~1(odd-parity sector): the microscopic kernel follows from the flux-sector
selection rule of the pure model, which fixes its algebraic
form---component-diagonal, nearest-neighbor separable, rank one in spin
space per bond---with $w_\alpha=J_K^2\chi_\alpha\propto r_\alpha
m_\alpha$, where the equal-time factor $m_\alpha$ follows from the
Hellmann--Feynman theorem and the retardation factor
$r_\alpha\simeq0.24$ is computed numerically with an analytic
cross-check. The resulting spin-vertex state is a unitary triplet with
a bond-locked anisotropic $\mathbf d$ vector, and its gap is a Gram sum
of bond sines, so nodal directions are fixed by geometry and do not move
under strain, pressure, or doping.

(ii) Model~2(even-parity sector): the kernel is fixed by the $C_3$ little-group
selection rule of a zone-corner chiral phonon, up to a Cooper-channel
extension whose momentum constraint we state explicitly. The
charge-vertex pairing is a $d+id$ singlet: the parity decomposition
$f_\pm=g_\pm+p_\pm$ projects the charge channel onto $g_\pm$, so the
chirality is slaved to the phonon polarization, with $C=\pm4$ and
$\kappa_{xy}/T=2\kappa_0$. Reversing the pump helicity reverses the
Chern number, Kerr rotation, and thermal Hall sign continuously.

(iii) General theorem: on odd-fold bond stars, the
angular-momentum ladder of a single-star chiral form factor alternates in
parity. Every chiral state selected by such a bond-local form factor is thus a
geometrically enforced singlet--triplet composite with universal amplitude
ratios ($n=3$: $d\oplus p$; $n=5$: $g\oplus p$ or $d\oplus f$); coherent
doubled-star combinations can restore parity purity. The weak-pairing Chern
number of the mixed state
follows the dominant on-shell component across a topological phase
boundary. The spin glue drives $C=\pm2$, while the charge kernel
keeps $C=\pm4$; the thermal Hall plateau discriminates the
microscopic vertex.\\

\end{abstract}
\maketitle

\section{Introduction}
\label{sec:intro}

The Kitaev honeycomb model is an exactly solvable model of a quantum spin liquid~\cite{Kitaev2006,Trebst2022}.
Its exchange is bond-dependent: on each of the three bond orientations of the honeycomb
lattice only one spin component couples. This
anisotropy underlies both its exact solution by Majorana
fermionization and its role in topological quantum
computation~\cite{Kitaev2001,Kitaev2006}. On the materials side,
spin-orbit-driven $4d/5d$ magnets such as $\alpha$-RuCl$_3$
~\cite{Banerjee2016,Kasahara2018}, Na$_2$IrO$_3$, and the
harmonic-honeycomb iridates $\beta/\gamma$-Li$_2$IrO$_3$
approximately realize this exchange~\cite{Winter2017,Trebst2022}, and evidence for a half-quantized thermal Hall
plateau in $\alpha$-RuCl$_3$ has intensified interest in its Majorana
physics~\cite{Kasahara2018}.

Once the lattice is doped and the local moments are coupled to itinerant
electrons by a Kondo coupling $J_K$, the same anisotropy takes on a
second role: the spin fluctuations of the liquid become a candidate
pairing glue, weighted by spin component and bond direction rather than
by momentum alone, and the doped system is placed in the broader
fractionalized-metal setting of a spin liquid coexisting with a
metal~\cite{Senthil2003}.
Spin-fluctuation theory has long proposed that a quantum liquid of
local moments can mediate superconductivity, as in heavy-fermion and
high-$T_c$ superconductors~\cite{BerkSchrieffer1966,%
Monthoux1991,Chubukov2008,Coleman2015}.
For the doped Kitaev model itself, superconductivity has been studied
for over a decade, and a consistent sign rule has emerged: ferromagnetic
Kitaev exchange favors triplet $p$-wave pairing coexisting with
ferromagnetism, while antiferromagnetic exchange turns to singlet $d+id$
at intermediate doping. This picture emerged at the mean-field
and wave-function level~\cite{You2012,Hyart2012,Okamoto2013} and has
recently received support from variational Monte Carlo calculations whose
wave function can represent the half-filled Kitaev spin liquid
exactly~\cite{Misawa2026}. Other numerics are more cautious: tensor-network studies of
the doped triangular-lattice Kitaev model find superconducting
correlations strongly suppressed~\cite{Hardy2026}, and on the honeycomb
lattice a pair-density-wave precursor rather than long-range order has
been reported~\cite{Peng2021}.  In parallel, chiral
phonons---circularly polarized vibrational modes carrying angular
momentum---have been proposed as a bosonic glue capable of mediating
chiral pairing~\cite{ZhangNiu2015,Gao2023}.
Mixed singlet--triplet chiral pairing with large Chern numbers $C=\pm2,\pm4$
has also been reported in moir\'e bilayers~\cite{Akbar2024}.
A closely related mechanism is the Majorana-mediated fractionalized
superconductivity of the Kitaev--Kondo lattice~\cite{Bunney2025}, in
which the glue selects chiral $d$- or $p$-wave pairing according to the
sign of the coupling. These works show that the pairing products of
Kitaev-type interactions are sign-sensitive and typically chiral. They
ask, for specific models, whether superconductivity appears and with
which symmetry.

We ask a different question: on which lattice structures are singlet
and triplet pairing forced to coexist, and on which are they forbidden? The answer depends only on the parity
of the bond-star arity, not on model details. Our results are of three kinds:

The first is microscopic, preparing the ground for what follows: we
derive the symmetry-dictated form of the pairing kernel rather than
assume it; for Model~1 the derivation is microscopic end to end, while
for Model~2 we state explicitly the momentum constraint under which
the microscopic kernel reduces to it. The exact
short-rangedness of the spin correlations of the pure Kitaev
model---a flux-sector selection rule valid at all frequencies---forces
the itinerant-electron kernel at second order in $J_K$ to be diagonal
in the spin component, nearest-neighbor separable, and rank one per
bond, and the Hellmann--Feynman theorem fixes its weight function
$w_\alpha=J_K^2\chi_\alpha\propto r_\alpha m_\alpha$ in closed form.

The second is a pair of solved models. In Model~1 the anisotropic
spin-fluctuation glue enters through the spin vertex and produces a
unitary triplet state with a bond-locked anisotropic $\mathbf d$ vector,
each component condensing with its own bond-sine form factor.
In Model~2 the pairing is mediated by zone-corner chiral phonons through the charge
vertex, and angular-momentum matching selects the chiral combination
$d_{x^2-y^2}\pm id_{xy}$ from the degenerate $d$-wave doublet. The two models are two parity sectors of a single chiral-boson
structure: the same phonon, projected through the spin vertex, would
select the chiral triplet $p\pm ip$. A structural rigidity is common to
both: the gap magnitude is a Gram sum of bond form factors, so its
zeros are fixed by the bond-star geometry rather than by the pairing
amplitudes. Nodal directions therefore do not move under strain,
pressure, or doping; only a longer-range coupling such as $J_3$, which
adds next-neighbor form factors, can displace them.

The third is a general theorem. On odd-fold
bond stars the angular-momentum ladder of a chiral bond form factor
alternates in parity, so chirality enforces singlet--triplet mixing with
geometrically fixed content and universal amplitude ratios; a bond-local
glue mediates the two parity sectors with locked orbital strength; and the
weak-pairing Chern number of the mixed state follows the dominant
on-shell component, across an exact topological phase boundary inside
the chiral phase. No spin-orbit coupling or inversion-symmetry breaking
is required: the mixing is enforced by the odd-fold bond star itself, in
a centrosymmetric crystal, in contrast to the familiar
noncentrosymmetric route in which antisymmetric spin-orbit coupling
splits the bands and mixes singlet with
triplet~\cite{GorkovRashba2001,BauerSigrist2012}.

Section~\ref{sec:background} collects the background and standard tools
we use. Section~\ref{sec:kitaev} analyses the Kitaev exchange-tensor structure
and the anisotropic spin fluctuations that follow from it. Sections~\ref{sec:micro}
and~\ref{sec:wfunc} derive the pairing kernel and its weight function.
Section~\ref{sec:model1} constructs and solves Model~1;
Sec.~\ref{sec:offdiag} treats the off-diagonal exchanges $\Gamma,\Gamma',J_3$
of real materials, and Sec.~\ref{sec:screening} the resulting materials
screening. Section~\ref{sec:model2} constructs and solves Model~2.
Section~\ref{sec:alternation} proves the parity-alternation classification
and the topological counting rule for mixed states, and applies them to
quasicrystal superconductors (Sec.~\ref{sec:qcapp}).
Section~\ref{sec:compete} develops the competing-order phase diagram that
follows from sharing one glue kernel. Section~\ref{sec:discuss} positions the
results relative to known work, collects the mutually exclusive
experimental fingerprints of the two states, and summarizes the key points
for experiment and computation; Sec.~\ref{sec:conclusion} concludes. The appendices give the BdG diagonalization
(Appendix~\ref{app:bdg}), the parity decomposition and the exact projection
of the chiral kernel onto the $E_2$ doublet (Appendix~\ref{app:proj}), the
Ginzburg--Landau analysis of chirality slaving (Appendix~\ref{app:gl}), the
flux-sector argument (Appendix~\ref{app:flux}), the Brillouin-zone
asymptotics of the weight function (Appendix~\ref{app:bz}), the bond-tensor
diagonalization (Appendix~\ref{app:tensor}), and the numerical verification
of the bond-star theorems (Appendix~\ref{app:verify}).

\section{Background: definitions, tools, and standard results}
\label{sec:background}

This section collects the standard results we use: each subsection defines its
objects, and almost every result is either cited or derived here.

\subsection{The Kitaev honeycomb model and its exact solution}
\label{sec:bgkitaev}

The Kitaev honeycomb model~\cite{Kitaev2006} places one
spin-$1/2$ on each site of the two-dimensional honeycomb lattice, with a
bond-dependent Ising-like exchange:
\begin{equation}
H_K=-\sum_{\alpha\in\{x,y,z\}}J_\alpha
\sum_{\langle ij\rangle_\alpha}\sigma_i^\alpha\sigma_j^\alpha ,
\label{eq:kitaev}
\end{equation}
where $\langle ij\rangle_\alpha$ denotes the nearest-neighbor bonds of
orientation $\alpha$ (the three orientations tile the honeycomb lattice, each
site touching one bond of each type), and $\sigma_i^\alpha$ are Pauli
matrices. Our convention throughout is that the $\sigma$'s are Pauli matrices
(eigenvalues $\pm1$), so that $J_\alpha>0$ is ferromagnetic coupling; the
materials literature often writes spin operators $\mathbf S=\boldsymbol
\sigma/2$ instead, and we flag the switch explicitly where it occurs
(Sec.~\ref{sec:offdiag}). Two properties of Eq.~\eqref{eq:kitaev} matter for
us. First, the exchange is \emph{compass-like}: on an $\alpha$ bond only the
$\alpha$ spin components interact, so the interaction matrix in spin space is
the rank-one tensor $J_\alpha\delta_{\mu\alpha}\delta_{\nu\alpha}$ (used in
Sec.~\ref{sec:kitaev}). Second, the model is frustrated in an extreme way:
terms sharing a site anticommute, so no classical configuration minimizes all
bonds simultaneously---this is what stabilizes the spin liquid
\cite{Kitaev2006,Trebst2022}.

The exact solution proceeds by enlarging the
Hilbert space~\cite{Kitaev2006,Trebst2022}. At each site introduce four
Majorana fermions $b_i^x,b_i^y,b_i^z,c_i$: Hermitian operators obeying
\begin{equation}
\{b_i^\alpha,b_j^\beta\}=2\delta_{ij}\delta_{\alpha\beta},\qquad
\{c_i,c_j\}=2\delta_{ij},\qquad
\{b_i^\alpha,c_j\}=0 .
\label{eq:majanticomm}
\end{equation}
(Majorana fermions are ``real'' fermions, $\gamma^\dagger=\gamma$,
$\gamma^2=1$; two Majoranas combine into one complex fermion,
Eq.~\eqref{eq:majorana}.) One checks directly from
Eq.~\eqref{eq:majanticomm} that the operators
\begin{equation}
\tilde\sigma_i^\alpha\equiv i\,b_i^\alpha c_i
\label{eq:sigrep}
\end{equation}
are Hermitian, square to $1$, and obey the Pauli algebra
$\tilde\sigma^x\tilde\sigma^y=i\tilde\sigma^z$ and cyclic permutations. The
enlarged Hilbert space contains unphysical states; the physical subspace is
selected by the constraint $D_i\equiv b_i^xb_i^yb_i^zc_i=+1$ at every site,
and all physical observables commute with $D_i$, so matrix elements computed
within the physical sector are exact~\cite{Kitaev2006,Pedrocchi2011}.

Substituting Eq.~\eqref{eq:sigrep} into
Eq.~\eqref{eq:kitaev}, the bond term becomes
\begin{equation}
-\sigma_i^\alpha\sigma_j^\alpha
=-(ib_i^\alpha c_i)(ib_j^\alpha c_j)
=u_{ij}\,(ic_ic_j),\qquad u_{ij}\equiv i b_i^\alpha b_j^\alpha ,
\label{eq:udef}
\end{equation}
so that
\begin{equation}
H_K=\frac{i}{4}\sum_{\langle ij\rangle_\alpha}\hat J_\alpha u_{ij}\,c_ic_j,
\end{equation}
with $\hat J_\alpha$ absorbing the convention-dependent factor. The bond
operators $u_{ij}$ have eigenvalues $\pm1$ and commute with $H_K$ and with
each other~\cite{Kitaev2006}: they are a \emph{static} $Z_2$ gauge field. The
gauge-invariant content of a configuration $\{u_{ij}\}$ is the set of
plaquette fluxes
\begin{equation}
W_p\equiv\prod_{(ij)\in\partial p}u_{ij}=\pm1 ,
\label{eq:wp}
\end{equation}
the product of the bond variables around plaquette $p$; configurations
differing only by gauge transformations $b_i^\alpha\to-b_i^\alpha$ (which send
$u_{ij}\to-u_{ij}$ on all bonds at $i$) are physically equivalent. A
plaquette with $W_p=-1$ is said to carry a \emph{$Z_2$ flux} (a ``vison'').
By Lieb's theorem applied to the reflection-symmetric honeycomb lattice, the
ground state lies in the zero-flux sector, $W_p=+1$ on every
plaquette~\cite{Kitaev2006}.

In any fixed gauge sector the
$c_i$ fermions are free: $H_K$ is a quadratic Majorana hopping problem,
diagonalizable by Fourier transformation~\cite{Kitaev2006}. In the zero-flux
sector with the two-site unit cell of the honeycomb lattice, the Bloch
Hamiltonian has spectrum $\pm|g(\mathbf k)|$ with
\begin{equation}
g(\mathbf k)=J_z+J_x\,e^{ik_1}+J_y\,e^{ik_2},
\label{eq:g}
\end{equation}
where $k_1,k_2$ are the crystal momenta along the two Bravais vectors (unit
bond length throughout). The excitation gap closes where $g(\mathbf k)=0$,
which has solutions if and only if the three couplings obey the triangle
inequalities $|J_x|\le|J_y|+|J_z|$ etc.: this is the \emph{gapless B phase},
a spin liquid with a Dirac cone of matter Majoranas; otherwise the system is
in the \emph{gapped A phase}, equivalent to the toric code and supporting
Abelian anyons, which a magnetic field renders non-Abelian
\cite{Kitaev2006,Trebst2022}. The boundary between A and B phases, e.g.\
$J_z=2J$ at $J_x=J_y=J$, will reappear in the weight function of
Sec.~\ref{sec:wfunc}. The relation of an ordinary (complex) fermion to
Majoranas,
\begin{equation}
c_{i\uparrow}=\tfrac12(\gamma_{i1}+i\gamma_{i2}),\qquad
c_{i\downarrow}=\tfrac12(\gamma_{i3}+i\gamma_{i4}),
\label{eq:majorana}
\end{equation}
makes the spin correlations of the pure Kitaev model exactly computable: a
spin operator [Eq.~\eqref{eq:sigrep}] is one matter Majorana times one bond
Majorana, and both sectors are free or static. This is the technical basis for the pairing kernel
derived below (Sec.~\ref{sec:micro}).

The property of the model most relevant here is that its
ground-state spin correlations are
\emph{ultrashort-ranged and bond-diagonal}: $\langle\sigma_i^\mu
\sigma_j^\nu\rangle$ vanishes unless $\mu=\nu$ and the two sites are either
coincident or joined by a $\mu$ bond~\cite{Kitaev2006,Baskaran2007}. The
dynamical spin structure factor is correspondingly gapped by the flux-pair
energy and nonzero only above a threshold~\cite{Knolle2014}. Both statements
are proved by the flux-counting argument reproduced in full in
Sec.~\ref{sec:micro}A and Appendix~\ref{app:flux}; we do not quote them as
black boxes.

Exact solvability belongs to the pure model
\eqref{eq:kitaev}. Once Kondo coupling and doping are introduced the $u_{ij}$
no longer commute with the Hamiltonian, fluxes become dynamical, and the full
model is not integrable (Note~3 in Sec.~\ref{sec:bdg1}; numerical treatments
are then required~\cite{Keskiner2025}). Our use of the exact solution is
strictly limited to computing the correlators of the undoped ground state
that build the pairing kernel.

\subsection{From a bosonic bath to a pairing interaction}
\label{sec:bgglue}

Both models follow the same general pattern: itinerant
electrons couple linearly to a bosonic bath (spin fluctuations of the
local moments in Model~1, chiral phonons in Model~2), and integrating
out the bath at second order produces a retarded electron-electron
interaction. We collect
the standard steps here once, with citations, so that
Secs.~\ref{sec:micro} and~\ref{sec:model2} can quote them without
re-derivation.

For a coupling
$H'=g\sum_i O_i^{e}\,O_i^{b}$ between an electron operator $O^e$ and a bath
operator $O^b$, second-order perturbation theory in $g$ eliminates the bath
and generates the effective electron-electron interaction
\begin{equation}
V_{\mathrm{eff}}(\omega)=g^2\,\chi_{bb}(\omega),
\label{eq:genericglue}
\end{equation}
where $\chi_{bb}$ is the retarded correlation function of the bath operator.
This is the standard paramagnon/boson-exchange construction, derived e.g.\
in Refs.~\cite{BerkSchrieffer1966,Monthoux1991} for spin fluctuations and in
the textbook treatment of Ref.~\cite{Coleman2015} (Ch.~13 and App.); for
phonons it reduces to the familiar $g^2D(\mathbf q,\omega)$ with $D$ the
phonon propagator~\cite{Coleman2015}.

Whether the induced interaction is
attractive in the singlet or triplet channel is fixed by the spin algebra of
the vertex. For a spin-dependent bath coupling $\mathbf S_i\cdot\boldsymbol
\sigma$, the exchange of the bath between two electrons attaches the factor
$\boldsymbol\sigma_1\cdot\boldsymbol\sigma_2$ to the pair. Its eigenvalues on
the two-particle spin states follow from
$(\boldsymbol\sigma_1+\boldsymbol\sigma_2)^2=6+2\,\boldsymbol\sigma_1\!\cdot
\!\boldsymbol\sigma_2=4S(S+1)$:
\begin{equation}
\boldsymbol\sigma_1\!\cdot\!\boldsymbol\sigma_2
=2S(S+1)-3=
\begin{cases}
+1,&S=1\ \text{(triplet)},\\
-3,&S=0\ \text{(singlet)}.
\end{cases}
\label{eq:fierz}
\end{equation}
Thus ferromagnetic fluctuations ($\chi>0$ with our conventions) are
attractive in the triplet channel and repulsive in the singlet channel, with
a relative factor $3$; antiferromagnetic fluctuations reverse the
pattern~\cite{BerkSchrieffer1966,SigristUeda1991}. This single computation
underlies the sign analysis of Note~2 (Sec.~\ref{sec:bondlock}) and the
phase-diagram estimate of Sec.~\ref{sec:pd1}.

Superconductivity is an
instability of the Fermi surface in the Cooper (particle-particle) channel.
Linearizing the self-consistency (gap) equation at $T_c$ for an interaction
$V(\vk,\vk')$ gives~\cite{Coleman2015}
\begin{equation}
\Delta(\vk)=-\sum_{\vk'}V(\vk,\vk')\,
\frac{\tanh(\xi_{\vk'}/2k_BT_c)}{2\xi_{\vk'}}\,\Delta(\vk'),
\label{eq:lingap}
\end{equation}
with $\xi_{\vk}$ the normal-state dispersion measured from the chemical
potential. The logarithmic divergence of the Cooper susceptibility,
$\int_0^{\omega_D}\!d\xi\,\tanh(\xi/2T)/\xi\simeq\ln(1.13\,\omega_D/k_BT)$
for a cutoff $\omega_D$ \cite{Coleman2015}, converts even an infinitesimal
attractive channel eigenvalue $\lambda$ into a finite
\begin{equation}
T_c=1.13\,\omega_D\,e^{-1/\lambda},
\label{eq:bcs}
\end{equation}
the BCS amplification that will recur whenever we quote eigenvalue ratios
[e.g.\ Eq.~\eqref{eq:eigratio}]. For an attractive \emph{separable} kernel
of rank one, $V(\vk,\vk')=-v\,f(\vk)f^\ast(\vk')$ with $v>0$, substitution
into Eq.~\eqref{eq:lingap} shows immediately that the only eigenvector with
nonzero eigenvalue is $f(\vk)$ itself, with channel eigenvalue
\begin{equation}
\lambda=v\,N(0)\,\big\langle|f|^2\big\rangle_{\mathrm{FS}},
\label{eq:rankone}
\end{equation}
where $N(0)$ is the density of states per spin at the Fermi level and
$\langle\cdot\rangle_{\mathrm{FS}}$ the Fermi-surface average. The condition
$1=V\chi$ for the instability temperature of a given channel is the Thouless
criterion used in Corollary~\ref{cor:seq} and Sec.~\ref{sec:pd1}. The
particle-hole channel obeys the analogous criterion with the corresponding
susceptibility (Sec.~\ref{sec:compete}).

\subsection{The pairing matrix: singlet-triplet decomposition}
\label{sec:bgmatrix}

A general spin-dependent pair
amplitude is a $2\times2$ complex matrix $\hat\Delta(\vk)$ in spin space.
Anticommutation of the electron operators imposes
\begin{equation}
\hat\Delta(\vk)=-\hat\Delta^{\mathsf T}(-\vk).
\label{eq:antisym}
\end{equation}
The four-dimensional complex space of $2\times2$ matrices splits into the
antisymmetric matrix $i\sigma_y$ and the three symmetric matrices $\sigma_i
i\sigma_y$ ($i=x,y,z$). Symmetry is checked directly: $(i\sigma_y)^{\mathsf
T}=-i\sigma_y$, while for each Pauli matrix $(\sigma_i i\sigma_y)^{\mathsf
T}=(i\sigma_y)^{\mathsf T}\sigma_i^{\mathsf T}=+\sigma_i i\sigma_y$ [using
$\sigma_{x,z}^{\mathsf T}=\sigma_{x,z}$, $\sigma_y^{\mathsf T}=-\sigma_y$,
and the anticommutation of distinct Pauli matrices]. Equation
\eqref{eq:antisym} then forces the expansion
\begin{equation}
\hat\Delta(\vk)=\big[\Delta_0(\vk)+\mathbf d(\vk)
\cdot\boldsymbol\sigma\big]i\sigma_y ,
\label{eq:gapexp}
\end{equation}
with an \emph{even} singlet amplitude $\Delta_0(\vk)=\Delta_0(-\vk)$ and an
\emph{odd} triplet vector $\mathbf d(\vk)=-\mathbf d(-\vk)$
\cite{SigristUeda1991,Vollhardt1990}. The pairing-matrix space is thereby
exhausted: every two-component pairing state, in any material, is a point in
this singlet$\,\oplus\,$triplet space. This is what Sec.~\ref{sec:intro} means by that both models ``remain within the
singlet-triplet framework''.

For a triplet state, the
vector $\mathbf d$ encodes the spin structure of the pair: the component
$d_i$ is the amplitude of the spin-zero projection along axis $i$, so a state
with $\mathbf d\parallel\hat z$ pairs the combinations
$|\!\uparrow\uparrow\rangle$ and $|\!\downarrow\downarrow\rangle$ with
respect to the $z$ axis (equal-spin pairing) and has no $S_z=0$ component;
the mapping $\hat\Delta=\mathbf d\cdot\boldsymbol\sigma\,i\sigma_y$ makes
this explicit, since $\sigma_z i\sigma_y=\begin{pmatrix}1&0\\0&-1\end{pmatrix}$
is diagonal in spin while $\sigma_{x,y}i\sigma_y$ are off-diagonal
\cite{SigristUeda1991}. A triplet state is called \emph{unitary} if
$i\,\mathbf d\times\mathbf d^\ast=0$, equivalently if $\mathbf d$ has a
common overall phase times a real vector; otherwise the two spin populations
see different gaps $|\mathbf d|\pm|\mathbf d\times\mathbf d^\ast|/|\mathbf
d|$ [the spectrum formula \eqref{eq:generalE} with $\Delta_0=0$;
Refs.~\cite{SigristUeda1991,Vollhardt1990}]. Unitary triplets preserve
time-reversal symmetry up to the overall phase and therefore show no
spontaneous Kerr rotation or thermal Hall signal---this is Corollary
\ref{cor:unitary} of Model~1.

The Knight shift is the NMR frequency shift produced by
the spin polarization of the electrons, and its drop below $T_c$ measures the
spin susceptibility $\chi_s$ of the condensate~\cite{SigristUeda1991}. The
two canonical behaviors are: (i) singlet---$\chi_s\to0$ as $T\to0$ because
the pairs are spin singlets (the Yosida-function suppression); (ii) unitary
triplet with $\mathbf d$ vector---$\chi_s$ is a tensor, unchanged for fields
along $\mathbf d$ (the pairs have no spin projection along $\mathbf d$ to
break) and suppressed for fields perpendicular to it
\cite{SigristUeda1991,Vollhardt1990}. An anisotropic, bond-locked $\mathbf d$
vector (Model~1) therefore predicts a bond-resolved anisotropic Knight shift
(Sec.~\ref{sec:sign1}).

\subsection{BdG theory, Chern numbers, edge modes, and the thermal Hall conductance}
\label{sec:bgbdg}

Mean-field pairing replaces the interaction by the
anomalous expectation values $\hat\Delta(\vk)$; defining the Nambu spinor
$\Psi_{\vk}=(c_{\vk\uparrow},c_{\vk\downarrow},
c^\dagger_{-\vk\uparrow},c^\dagger_{-\vk\downarrow})^{\mathsf T}$ (sublattice
indices suppressed; see Appendix~\ref{app:bdg} for the multiband status), the
quadratic Hamiltonian is $H_{\mathrm{BdG}}=\frac12\sum_{\vk}
\Psi_{\vk}^\dagger H_{\mathrm{BdG}}(\vk)\Psi_{\vk}$ with
\begin{equation}
H_{\mathrm{BdG}}(\vk)=
\begin{pmatrix}
\xi_{\vk}\,\mathbf 1 & \hat\Delta(\vk)\\
\hat\Delta^\dagger(\vk) & -\xi_{\vk}\,\mathbf 1
\end{pmatrix}.
\label{eq:bdg}
\end{equation}
Particle-hole symmetry---the identity
$\Xi H_{\mathrm{BdG}}(\vk)\Xi^{-1}=-H_{\mathrm{BdG}}(-\vk)$ with
$\Xi=\tau_x K$ built from Nambu-space Pauli matrices and complex
conjugation---pairs every eigenvalue $E_{\vk}$ with $-E_{-\vk}$; only the
positive branch is physical~\cite{ReadGreen2000,Nayak2008}. Squaring
Eq.~\eqref{eq:bdg} and using that the scalar $\xi_{\vk}$ commutes with
$\hat\Delta$, the off-diagonal blocks vanish and the spectrum follows from
$\hat\Delta\hat\Delta^\dagger$ alone; the full algebra is displayed in
Appendix~\ref{app:bdg} and gives
\begin{equation}
E_{\vk,\pm}^2=\xi_{\vk}^2+|\Delta_0|^2+|\mathbf d|^2
\pm|\mathbf B|,
\label{eq:generalE}
\end{equation}
with
$\mathbf B=\Delta_0\mathbf d^\ast+\Delta_0^\ast\mathbf d
+i\,\mathbf d\times\mathbf d^\ast$.
For a real kernel (relative phase $0$ or $\pi$) this reduces to
$E^2=\xi^2+(|\Delta_0|\pm|\mathbf d|)^2$; for the unitary $\pi/2$-locked case
$\mathbf B=0$ and the square-sum form is recovered (Appendix~\ref{app:bdg}).

A $2\times2$ Nambu block with
a chiral gap $\Delta(\vk)=|\Delta(\vk)|e^{i\phi(\vk)}$ defines a texture
$\hat{\mathbf d}(\vk)$ of the effective two-level Hamiltonian
$H_{\mathrm{block}}=\boldsymbol{\mathcal D}(\vk)\cdot\boldsymbol\tau$, and
its Chern number---the Berry flux of the lower band through the Brillouin
zone---is
\begin{equation}
C_{\mathrm{block}}=\frac{1}{4\pi}\int_{\mathrm{BZ}}d^2k\;
\hat{\mathbf d}\cdot\big(\partial_{k_x}\hat{\mathbf d}\times
\partial_{k_y}\hat{\mathbf d}\big),
\label{eq:chernblock}
\end{equation}
the integer counting how many times the texture wraps the Bloch sphere
\cite{ReadGreen2000,Nayak2008}. We fix the sign convention
$C_{\mathrm{block}}=+1$ for $\Delta\propto k_x+ik_y$ at $\mu>0$
\cite{ReadGreen2000}. In the weak-pairing limit ($|\Delta|$ much smaller than
the Fermi energy, $\mu>0$) the texture is trivial except in a thin shell
around the Fermi surface, where $\xi_{\vk}\approx0$; the integral then
collapses to the winding number of the gap phase along the Fermi surface,
\begin{equation}
C_{\mathrm{block}}=\frac{1}{2\pi}\oint_{\mathrm{FS}}d\phi(\vk),
\label{eq:weakwinding}
\end{equation}
the Read--Green formula~\cite{ReadGreen2000}. Every Chern-number statement below is an application of Eq.~\eqref{eq:weakwinding} to an explicitly
known gap function, so the topology of our states reduces to counting the
winding of a trigonometric polynomial---the task solved exactly in
Sec.~\ref{sec:counting}.

A two-dimensional superconductor with block Chern numbers summing to $C$
supports $|C|$ chiral Majorana edge branches per edge
\cite{ReadGreen2000,Volovik2003,Nayak2008}. Each chiral Majorana branch is a
one-dimensional chiral conformal field theory of central charge $c=1/2$
carrying the thermal Hall (heat) current of such a mode; summing the
branches gives the quantized low-temperature thermal Hall conductance
\begin{equation}
\frac{\kappa_{xy}}{T}=c\,\frac{\pi^2k_B^2}{3h}
\equiv c\,\kappa_0,\qquad c=\frac{C}{2},
\label{eq:thallunit}
\end{equation}
where $\kappa_0\equiv\pi^2k_B^2/(3h)$ is the thermal Hall quantum per unit
central charge~\cite{ReadGreen2000,Volovik2003}; the half-quantized plateau
reported for $\alpha$-RuCl$_3$ corresponds to $c=1/2$, a single Majorana
branch~\cite{Kasahara2018}. Equation \eqref{eq:thallunit} is why we track
$C$: the thermal Hall plateau reads off the relative angular momentum of the
Cooper pair. ``Surface Fermi arcs'', by contrast, are a property of
three-dimensional (e.g.\ layered) generalizations with Weyl or Weyl-nodal
structure~\cite{Wan2011}; the boundary of a strictly two-dimensional chiral
superconductor is a one-dimensional chiral edge mode.

\subsection{Three mathematical identities used throughout}
\label{sec:bgmath}

The analysis below repeatedly uses three elementary identities for
bond form factors. We state each with its proof, so that later sections
can quote them by name.

\emph{(i) Roots-of-unity filter.} For integers $n\ge2$ and $q$,
\begin{equation}
\sum_{j=0}^{n-1}e^{2\pi i q j/n}=n\,\delta_{q\equiv0\,(\mathrm{mod}\,n)}.
\label{eq:rootfilter}
\end{equation}
\emph{Proof:} if $e^{2\pi iq/n}=1$ every term is $1$ and the sum is $n$;
otherwise it is the geometric series
$(e^{2\pi iq}-1)/(e^{2\pi iq/n}-1)=0$. \qed

\emph{(ii) Jacobi--Anger expansion.}
\begin{equation}
e^{ik\cos\alpha}=\sum_{h=-\infty}^{\infty}i^{h}J_{h}(k)\,e^{ih\alpha},
\label{eq:jacobianger}
\end{equation}
with $J_h$ the Bessel function of the first kind. \emph{Proof:} the
generating function of the Bessel functions,
$e^{k(t-t^{-1})/2}=\sum_h J_h(k)\,t^h$ (definition of $J_h$ as Laurent
coefficients), evaluated at $t=ie^{i\alpha}$ gives
$k(t-t^{-1})/2=ik\cos\alpha$. \qed\ The symmetry
$J_{-h}(k)=(-1)^hJ_h(k)$ [read off the generating function under
$t\to-1/t$] implies
\begin{equation}
i^{h}J_{h}(k)=i^{|h|}J_{|h|}(k)\qquad(h<0),
\label{eq:besselsym}
\end{equation}
the identity that makes all ladder Fourier coefficients positive multiples of
$i^{|h|}$ in Lemma~\ref{lem:circle}. At small argument,
\begin{equation}
J_\nu(k)=\frac{(k/2)^\nu}{\nu!}\big[1+O(k^2)\big],
\qquad \nu\ge0,
\label{eq:besselsmall}
\end{equation}
from the Taylor series defining $J_\nu$.

\emph{(iii) Circular Fermi-surface average of a plane wave.} On a circular
Fermi surface of radius $k$,
\begin{equation}
\big\langle e^{i\vk\cdot\vd}\big\rangle_{\mathrm{FS}}
\equiv\frac{1}{2\pi}\int_0^{2\pi}d\theta\,e^{ik|\vd|\cos\theta}
=J_0(k|\vd|),
\label{eq:fsaverage}
\end{equation}
which is the standard integral representation of $J_0$ [and the $h=0$ case of
Eq.~\eqref{eq:jacobianger}]. Equation \eqref{eq:fsaverage} converts every
Fermi-surface overlap of bond phase factors into a Bessel function of the
bond separation; this is the origin of the $J_0$'s in the channel weights
\eqref{eq:Wg}--\eqref{eq:Wp} and \eqref{eq:f2}--\eqref{eq:p2}.

\subsection{Conventions and notation}
\label{sec:notation}

We set $\hbar=k_B=1$ except where units are displayed, and the nearest-neighbor
bond length to $1$ (so momenta are measured in inverse bond lengths and the
zone-corner magnitude is $|K|=\frac{4\pi}{3\sqrt3}\approx2.42$). Pauli matrices are
used for the local moments [Eq.~\eqref{eq:kitaev}] except in
Sec.~\ref{sec:offdiag}, where the materials convention with spin-$1/2$
operators $\mathbf S=\boldsymbol\sigma/2$ is adopted and flagged.
$\langle\cdot\rangle_{\mathrm{BZ}}$ denotes the Brillouin-zone average
normalized to unity; $\langle\cdot\rangle_{\mathrm{FS}}$ the Fermi-surface
average. $J_\nu$ always denotes the Bessel function of the first kind---not
an exchange coupling---and $J_x,J_y,J_z$ (or $J,\Gamma,\Gamma',K$ in
Sec.~\ref{sec:offdiag}) the exchanges. Table~\ref{tab:notation} collects the
recurring symbols.

\begin{table}[t]
\caption{\label{tab:notation}Recurring notation. All symbols are defined at
first use in the text; this table is a lookup aid.}
\begin{ruledtabular}
{\scriptsize\setlength{\tabcolsep}{1pt}
\begin{tabular}{ll}
Symbol & meaning\\
\hline
$J_\alpha$, $J_K$, $t$ & Kitaev exchange, Kondo coupling, hopping\\
$u_{ij}$, $W_p$, $\Delta_F$ & $Z_2$ bond variable, plaquette flux, flux gap\\
$\chi_\alpha$, $m_\alpha$, $r_\alpha$ & bond susceptibility: static, equal-time, dynamical factor\\
$w_\alpha$, $v_\alpha$ & kernel weight $J_K^2\chi_\alpha$; same quantity per bond\\
$\vd_\alpha$, $\phi_\alpha$ & nearest-neighbor bond vector, its azimuth $2\pi\alpha/3$\\
$\hat\Delta$, $\Delta_0$, $\mathbf d$ & pairing matrix, singlet amplitude, triplet vector\\
$f_\pm$, $g_\pm$, $p_\pm$ & chiral form factor and its even/odd parts\\
$\chi_m^{(n)}$, $\chi^{(e/o)}$ & general bond-star form factor; parity-pure combinations\\
$W_\nu(k)$, $\varepsilon$ & on-shell weight $nJ_\nu(k)$; ladder tail\\
$C$, $c$, $\kappa_0$ & BdG Chern number, chiral central charge, $\pi^2k_B^2/3h$\\
$R$ & anisotropy transfer coefficient [Eq.~\eqref{eq:kappa}]\\
$A_J$, $A_K$, $A_w$ & anisotropy measures [Eqs.~\eqref{eq:criterion}--\eqref{eq:Aw}]\\
\end{tabular}}
\end{ruledtabular}
\end{table}

\section{Kitaev exchange tensor and anisotropic fluctuations}
\label{sec:kitaev}

This section makes two structural observations specific to the pairing
problem: the decomposition of the exchange into bond-summed (uniform)
and bond-resolved parts, and the resulting component-weighted
fluctuation ansatz.

\subsection{Exchange-tensor structure}

The exchange on a $\gamma$ bond of Eq.~\eqref{eq:kitaev} is the rank-1 tensor
$T^{(\gamma)}_{\mu\nu}=J_\gamma\,\delta_{\mu\gamma}\delta_{\nu\gamma}$: only
the $\gamma$ spin components interact. There is no canonical operator decomposition of a single-bond term into a
Heisenberg part plus a remainder: any such split depends on the chosen
reference coupling. What is unambiguous is the separation of the
site-summed exchange tensor into its trace (uniform) and traceless
(bond-resolved) parts. Writing $J_\alpha=\bar J+\delta J_\alpha$ with
$\bar J=\tfrac13(J_x+J_y+J_z)$ as parameter bookkeeping, the exchange
tensor summed over the three bonds meeting at one site
\begin{equation}
\begin{split}
S_{\mu\nu}&=\sum_{\gamma}T^{(\gamma)}_{\mu\nu}
=\mathrm{diag}(J_x,J_y,J_z)\\
&=\bar J\,\delta_{\mu\nu}
+\mathrm{diag}(\delta J_x,\delta J_y,\delta J_z)
\end{split}
\label{eq:tbond}
\end{equation}
is bond-diagonal (each diagonal entry comes from exactly one bond type,
and off-diagonal entries vanish) but is \emph{not} isotropic in general;
only at $J_x=J_y=J_z$ does it reduce to $\bar J\delta_{\mu\nu}$. Its
isotropic (trace) part is $\bar J\,\delta_{\mu\nu}$ with
$\bar J=\tfrac13\mathrm{Tr}\,S$. Physically, uniform
($\mathbf q\approx0$) spin fluctuations feel the sum of the exchanges on
the three bonds meeting at a site, whose trace part is the
Heisenberg-like $\bar J$ (anisotropic corrections being
$O(\delta J_\gamma)$), whereas bond-resolved correlations experience the
individual $J_\alpha$. The exact spin correlations of the pure Kitaev model are
ultrashort-ranged and bond-diagonal [Sec.~\ref{sec:bgkitaev};
Refs.~\cite{Kitaev2006,Baskaran2007}]; the pairing kernels below are
therefore naturally bond-resolved, while the uniform channel enters the
competition analysis of Sec.~\ref{sec:pd1}.

The anisotropy tensor
$T_{\alpha\beta}=\mathrm{diag}(\delta J_x,\delta J_y,\delta J_z)$ is symmetric.
Its antisymmetric part---which would be the Dzyaloshinskii--Moriya (DM) cross
product $\mathbf S_i\times\mathbf S_j$, the time-reversal-even chiral
exchange familiar from noncentrosymmetric magnets
\cite{Rau2014,Winter2017}---vanishes identically: the Kitaev model contains
no DM term, and the anisotropic exchange excites a symmetric
spin-quadrupole-type fluctuation structure rather than a wedge-type cross
product. This structural observation underlies the glue chain of Model~1
(Note~1, Sec.~\ref{sec:model1}).

\subsection{Anisotropic spin fluctuations}

When the three exchanges differ, the local-moment fluctuation spectrum inherits the
exchange anisotropy: to leading order the component strengths are weighted by
$J_\alpha$,
\begin{equation}
\chi_{\alpha\beta}(\mathbf q,\omega)=w_\alpha(\mathbf q,\omega)
\,\delta_{\alpha\beta},\qquad w_\alpha\propto J_\alpha/\omega_\alpha ,
\label{eq:chi}
\end{equation}
that is, anisotropic exchange weights the fluctuation strength by component
and by bond direction. This structure is the source of the anisotropic
pairing kernel of Model~1. The weighting receives qualitative support from
the exactly known spin correlations of the pure Kitaev model (bond-diagonal,
component-diagonal, and short-ranged---Sec.~\ref{sec:bgkitaev}
\cite{Kitaev2006,Baskaran2007}); we use Eq.~\eqref{eq:chi} as a
phenomenological ansatz whose microscopic basis is established in
Secs.~\ref{sec:micro}--\ref{sec:wfunc}, where $w_\alpha$ is computed in
closed form.

\section{Microscopic origin of the pairing kernel}
\label{sec:micro}

The itinerant electrons hop on the same honeycomb lattice and couple to the
local moments by a Kondo exchange. The full Hamiltonian is
\begin{align}
H&=H_K+H_t+H_J,\nonumber\\
H_t&=-t\sum_{\langle ij\rangle,\sigma}
(c^\dagger_{i\sigma}c_{j\sigma}+\mathrm{h.c.}),\nonumber\\
H_J&=J_K\sum_i\mathbf
S_i\cdot(\psi_i^\dagger\boldsymbol\sigma\psi_i),
\label{eq:fullH}
\end{align}
with $H_K$ of Eq.~\eqref{eq:kitaev}. Here $\mathbf S_i=\boldsymbol\sigma_i/2$;
the resulting factor $1/4$ in the second-order kernel is absorbed into the
definition of $J_K$ throughout. This section derives the pairing kernel
that results at second order in $J_K$: the flux-sector selection rule of the
pure Kitaev model forces exactly the component-diagonal, nearest-neighbor
separable, rank-one-in-spin-space-per-bond structure on which Model~1 rests. The logic is:
selection rule (Sec.~\ref{sec:micro}A) $\Rightarrow$ kernel structure
(Sec.~\ref{sec:micro}B) $\Rightarrow$ channel eigenvalues
(Sec.~\ref{sec:parityeig}); the weight magnitudes are then computed in
Sec.~\ref{sec:wfunc}.

\subsection{The flux-sector selection rule}
\label{sec:selection}

Write the spins in the Majorana representation, $\sigma_i^\alpha=i
b_i^\alpha c_i$ [Eq.~\eqref{eq:sigrep}], with the bond variables
$u_{ij}=ib_i^\alpha b_j^\alpha$ conserved, $[u_{ij},H_K]=0$
[Eq.~\eqref{eq:udef}]. The ground state lies in the zero-flux sector
(Sec.~\ref{sec:bgkitaev}). The central structural fact is the exact
short-rangedness of the spin correlation function
\cite{Baskaran2007,Knolle2014}: in the zero-flux ground state,
\begin{align}
\chi^{\mu\nu}_{ij}(t)&\equiv-i\theta(t)\,
\big\langle[\sigma_i^\mu(t),\sigma_j^\nu(0)]\big\rangle\;\neq0
\nonumber\\
&\Longleftrightarrow\quad
\mu=\nu\ \ \text{and}\ \ i=j\ \text{or}\ (ij)\ \text{a }\mu\text{-bond}.
\label{eq:selection}
\end{align}
The reason is flux-sector orthogonality, not integrability in any deeper
sense, and the statement holds at every frequency. The argument, spelled out
step by step in Appendix~\ref{app:flux}, is: the operator $\sigma_j^\nu$
contains $b_j^\nu$, which toggles the bond variable on the $\nu$ bond at $j$
($u\to-u$) and therefore creates a flux pair on the two plaquettes sharing
that bond; the matrix element
$\langle\mathrm{GS}|\sigma_i^\mu(t)\sigma_j^\nu(0)|\mathrm{GS}\rangle$
vanishes unless $\sigma_i^\mu$ annihilates the same flux pair, which requires
the same component on the same bond (or $i=j$, where the matter-Majorana
insertion toggles twice and no flux is created). Since the flux label is
conserved under $H_K$, this bookkeeping is exact at all times and
frequencies. Equation \eqref{eq:selection} is the exact dynamical spin
structure factor statement that the pure model's correlations live on one
site and one bond only~\cite{Baskaran2007,Knolle2014}.

\subsection{The pairing kernel at second order in $J_K$}
\label{sec:kernel}

Integrating out the local moments at second order in the Kondo coupling
produces the retarded itinerant interaction
$V_{\mu\nu}(ij;\omega)=J_K^2\,\chi^{\mu\nu}_{ij}(\omega)$
[Eq.~\eqref{eq:genericglue};
Refs.~\cite{BerkSchrieffer1966,Monthoux1991,Coleman2015}]. By
\eqref{eq:selection} only on-site and nearest-neighbor terms survive, and
the spin structure on an $\alpha$ bond is the rank-one projector
$\delta_{\mu\alpha}\delta_{\nu\alpha}$. Fourier transforming the surviving
bond term---the phase factor $e^{i(\vk-\vk')\cdot\mathbf R}$ summed over the
positions $\mathbf R$ of the $\alpha$ bonds collapses, in the single-band
reduction, to the cosine of the momentum transfer along the bond direction---
gives in momentum space (sublattice indices suppressed; the single-band
status is that of Appendix~\ref{app:bdg}):
\begin{align}
V_{\mu\nu}(\mathbf k,\mathbf k';i\omega_n)
=J_K^2\sum_{\alpha}\chi_\alpha(i\omega_n)\,
\delta_{\mu\alpha}\delta_{\nu\alpha}
\nonumber\\ \times
\cos\big[(\mathbf k-\mathbf k')\cdot\boldsymbol\delta_\alpha\big]
+V_{\mathrm{on}},
\label{eq:kernelmicro}
\end{align}
where $\chi_\alpha$ is the susceptibility on an $\alpha$ bond,
$\boldsymbol\delta_\alpha=(\cos\phi_\alpha,\sin\phi_\alpha)$ with
$\phi_\alpha=2\pi\alpha/3$ the nearest-neighbor vectors, and
$V_{\mathrm{on}}$ is the momentum-independent on-site remainder (a
spin-channel contact term). It is structurally irrelevant to the triplet
sector: a momentum-independent kernel projects only onto the $m=0$
harmonic on the Fermi surface, and the angular integrals
$\int_0^{2\pi}d\phi\,e^{im\phi}$ with odd $m$ vanish identically; at its
only landing spot, the $s$-wave singlet channel, it enters with the
reversed Fierz sign [Eq.~\eqref{eq:fierz}], repulsive for precisely the
ferromagnetic sign of the correlations that makes the bond terms
attractive in the triplet. Hence $V_{\mathrm{on}}$ does not enter the pairing
channels below. This is the kernel class on which
Model~1 (Sec.~\ref{sec:model1}) is built: component-diagonal,
nearest-neighbor separable, rank one in spin space per bond---here derived from the
flux-sector selection rule rather than assumed.

\begin{remark}[Retardation window]
The bond susceptibility is gapped by flux physics: a single spin flip
creates a flux pair (Sec.~\ref{sec:selection}), so the absorptive part
$\chi_\alpha''(\omega)=0$ below the two-flux threshold $2\Delta_F$
(the nearest-neighbor flux-pair gap is isotropically
$2\Delta_F\approx0.26|J|$ \cite{Knolle2014,Pedrocchi2011}). Since the resulting $T_c$ depends exponentially
on $1/(J_K^2\chi_\alpha)$ [Eq.~\eqref{eq:bcs}], the pairing frequencies
satisfy $\omega_{\mathrm{pair}}\sim T_c\ll2\Delta_F$ self-consistently, and
the interaction is effectively static: what enters the gap equation is
$\chi_\alpha\equiv\chi_\alpha(i\omega_n\to0)$, the static bond
susceptibility. The dynamical shape of $\chi_\alpha''(\omega)$ modifies
only the prefactor, an order-unity retardation factor $r_\alpha$ computed
in Sec.~\ref{sec:wfunc}C [Eq.~\eqref{eq:rnum}]. The same hierarchy $T_c\ll2\Delta_F$ also
controls the thermal correction to the equal-time correlator $m_\alpha$
itself: at $T\sim T_c$ the deviation from its zero-temperature value
comes from thermally excited flux pairs and matter fermions---of order
$e^{-2\Delta_F/T_c}$ for the former (the flux gap stays finite throughout
the phase diagram, Sec.~\ref{sec:selection}) and of order $(T_c/|J|)^2$
for the latter in the gapless phase (exponentially suppressed in the
gapped phase)---both controlled small parameters relative to the static
value.
\end{remark}

\subsection{Channel competition: the parity-sector eigenvalue ratio}
\label{sec:parityeig}

The kernel \eqref{eq:kernelmicro} is diagonal in bond space but not in
parity: each bond term decomposes exactly as
\begin{align}
\cos\big[(\mathbf k-\mathbf k')\cdot\boldsymbol\delta_\alpha\big]
&=\cos(\mathbf k\!\cdot\!\boldsymbol\delta_\alpha)
 \cos(\mathbf k'\!\cdot\!\boldsymbol\delta_\alpha)\nonumber\\
&\quad+\sin(\mathbf k\!\cdot\!\boldsymbol\delta_\alpha)
 \sin(\mathbf k'\!\cdot\!\boldsymbol\delta_\alpha),
\label{eq:rank2micro}
\end{align}
which is simply the cosine addition formula; under $\vk\to-\vk$ the first
term is even (singlet channel) and the second odd (triplet channel) by the
antisymmetry constraint of Sec.~\ref{sec:bgmatrix}. The same glue therefore
mediates the even and odd chiral channels with locked strength
$v_\alpha=J_K^2\chi_\alpha$---the first instance of the locked two-sector
mediation generalized in Proposition~\ref{prop:locked}. On a circular
Fermi surface of radius $k_F$ the chiral channel eigenvalues follow from
the rank-one formula \eqref{eq:rankone} as
$\lambda_{s}=vN(0)\,W_g(k_F)$ and $\lambda_{t}=vN(0)\,W_t(k_F)$, with
the exact Fermi-surface weights
\begin{align}
W_g(k)&=\tfrac32\big[1-J_0(\sqrt3\,k)+J_0(2k)-J_0(k)\big]
\label{eq:Wg}\\
&=\tfrac{9}{64}\,k^4+O(k^6),\nonumber\\[1mm]
W_t(k)&=\tfrac32\big[1-J_0(\sqrt3\,k)-J_0(2k)+J_0(k)\big]
\label{eq:Wp}\\
&=\tfrac{9}{4}\,k^2+O(k^4),\nonumber
\end{align}
which are nothing but the Fermi-surface averages
$W_g=\langle|g_\pm|^2\rangle_{\mathrm{FS}}$ and
$W_t=\langle|p_\pm|^2\rangle_{\mathrm{FS}}$ of the even and odd parts of the
three-bond form factor (the triplet weight $W_t$ is unrelated to the
plaquette flux $W_p$ of Eq.~\eqref{eq:wp}), evaluated in closed form in
Appendix~\ref{app:weights} [Eqs.~\eqref{eq:f2}--\eqref{eq:p2}] via the
circular average identity \eqref{eq:fsaverage}; the small-$k$ expansions
follow from Eq.~\eqref{eq:besselsmall} (the $k^2$ cancellation in $W_g$ is
exact, order by order: the three bracketed $J_0$'s carry $k^2$ coefficients
$3/4$, $-1$, $+1/4$, which sum to zero). Therefore
\begin{equation}
\frac{\lambda_t}{\lambda_s}=\frac{W_t(k_F)}{W_g(k_F)}
=\frac{16}{k_F^2}\,\big[1+O(k_F^2)\big],
\label{eq:eigratio}
\end{equation}
numerically $\approx1.6\times10^3$ at $k_F=0.1$, $\approx63$ at
$k_F=0.5$, and $\approx24$ at $k_F=0.8$. Two structural facts further
confine the bond-locked spin glue to the odd sector. First, the bond-axis
spin structure pairs equal spin components along the locking axis, so the
even-parity half of the kernel is Pauli-inaccessible to the bond-axis
triplet (in the singlet channel the same kernel instead enters with the
reversed, repulsive sign for $w_\alpha>0$ [Eq.~\eqref{eq:fierz}], not by
exclusion). Second, the leading odd-parity $m=0$ combination
$\sum_\alpha\sin(\mathbf k\cdot\boldsymbol\delta_\alpha)$ is an $f$-wave
object of weight $O(k_F^6)$ [its $k^1$ term vanishes by
$\sum_\alpha\vd_\alpha=0$, and the cubic term reconstructs the $f$-wave
harmonic], parametrically below $W_t=O(k_F^2)$. The chiral triplet $p\pm ip$
is thus the leading chiral instability of the bond-locked spin glue, by a
margin that BCS amplification [$T_c\propto e^{-1/\lambda}$,
Eq.~\eqref{eq:bcs}] renders exponential in the coupling.

A final step translates the channel eigenvalue ratio into the condensate
amplitude ratio---the input of the dominance theorem
(Theorem~\ref{thm:dominance}, stated and proved in
Sec.~\ref{sec:counting}). Anticipating that result, the relevant
comparison threshold is
$|a_t/a_s|>J_2(k_F)/J_1(k_F)\simeq k_F/4$ at small $k_F$
[Eq.~\eqref{eq:dominance}]; we quote it here to make the following
comparison self-contained. The linearized gap equations of the two
parity channels are independent, so each has its own transition
temperature given by Eq.~\eqref{eq:bcs}, and
\begin{equation}
\ln\frac{T_c^{(t)}}{T_c^{(s)}}
=\frac{1}{\lambda_s}-\frac{1}{\lambda_t}
\simeq\frac{1}{\lambda_s}\Big(1-\frac{k_F^2}{16}\Big),
\label{eq:tcsep}
\end{equation}
i.e.\ the two $T_c$'s are exponentially separated: in the wide window
$T_c^{(s)}<T<T_c^{(t)}$ only the triplet component condenses ($a_s=0$)
and no mixed state appears at all. The singlet component enters only
below $T_c^{(s)}$, where the channel-resolved weak-coupling expansion
$D^2\propto\lambda(1-T/T_c)$ (Proposition~\ref{prop:bondlock};
Ref.~\cite{Coleman2015}; the kernel is diagonal in parity, so the
quadratic terms of the two channels decouple exactly and cross-channel
coupling appears only at quartic order) gives
\begin{equation}
\Big|\frac{a_t}{a_s}\Big|^2
=\frac{\lambda_t}{\lambda_s}\,
\frac{1-T/T_c^{(t)}}{1-T/T_c^{(s)}}
\ge\frac{\lambda_t}{\lambda_s}\simeq\frac{16}{k_F^2},
\label{eq:ampratio}
\end{equation}
the second factor being no smaller than $1$ for $T\le T_c^{(s)}$ (the
zero-temperature BCS value $|a_t/a_s|\sim e^{1/\lambda_s-1/\lambda_t}$
is larger still). Comparing with the dominance threshold
\eqref{eq:dominance}: the triplet side requires $|a_t/a_s|>k_F/4$, and
the lower bound $4/k_F$ on the amplitude ratio clears that threshold by
exactly a factor $16/k_F^2$---the eigenvalue ratio itself. In the leading
two-harmonic approximation, $|a_t/a_s|\gtrsim\sqrt{\lambda_t/\lambda_s}\simeq
J_1(k_F)/J_2(k_F)$, which exceeds the dominance threshold $J_2(k_F)/J_1(k_F)$
whenever $J_1>J_2$, i.e.\ for all $k_F$ below the first $J_1$--$J_2$
crossing at $2.63>|K|=\frac{4\pi}{3\sqrt3}\simeq2.42$ (at
leading-harmonic resolution of the dominance boundary); the full-ladder
check in
Appendix~\ref{app:topology} supports the same ordering throughout the zone.
Amplified by
the exponential $T_c$ dependence, the eigenvalue ratio thus propagates
quantitatively into the amplitude ratio, and a Kitaev-spin-glue-driven
chiral state lands on the triplet side of the topological boundary; the
precise material-by-material ratio follows from the coupled gap
equations (Sec.~\ref{sec:conclusion}).

Three consequences follow.

(i) \emph{Chern-number verdict.} By the dominance theorem
(Sec.~\ref{sec:counting}, Theorem~\ref{thm:dominance}), a chiral state
driven by the present kernel lies on the triplet side of the topological
phase boundary: $C=\pm2$, two Majorana edge branches, and
$\kappa_{xy}/T=\kappa_0$ [Eq.~\eqref{eq:thallunit}]. The singlet $d+id$
phase with $C=\pm4$ requires the charge-vertex channel of Model~2
(Sec.~\ref{sec:model2}), whose kernel is parity-projected: its triplet
projection vanishes exactly by parity and its $s$-wave projection vanishes
exactly by angular-momentum matching (Proposition~\ref{prop:chiral}), so no
competition arises there.

(ii) \emph{Coexistence criterion.} With both vertices active, the
singlet channel wins the transition only if
\begin{equation}
\frac{v_{\rm charge}}{v_{\rm spin}}>\frac{W_t(k_F)}{W_g(k_F)}
\simeq\frac{16}{k_F^2},
\label{eq:vertexcond}
\end{equation}
a kinematic handicap of two to three orders of magnitude at small
$k_F$. The measured Chern number (equivalently the thermal-Hall
plateau, $2\kappa_0$ vs $\kappa_0$) thereby discriminates which
microscopic vertex dominates the pairing in a given material.

(iii) \emph{Anisotropy caveat.} Equation \eqref{eq:eigratio} is the
isotropic value. For $\chi_z\neq\chi_{x,y}$ the broken $C_3$ pins real
combinations---the nodal bond-locked state of
Sec.~\ref{sec:formfactors}---and re-weights the channels without closing
the hierarchy as long as the chiral eigenchannels remain meaningful; the
Pauli exclusion of the even channel for the bond-axis spin glue is
structural and unaffected.

\section{The weight function $w_\alpha(J_x,J_y,J_z)$}
\label{sec:wfunc}

\subsection{Exact expression via Hellmann--Feynman}

In the zero-flux sector the matter Majoranas hop on a bipartite lattice with
spectrum $\pm|\,g(\mathbf k)|$, with $g(\mathbf k)$ of Eq.~\eqref{eq:g}. The
ground-state energy per unit cell is the sum over the filled (negative-energy)
band,
\begin{equation}
e_{\mathrm{cell}}=-\langle |g|\rangle_{\mathrm{BZ}},
\label{eq:ecell}
\end{equation}
with $\langle\cdot\rangle_{\mathrm{BZ}}$ the normalized Brillouin-zone
average. The equal-time bond correlator follows from the Hellmann--Feynman
theorem, $\partial E_0/\partial\lambda=\langle\mathrm{GS}|\partial
H/\partial\lambda|\mathrm{GS}\rangle$~\cite{Coleman2015}. Since $\partial H_K/\partial J_\alpha=-\sum_{\langle
ij\rangle_\alpha}\sigma_i^\alpha\sigma_j^\alpha$, the bond correlator is
\begin{equation}
m_\alpha\equiv\langle\sigma_i^\alpha\sigma_j^\alpha\rangle_{\alpha\text{-bond}}
=-\frac{\partial e_{\mathrm{cell}}}{\partial J_\alpha}.
\label{eq:hf}
\end{equation}
where $m_\alpha$ denotes the equal-time bond correlator; the letter is
chosen for the bond-energy scale $J_\alpha m_\alpha$, not for a mass.

Differentiating Eq.~\eqref{eq:ecell} with $\partial g/\partial J_z=1$,
$\partial g/\partial J_x=e^{ik_1}$, $\partial g/\partial J_y=e^{ik_2}$
[Eq.~\eqref{eq:g}] and
$\partial|g|/\partial J_\alpha=\mathrm{Re}(g^\ast\,\partial g/\partial
J_\alpha)/|g|$ (from $|g|=\sqrt{g^\ast g}$), one obtains
\begin{align}
m_z&=\Big\langle\frac{\mathrm{Re}\,g}{|g|}\Big\rangle_{\!\mathrm{BZ}},
\label{eq:mz}\\
m_x&=\Big\langle\frac{\mathrm{Re}\!\big(g\,e^{-ik_1}\big)}{|g|}
\Big\rangle_{\!\mathrm{BZ}},\qquad m_y\ \text{similarly}
\label{eq:mx}
\end{align}
(using $\mathrm{Re}(g^\ast e^{ik_1})=\mathrm{Re}(g e^{-ik_1})$). The
normalization is fixed so that $m_\alpha\to1$ in the isolated-dimer
limit---as $J_z\gg J_x,J_y$, $|g|\to J_z$ and $m_z\to1$ directly from
Eq.~\eqref{eq:mz}---and $e_{\mathrm{cell}}=-\sum_\alpha J_\alpha m_\alpha$
identically [integrate Eq.~\eqref{eq:hf} or differentiate
Eq.~\eqref{eq:ecell} along a radial path in coupling space].

\subsection{Properties of the bond correlator}
\label{sec:wprops}

The following hold exactly or to the stated numerical accuracy; the
derivations are in Appendix~\ref{app:bz}, and we indicate the mechanism of
each at the statement.

\begin{result}[Vanishing at decoupling]
\label{res:zero}
A switched-off component carries no weight, $m_\alpha\to0$ as
$J_\alpha\to0$: setting $J_z=0$ in \eqref{eq:mz}, the integrand factorizes,
$\mathrm{Re}\,g/|g|=\cos u\,\mathrm{sgn}(\cos v)$ with
$u=(k_1+k_2)/2$, $v=(k_1-k_2)/2$, whose BZ average vanishes because the
averages of $\cos u$ and of $\mathrm{sgn}(\cos v)$ both vanish (the
change of variables $(k_1,k_2)\to(u,v)$ has constant Jacobian and
factorizes the average; Appendix~\ref{app:bz}). The ansatz $w_\alpha\propto J_\alpha$ is therefore
exact at the endpoints---there is no intercept: a vanishing exchange takes
its pairing channel with it.
\end{result}

\begin{result}[Transverse weight at the boundary]
\label{res:2pi}
At $J_z=0$ with $J_x=J_y$, $m_x=m_y=\langle|\cos v|\rangle_{\mathrm{BZ}}=2/\pi\approx
0.6366$ [the average of $|\cos|$ over a period]. A component that is switched
off leaves its two neighbors with \emph{enhanced} weights, $2/\pi>0.5249$:
the spectral weight released by the decoupled bond is absorbed by the
remaining ones.
\end{result}

\begin{result}[Line-node logarithm]
\label{res:log}
As $J_z\to0^+$ in the gapless phase,
\begin{equation}
m_z(J_z)=\frac{J_z}{2\pi}\,\ln\frac{c}{J_z}+O(J_z),
\qquad c\approx e^{0.522\times2\pi}\approx 27,
\label{eq:log}
\end{equation}
where the coefficient $1/(2\pi)$ is fixed analytically by the divergence
of $\partial m_z/\partial J_z=\langle(\mathrm{Im}\,g)^2/|g|^3\rangle$ at
the line node $\cos v=0$ of the $J_z=0$ model (the locus where
$J_xe^{ik_1}+J_ye^{ik_2}=0$; details in Appendix~\ref{app:bz});
numerically $m_z\simeq J_z\,[0.522+0.159\,\ln(1/J_z)]$ for
$0.02\lesssim J_z\lesssim0.5$, with $0.159=1/(2\pi)$ to three digits. The
small component's weight is thus \emph{logarithmically enhanced} over the
linear ansatz at strong anisotropy.
\end{result}

\begin{result}[Isotropic value and monotonicity]
$m_\alpha=0.5249$ at $J_x=J_y=J_z$ \cite{Baskaran2007}; $m_z$ increases
monotonically and concavely from $0$ through the gapless (B) phase to the
dimer limit $m_z\to1$ in the A phase ($J_z>J_x+J_y$, and permutations), where
$m_{x,y}\to0$.
At isotropy the three components coincide, with common value
$\bar m\equiv m_\alpha|_{\mathrm{iso}}$ and common kernel weight
$\bar w\equiv w_\alpha|_{\mathrm{iso}}$. The anisotropy transfer
coefficient at isotropy is
\begin{equation}
R\equiv\frac{\bar J}{\bar m}\frac{\partial m_z}{\partial J_z}
\bigg|_{\mathrm{iso}}\approx0.70,
\label{eq:kappa}
\end{equation}
i.e.\ at fixed retardation factor a $1\%$ exchange anisotropy induces a
$0.7\%$ static-weight anisotropy: $\delta m_\alpha/\bar m=R\,\delta
J_\alpha/\bar J$. Anisotropy of $r_\alpha$ is an additional correction
(Sec.~\ref{sec:rfactor}).
\end{result}

\subsection{From $m_\alpha$ to $w_\alpha$: the retardation factor}
\label{sec:rfactor}

\paragraph{Spectral representation.}
At $T\ll\Delta_F$, the static susceptibility entering the kernel
\eqref{eq:kernelmicro} is the zero-temperature imaginary-time integral of
the correlator, not the equal-time value:
\begin{equation}
\chi_\alpha=\int_0^\beta d\tau\,
\langle\sigma_i^\alpha(\tau)\sigma_j^\alpha(0)\rangle
=r_\alpha\,\frac{m_\alpha}{\Delta_F},
\label{eq:r}
\end{equation}
where the second form extracts the equal-time part $m_\alpha$ and collects
the dynamics into a dimensionless factor $r_\alpha$ of order the flux-pair
lifetime over $\Delta_F^{-1}$ [the correlator decays on the scale
$\Delta_F^{-1}$ because the intermediate state contains a flux pair of
energy $2\Delta_F$; Sec.~\ref{sec:selection}]. In the gapless phase the
matter continuum makes $r_\alpha$ weakly
coupling- and anisotropy-dependent; we expect $r_\alpha$ smooth, with its magnitude
constrained by the rigorous bound \eqref{eq:rbound}, and
treat it as a prefactor below.

\paragraph{First-moment identity.}
This expectation can be upgraded to a
controlled statement. The spectral representations
$\chi_\alpha=\int_0^\infty\frac{d\omega}{\pi}\,A_\alpha(\omega)/\omega$
and $m_\alpha=\int_0^\infty\frac{d\omega}{\pi}\,A_\alpha(\omega)$ (with
$A_\alpha=-\mathrm{Im}\,\chi_\alpha^{\rm ret}$ the bond-resolved spectral
function) give the closed-form definition
\begin{equation}
r_\alpha=\Delta_F\,
\frac{\int_0^\infty d\omega\,A_\alpha(\omega)/\omega}
{\int_0^\infty d\omega\,A_\alpha(\omega)},
\label{eq:rspec}
\end{equation}
i.e.\ the ratio of the $1/\omega$-weighted moment to the zeroth moment
of the spectrum, in units of $\Delta_F$. Since a single spin flip
creates a flux pair, the intermediate states live in the two-flux sector
and the spectral threshold $\omega_{\min}=2\Delta_F$ is strictly
positive throughout the phase diagram (Sec.~\ref{sec:selection}). The
spectral weight obeys two exact sum rules: the zeroth moment
$\int(d\omega/\pi)\,A_\alpha(\omega)=m_\alpha$ is the denominator of
\eqref{eq:rspec}, and the first moment vanishes identically,
\begin{equation}
\int_0^\infty\frac{d\omega}{\pi}\,\omega\,A_\alpha(\omega)=0.
\label{eq:fsum}
\end{equation}
The moment equals
$\langle\psi_0|c_i(H'_K-E_0)c_j|\psi_0\rangle=\langle[c_i,H_K]c_j\rangle+\langle c_iVc_j\rangle$
(using $\langle\psi_0|H_K=E_0\langle\psi_0|$), where $V=H'_K-H_K$
contains only the flipped bond and, since $u_{ij}$ changes sign, equals
$-2$ times that bond's term. The Wick contraction of the first term is
$i(A_0C)_{ij}$, with $A_0$ the zero-flux hopping matrix and
$C=1+\mathrm{sign}(iA_0)$ the covariance matrix: since
$A_0\,\mathrm{sign}(iA_0)=-i|iA_0|$, and chiral (sublattice) symmetry
makes $|iA_0|$ block diagonal with vanishing cross-sublattice element
$(ij)$, one obtains $i(A_0C)_{ij}=i(A_0)_{ij}$; the second term
contracts directly to $-i(A_0)_{ij}$. The two cancel identically,
irrespective of any detail of the spectrum.

\paragraph{Why the naive bound fails.}
Equation~\eqref{eq:fsum} has an inescapable consequence: since
$m_\alpha\neq0$ while the first moment vanishes, the spectral weights
cannot be everywhere nonnegative---the continuum necessarily carries
negative-weight components. Pointwise estimates
$1/\omega_n\le1/\omega_{\min}$ therefore \emph{cannot} be converted into
an upper bound on $r_\alpha$: the gap alone does not constrain this
ratio. The only unconditional rigorous bound follows from
Cauchy--Schwarz,
$\sum_n|W_n|\le\sqrt{\langle\psi_0|c_i^2|\psi_0\rangle\langle\psi_0|c_j^2|\psi_0\rangle}=1$,
so $|\chi_\alpha|\le\omega_{\min}^{-1}\sum_n|W_n|\le1/\omega_{\min}$,
i.e.\
\begin{equation}
|r_\alpha|\;\le\;\frac{\Delta_F}{\omega_{\min}}\,\frac{1}{m_\alpha}\;=\;\frac{1}{2m_\alpha},
\label{eq:rbound}
\end{equation}
which gives $|r_\alpha|\lesssim0.95$ at the isotropic point
$m_\alpha\simeq0.525$. Loose as it is, this bound assumes nothing about
the spectral shape; sharper quantitative statements must come from the
spectral function itself. The actual magnitude of $r_\alpha$ is set by
the degree of cancellation between the positive- and negative-weight
parts of the continuum under the $1/\omega$ weighting, not by the gap.

\paragraph{Short-time expansion.}
Equivalently, in the time domain, the same identity states that the
short-time expansion of the imaginary-time correlator skips the
linear order,
\begin{equation}
G_\alpha(\tau)\equiv i\,\langle\psi_0|c_i\,e^{-(H'_K-E_0)\tau}c_j|\psi_0\rangle
=m_\alpha+\frac{\mu_2}{2}\,\tau^2+O(\tau^3)
\label{eq:shorttime}
\end{equation}
(the factor $i$ normalizes $G_\alpha(0)=m_\alpha$), i.e.\ the crossed
bond correlator starts with zero initial slope---the flux pair's decay
is inertial. The curvature is again closed form: writing
$\mu_2=\langle a|b\rangle$ with $|a\rangle=(H'_K-E_0)c_i|\psi_0\rangle$
and $|b\rangle=(H'_K-E_0)c_j|\psi_0\rangle$ and substituting
$H'_K=H_K+V$, the flipped bond's own term cancels exactly between
$[H_K,c]$ and $Vc$, leaving
\begin{equation}
\mu_2=i\!\!\sum_{\substack{k\sim i,\;k\neq j}}\;\sum_{\substack{l\sim j,\;l\neq i}}(A_0)_{ik}(A_0)_{jl}\,C_{kl},
\label{eq:mu2}
\end{equation}
where $k\sim i$ denotes the nearest neighbors of $i$: the sums cover
only the \emph{other} two bonds at each endpoint---the opened bond
itself does not contribute, and the short-time curvature is entirely
fixed by the four cross-sublattice covariances $C_{kl}$ surrounding it.

\paragraph{Quench protocol and hard checks.}
Its quantitative computation follows a standard route: acting with
$\sigma_j^\alpha$ on the ground state is equivalent to flipping $u_{ij}$
on that $\alpha$ bond (creating an adjacent flux pair) while adding one
matter Majorana, so $\chi_\alpha$ is the free-Majorana correlator of a
sudden switch (quench) in the two-flux sector
\cite{Kitaev2006,Baskaran2007,Knolle2014}---diagonalize the matter
quadratic form in the zero-flux and two-flux sectors separately and
assemble the spectral function from Gaussian-state overlaps. The
assembled result faces four hard checks: the zeroth moment must
reproduce the exact values of \eqref{eq:mz}--\eqref{eq:mx} point by
point, the spectral threshold must sit at $2\Delta_F$, the first
moment must vanish (Eq.~\eqref{eq:fsum}), and the second moment must
reproduce the closed form \eqref{eq:mu2}. The four differ in logical
status: the first two test the spectral assembly itself (normalization
and sector), and their failure signals an implementation error; the
last two are structural---Eqs.~\eqref{eq:fsum} and \eqref{eq:mu2} are
fixed uniquely by the operator content, so their failure indicates that
the operator content itself is wrong (for instance, quenching with
$\sigma_i^\alpha\sigma_j^\alpha$ instead of a single
$\sigma_j^\alpha$).
We have carried out this computation along the route just described. On
the periodic $2L^2$-site clusters with $L=4,5,7$ ($32$, $50$, and $98$
sites) we diagonalize the matter-Majorana quadratic form in the zero-flux
and two-flux sectors separately and assemble $\chi_\alpha$ from
Gaussian-state (Pfaffian) overlaps, including matter-fermion sectors with
up to five fermions (seven for $L=4$); at $L=4$ the assembly agrees with a
direct exact-diagonalization Lehmann summation to machine precision, and at
every computed point all four hard checks above pass [the zeroth moment
reproduces \eqref{eq:mz}--\eqref{eq:mx} to better than $10^{-8}$, the
threshold sits
at $2\Delta_F+\varepsilon_{\min}$, the first moment stays below $10^{-7}$,
and the second moment matches \eqref{eq:mu2}]. At the isotropic point we
find
\begin{equation}
\begin{split}
r_z=r_x=0.25043,\;0.24201,\;0.24361,\;0.24210\\
\qquad(L=4,5,7,8),
\label{eq:rnum}
\end{split}
\end{equation}
with a size-converged estimate
$r_\alpha^{(\infty)}=0.243\pm0.002$ (Sec.~\ref{app:ralpha})---about a
quarter of the kinematic bound
$|r_\alpha|\leq1/(2m_\alpha)\approx0.945$ [Eq.~\eqref{eq:rbound}]:
retardation softens the bond glue uniformly without reshaping it. The gap
between the computed $r_\alpha\simeq0.24$ and the bound $\approx0.95$
reflects the cancellation of positive- and negative-weight spectral
components under the $1/\omega$ weighting: the bound saturates only for a
single positive pole at the threshold, whereas the actual continuum
distributes weight over a wide energy range with sign changes enforced by
the first-moment identity \eqref{eq:fsum}. Three
structural facts stand out. (i) \emph{Single-pole dominance}: the
one-matter-fermion sector carries $99.96\%$ of $\chi_\alpha$ (the
three-fermion sector $0.04\%$, higher sectors below $10^{-7}$), and the
single lowest pole at $\theta_{\min}=2\Delta_F+\varepsilon_{\min}$ alone
exhausts $\chi_\alpha$ within a few percent, the remainder being
cancellations among subleading poles [explicitly, at isotropy
$\varepsilon_{\min}=0.458$, $0.332$, $0.341$ ($\simeq2.4$--$2.7\,\Delta_F$)
for $L=4,5,7$, and the lowest pole alone gives $r_z\approx0.26$ at $L=4$,
against the full $0.2504$]. (ii) \emph{Anisotropy amplification}: $r_z$ rises
from $0.25$ at isotropy to $0.375$ at $J_z=1.5$ ($0.3749$ at $L=4$,
$0.3778$ at $L=7$) and approaches the dimer-limit value $1/2$ ($0.4987$,
$0.4998$, $0.4996$ at $J_z=4$ for $L=4,5,7$)---the exact isolated-dimer
endpoint, where the bond operator is nearly diagonal on the dimer ground
state and the two-flux vacuum pole exhausts the spectrum---while $r_x$ falls to
$0.19$--$0.17$ at $J_z=1.5$ [$L=4$ and $L=7$; on the $L=5$ cluster the flux gap collapses at this coupling ($2\Delta_F\simeq0.009$), so we quote the two clean sizes];
the susceptibility ratio
$\chi_z/\chi_x\approx2.2$--$2.5$ at $J_z=1.5$ therefore substantially
exceeds the static ratio $m_z/m_x\approx1.5$--$1.6$, so the static
transfer coefficient $R\approx0.70$ [Eq.~\eqref{eq:kappa}]
\emph{underestimates} the operative anisotropy transfer and the screening
estimates of Sec.~\ref{sec:screening} are conservative. (iii)
\emph{Reliability window}: in the interval $J_z\approx1.75$--$2$ the flux
gap $2\Delta_F$ collapses on finite clusters with size-dependent sign
fluctuations, and the assembled numbers there are finite-size artifacts
and excluded. Deep in the A phase ($J_z\gtrsim2.25$) the flux gap is
intrinsically small and changes sign on the smallest cluster ($L=4$), so
the $J_z=4$ values above are used only for the dimer trend, on which all
three sizes agree to $0.2\%$. With $r_\alpha$ thus determined, the kernel
weight is
\begin{equation}
\boxed{\;w_\alpha=J_K^2\,\chi_\alpha
\propto r_\alpha\,m_\alpha(J_x,J_y,J_z)\;}
\label{eq:weight}
\end{equation}
with $m_\alpha$ given exactly by \eqref{eq:mz}--\eqref{eq:mx}.

\paragraph{Analytic cross-check from exact moments.} The low moments of the
spectral function are themselves computable in closed form:
$\mu_n=i\langle\psi_0|c_i(H'_K-E_0)^nc_j|\psi_0\rangle$ is a finite sum of
$(2n{+}2)$-point Majorana correlators in the zero-flux Gaussian state,
evaluated by Pfaffian recursion. Beyond $\mu_0=m_\alpha$ (exact), $\mu_1=0$
[Eq.~\eqref{eq:fsum}], and $\mu_2$ [Eq.~\eqref{eq:mu2}]---which evaluates to
$-2$ at $L=4$ and $-1.894209$ at $L=5$, both matching the closed form
\eqref{eq:mu2} digit by digit---the third moment is obtained from the same
recursion as a sum of $8$-point Pfaffians, giving $\mu_3=-11.090170$ at
isotropy ($L=4$), verified against direct Fock-space diagonalization to
machine precision. A two-pole Pad\'e approximant of
$S(z)=\int d\omega\,A(\omega)/(z-\omega)$ built on $\mu_{0..3}$ then gives
$r_\alpha^{\mathrm{Pad\acute{e}}}=-\Delta_F m_\alpha\,\mu_3/\mu_2^2=0.2497$,
within $0.3\%$ of the full spectral value $0.2504$ [Eq.~\eqref{eq:rnum}].
This is an analytic cross-check, not an exact derivation: the accuracy
degrades when the spectral weight spreads over many poles ($-9\%$ at $L=5$),
and the two-pole fit degenerates as $\mu_2\to0$ toward the dimer limit, where
$r_\alpha\to1/2$ instead follows from single-pole exhaustion.

\begin{remark}[Status of the ansatz]
Compared with the phenomenological weighting $w_\alpha\propto J_\alpha$,
the exact weight (i) agrees at both endpoints (zero intercept at
$J_\alpha=0$, Result~\ref{res:zero}; dimer saturation), (ii) is
logarithmically enhanced for a weak component (Result~\ref{res:log}), and
(iii) dilutes anisotropy by $R\approx0.70$ near isotropy
[Eq.~\eqref{eq:kappa}], a dilution that the computed retardation factor in
fact partly reverses [$\chi_z/\chi_x\approx2.2$--$2.5$ versus
$m_z/m_x\approx1.5$--$1.6$ at $J_z=1.5$, Eq.~\eqref{eq:rnum} and the
surrounding discussion]. None of the structural conclusions of the
bond-locked model change; the materials criterion does
(Sec.~\ref{sec:screening}).
\end{remark}

\section{Model 1: anisotropic spin-triplet pairing with a bond-locked $\mathbf d$ vector}
\label{sec:model1}

\subsection{Construction and effective pairing interaction}
\label{sec:construct1}

For the doped Kitaev honeycomb lattice of Sec.~\ref{sec:micro}
[$H=H_K+H_t+H_J$, Eq.~\eqref{eq:fullH}], the electron-spin-fluctuation vertex
is spin-channel (vector glue), $\psi^\dagger\boldsymbol\sigma\psi$ with a
Hermitian vertex. Second-order perturbation theory in $J_K$ integrates out
the high-energy local-moment fluctuations and yields the effective pairing
interaction [Eq.~\eqref{eq:genericglue};
Refs.~\cite{Coleman2015,Monthoux1991}]
\begin{equation}
V_{\alpha\beta}^{(1)}(\mathbf k,\mathbf k')
=-\sum_\gamma w_\gamma\,\delta_{\alpha\gamma}\delta_{\beta\gamma}\,
\sin(\mathbf k\!\cdot\!\boldsymbol\delta_\gamma)
\sin(\mathbf k'\!\cdot\!\boldsymbol\delta_\gamma) ,
\label{eq:hpair1}
\end{equation}
where $w_\gamma=J_K^2\chi_\gamma$ is the attractive triplet-channel
strength. The derivation of Sec.~\ref{sec:micro} makes the full kernel
precise: the flux-sector selection rule forces it to be anisotropic in spin
space and diagonal in the component index
[Eq.~\eqref{eq:kernelmicro}], summarized as
\begin{equation}
V_{\alpha\beta}(\mathbf k,\mathbf k')\propto
w_\alpha\,\delta_{\alpha\beta}\,
\chi(\mathbf k-\mathbf k') ,
\label{eq:kernel1}
\end{equation}
with the bond-resolved momentum dependence of Eq.~\eqref{eq:kernelmicro} and
the microscopic weights $w_\alpha=J_K^2\chi_\alpha$ of
Eq.~\eqref{eq:weight}: component-diagonal, nearest-neighbor separable, rank
one per bond.

\paragraph{Note 1 (glue chain).}
The anisotropic exchange is a symmetric spin tensor (quadrupole type,
Sec.~\ref{sec:kitaev}A), not a wedge product: there is no DM term in the
model, so the anisotropic $\mathbf d$ vector cannot originate from any
Dzyaloshinskii--Moriya structure. It originates in the anisotropic
fluctuation weights $w_\alpha$ within the vector channel. Pairing occurs in
the spin-triplet channel (vector glue, Fierz sign $+1$---the eigenvalue
computation of Eq.~\eqref{eq:fierz};
Refs.~\cite{BerkSchrieffer1966,SigristUeda1991}), and the anisotropic
$\mathbf d$ vector is its anisotropic variant within that channel.

\subsection{Bond locking and its exact consequences}
\label{sec:bondlock}

Because the kernel \eqref{eq:kernel1} is diagonal in the spin component, the
linearized gap equation \eqref{eq:lingap} is diagonal in $\alpha$: each
triplet component condenses with its own coupling
\begin{equation}
\lambda_\alpha=w_\alpha=J_K^2\chi_\alpha ,
\label{eq:lambda}
\end{equation}
and the odd-parity channel of the $\alpha$-bond kernel is the rank-one
separable form $w_\alpha\sin(\mathbf k\!\cdot\!\boldsymbol\delta_\alpha)
\sin(\mathbf k'\!\cdot\!\boldsymbol\delta_\alpha)$
[Eq.~\eqref{eq:rank2micro}], whose unique eigenvector is the bond sine
itself [Eq.~\eqref{eq:rankone}].

\begin{proposition}[Bond-locked anisotropic $\mathbf d$ vector]
\label{prop:bondlock}
From the anisotropic weights of the kernel \eqref{eq:kernel1}, the triplet
components are
\begin{equation}
d_\alpha(\mathbf k)=D_\alpha\,\sin(\mathbf k\!\cdot\!\boldsymbol\delta_\alpha),
\qquad D_\alpha^2\propto w_\alpha\ \text{just below condensation},
\label{eq:bondlock}
\end{equation}
with $\boldsymbol\delta_\alpha$ the $\alpha$-bond nearest-neighbor vector:
each component condenses with its own bond-locked sine form factor (triplet
pairing is forced odd in $\mathbf k$ by fermionic antisymmetry,
Sec.~\ref{sec:bgmatrix}, and the sine is the odd member of the pair
\eqref{eq:rank2micro}). The three components of $\mathbf d(\mathbf
k)=(\Delta_{yz},\Delta_{zx},\Delta_{xy})$ therefore scale with the exchange
anisotropy and lock to the bond directions: the $\alpha$ component carries
squared amplitude proportional to $w_\alpha$ to leading order [the linearized gap
equation at $T_c^{(\alpha)}$ gives $D_\alpha=\lambda_\alpha\langle
\sin^2\rangle_{\mathrm{FS}}\chi_0\,D_\alpha$, and just below the transition
$D_\alpha^2\propto\lambda_\alpha(1-T/T_c^{(\alpha)})$ by the standard
weak-coupling expansion \cite{Coleman2015}].
\end{proposition}

The triplet block of $V_{\alpha\beta}$ is its diagonal part, with Fierz sign
$+1$ [Eq.~\eqref{eq:fierz};
Refs.~\cite{BerkSchrieffer1966,Vollhardt1990}]; the weights $w_\alpha$ enter
the component kernels directly, hence the squared component amplitudes track
$w_\alpha\propto r_\alpha m_\alpha$ (Secs.~\ref{sec:micro}--\ref{sec:wfunc}).
The bond direction sets where each weight sits, and the locking follows.
The locking here is a correspondence between spin indices and bond
indices---the kernel on an $\alpha$ bond couples only to the $\alpha$
spin component of $\mathbf d$---not a pinning of the geometric direction
of $\mathbf d$ to a real-space bond axis: $\mathbf d$ always lives in
spin space, the anisotropy entering through the weights of the component
amplitudes $D_\alpha$ (Corollary~\ref{cor:weights}). Four
consequences are exact within this kernel class.

\begin{corollary}[Weight anisotropy, not texture anisotropy]
\label{cor:weights}
With the bond-sine form factors of Eq.~\eqref{eq:bondlock},
\begin{equation}
|\mathbf d(\mathbf k)|^2=\sum_\alpha |D_\alpha|^2
\sin^2(\mathbf k\!\cdot\!\boldsymbol\delta_\alpha) ,
\label{eq:weights}
\end{equation}
a positive-semidefinite Gram sum (a sum of squares with non-negative
coefficients). The anisotropy fixes the orientation of $\mathbf d$ and
modulates the gap magnitude; the zeros of $|\mathbf d|$ are the
\emph{common} zeros of the three bond sines, weight-independent by
construction: a sum of non-negative terms vanishes only where every term
vanishes.
\end{corollary}

\begin{corollary}[Component-resolved sequential condensation]
\label{cor:seq}
The linearized gap equation is diagonal in $\alpha$, so each component has
its own Thouless criterion [Eq.~\eqref{eq:rankone} and the line following
it],
$1=\lambda_\alpha\,\langle\sin^2(\mathbf k\!\cdot\!\boldsymbol\delta_\alpha)
\rangle_{\mathrm{FS}}\,\chi_0(T_c^{(\alpha)})$,
with $T_c^{(\alpha)}$ generically distinct. A single component condenses at
the highest $T_c$, with $\mathbf d$ pinned to the dominant bond direction;
the remaining components enter successively at lower temperature. The
anisotropic Kitaev glue therefore predicts multiple superconducting
specific-heat anomalies as a structural feature, not as fine tuning.
\end{corollary}

\begin{corollary}[Unitarity: no spontaneous time-reversal breaking]
\label{cor:unitary}
The weights $w_\alpha$ are real, so the kernel generates no relative phase
between the components at quadratic order; the mean-field solution can be
chosen with real $D_\alpha$, hence $\mathbf d\times\mathbf d^\ast=0$. By the
definitions of Sec.~\ref{sec:bgmatrix} the bond-locked state is a unitary
triplet: it preserves time reversal and shows no Kerr rotation and no
thermal Hall effect. Chiral triplet order would require an additional
ingredient (a complex kernel or a quartic-level degeneracy) and is not
claimed.
\end{corollary}

\paragraph{Note 2 (sign analysis of the channels).}
Whether the triplet channel is attractive depends on the product of two
signs: that of the Kitaev exchange and that of the Kondo coupling. With the
convention of Eq.~\eqref{eq:kitaev}, $J_\alpha>0$ is ferromagnetic exchange;
the fluctuations are then ferromagnetic and the Fierz sign [$+1$ in the
triplet sector, $-3$ in the singlet sector, Eq.~\eqref{eq:fierz}
\cite{BerkSchrieffer1966}] favors the triplet. For $J_\alpha<0$
(antiferromagnetic) the leading channel moves to the singlet sector (e.g.\
$d+id$), consistent with existing mean-field and numerical
results~\cite{You2012,Hyart2012,Okamoto2013,Laurell2024}. The surviving channel of Model~1 is the triplet for
ferromagnetic Kitaev exchange. The weight function
$w_\alpha=J_K^2\chi_\alpha\propto r_\alpha m_\alpha$ is derived in
Secs.~\ref{sec:micro}--\ref{sec:wfunc}; the linear weighting
$w_\alpha\propto J_\alpha$ is exact at the endpoints and corrected
quantitatively in between (Result~\ref{res:log} and the transfer coefficient
$R\approx0.70$).

\subsection{Refined form factors and the nodal structure}
\label{sec:formfactors}

The bond-sine locking of Proposition~\ref{prop:bondlock} is the same
bond-dependent triplet universality class found in mean-field studies of
doped Kitaev models \cite{You2012,Hyart2012,Okamoto2013}; here it is exact at
the stated order rather than a mean-field outcome. Since $|\mathbf
d|^2=\sum_\alpha|D_\alpha|^2 \sin^2(\mathbf k\!\cdot\!\boldsymbol\delta_\alpha)$
is a positive semidefinite Gram sum [Eq.~\eqref{eq:weights}], its zeros are
the \emph{common} zeros of the three bond sines---weight-independent by
construction. We determine these zeros geometrically. The zeros of
$\sin(\vk\cdot\vd_\alpha)$ are the parallel lines
$\vk\cdot\vd_\alpha=p\pi$ ($p\in\mathbb Z$). Within the first Brillouin zone
only $p=0$ occurs: the zone is contained in the disk
$|\vk|\le|K|=\frac{4\pi}{3\sqrt3}\approx2.42<\pi$ (unit bond length), so
$|\vk\cdot\vd_\alpha|\le|\vk|<\pi$. The three $p=0$ lines
$\vk\cdot\vd_\alpha=0$ ($\alpha=x,y,z$) are distinct directions through the
origin, and any two of them meet only at $\vk=0$: the only common zero is
$\Gamma$. Hence, for any Fermi surface enclosing $\Gamma$ (the
generic doped case):

\begin{corollary}[Strengthened node structure]
\label{cor:nodes}
(i) \emph{Single-component window} ($T_c^{(2)}<T<T_c^{(1)}$): the gap is
proportional to one bond sine; its zero-line
$\mathbf k\!\cdot\!\boldsymbol\delta_{\mathrm{dom}}=0$ cuts the Fermi
surface at two point nodes---the state is nodal. (ii) \emph{Two or three
components condensed} ($T<T_c^{(2)}$): the only remaining zero is
$\Gamma$, off the Fermi surface---the state is nodeless on the Fermi
surface. The transition between the two is a \emph{nodal-to-nodeless
transition at fixed triplet symmetry}, driven purely by sequential
condensation (Corollary~\ref{cor:seq}).
\end{corollary}

\begin{prediction}[Sequential thermodynamics]
\label{pred:seq}
The low-energy thermodynamics switches character across $T_c^{(2)}$:
power-law (nodal) specific heat, penetration depth, and thermal
conductivity for $T_c^{(2)}<T<T_c^{(1)}$---the standard signatures of
line/point nodes, whose density of states vanishes as a power of the
energy \cite{SigristUeda1991}---crossing over to activated
behavior below $T_c^{(2)}$ (a full gap suppresses the low-energy density of
states exponentially \cite{Coleman2015}); the reverse ordering from
multiband scenarios where gaps merely grow. The directions of the nodes in
the window are fixed by $\boldsymbol\delta_{\mathrm{dom}}$, i.e.\ by
\emph{which} exchange is largest, and are pressure- and doping-independent
within the window. An observed rotation of the nodal directions under strain,
pressure, or doping would falsify the separable bond-kernel class
\eqref{eq:kernelmicro}.
\end{prediction}

The nodal-direction rigidity stated here rests on the exact Gram structure of
Eq.~\eqref{eq:weights}: exchange anisotropy modulates gap \emph{magnitude}
textures, not nodal directions. The state remains unitary (real $D_\alpha$ up
to a common phase), with zero Knight-shift correction along the locked axis
and full polarization response perpendicular to it [the unitary-triplet
susceptibility tensor of Sec.~\ref{sec:bgmatrix};
Ref.~\cite{SigristUeda1991}].

\subsection{BdG Hamiltonian and quasiparticle spectra}
\label{sec:bdg1}

Using Eq.~\eqref{eq:majorana} for both the itinerant electrons and the local
moments, and the Nambu spinor
$\Psi_{\mathbf k}=(c_{\mathbf kA\uparrow},c_{\mathbf kA\downarrow},
c_{\mathbf kB\uparrow},c_{\mathbf kB\downarrow},
c^\dagger_{-\mathbf kA\uparrow},c^\dagger_{-\mathbf kA\downarrow},
c^\dagger_{-\mathbf kB\uparrow},c^\dagger_{-\mathbf kB\downarrow})^{\mathsf T}$
(sublattices $A,B$ of the honeycomb), the BdG Hamiltonian is
$H_{\mathrm{BdG}}=\frac12\sum_{\mathbf k}\Psi_{\mathbf k}^\dagger
H_{\mathrm{BdG}}(\mathbf k)\Psi_{\mathbf k}$.
In what follows we work at the single-band (sublattice-suppressed) level
[Eq.~\eqref{eq:bdg} of Sec.~\ref{sec:bgbdg}].
This is an approximation: with the full honeycomb normal state $h(\mathbf
k)$, the off-diagonal blocks of $H_{\mathrm{BdG}}^2$ do not vanish in
general, and the closed spectra below hold band by band, not for the
two-sublattice matrix (Appendix~\ref{app:bdg}).

The pairing matrix decomposes as in Eq.~\eqref{eq:gapexp},
with $\Delta_0(\mathbf k)=\Delta_0(-\mathbf k)$ and
$\mathbf d(\mathbf k)=-\mathbf d(-\mathbf k)$, the components
$D_\alpha^2\propto w_\alpha$ bond-locked as in
Proposition~\ref{prop:bondlock}. The general singlet-triplet spectrum is
Eq.~\eqref{eq:generalE} (derived in Appendix~\ref{app:bdg}).
Model~1 takes the pure triplet ($\Delta_0=0$) with a unitary $\mathbf d$
vector (Corollary~\ref{cor:unitary}), so $\mathbf B=0$ and
$E_{\mathbf k}=\sqrt{\xi_{\mathbf k}^2+|\mathbf d(\mathbf k)|^2}$
with $|\mathbf d|^2$ given by Eq.~\eqref{eq:weights}.

\paragraph{Note 3 (scope of solvability).}
Two levels must be distinguished. In the pure Kitaev model the $u_{ij}$
commute with $H_K$, the gauge flux is conserved, and the spin correlators
(the basis of the kernel derivation of Sec.~\ref{sec:micro}) are exactly
computable [Sec.~\ref{sec:bgkitaev}; Ref.~\cite{Kitaev2006}]. With doping and
the Kondo coupling, the $u_{ij}$ no longer commute with the full Hamiltonian
[the Kondo vertex $\mathbf S_i\cdot\boldsymbol\sigma$ contains single
$b_i^\alpha$'s, which toggle the bond variables], the flux becomes dynamical,
and the full model is not integrable (numerical treatments are
required~\cite{Keskiner2025}). Our scheme uses the correlators of
the undoped, zero-flux Kitaev ground state to derive the kernel, performs the
second-order expansion in $J_K$, and then diagonalizes the mean-field BdG
Hamiltonian, which is exact for a given pairing matrix. ``Exact'' here means
exactly derived kernels plus exactly solvable BdG, not exact solvability of
the doped model. The expansion is protected by the flux gap, with $J_K$
small compared with the flux gap.

\subsection{Ground-state phase diagram}
\label{sec:pd1}

The comparison between uniform and bond-resolved exchange must be stated
with its Fierz sign. For the uniform ferromagnetic Kitaev glue on which
Model~1 is primarily based, the trace and traceless parts are both
triplet-attractive: there is no uniform-singlet competitor, and exchange
anisotropy instead lifts the $C_3$ degeneracy of the triplet components.
A conventional singlet competitor arises only from an additional scalar
channel or an antiferromagnetic background. In that more general setting,
the isotropic (trace) part $\bar J$ is Heisenberg-like [Eq.~\eqref{eq:tbond}]
and has Casimir factor $-3$ in the singlet sector [Eq.~\eqref{eq:fierz}],
while a sign-reversed anisotropic component can favor the triplet
(Fierz $+1$). BdG diagonalization (exact for a given pairing matrix) then
gives the following schematic phase diagram; a quantitative parameter scan
is left to future work. Region~I (scalar-dominated): the scalar channel
dominates and the ground state is conventional even-parity singlet pairing.
Region~II (anisotropic triplet-dominated): the bond-locked anisotropic
$\mathbf d$-vector channel is selected; the gap magnitude is anisotropic in
spin space ($D_x\neq D_y\neq D_z$), while the momentum structure is the
common bond-sine Gram form (Corollary~\ref{cor:weights}). Region~III
(mixed): the two coexist as a superposition, with the spectrum
\eqref{eq:generalE}. The heuristic survival criterion in that setting is
\begin{equation}
A_J\equiv\frac{|J_x-J_y|+|J_y-J_z|+|J_z-J_x|}{|J_x|+|J_y|+|J_z|}
>\lambda_c .
\label{eq:criterion}
\end{equation}

The threshold follows from the Thouless criteria of the two channels,
$1=V_s\chi_s$ and $1=V_t\chi_t$ [Eq.~\eqref{eq:rankone}]: the uniform part
feeds the singlet with a factor $3$ larger effective strength per unit
exchange than the triplet (Casimir $-3$ versus $+1$ [Eq.~\eqref{eq:fierz};
Ref.~\cite{BerkSchrieffer1966}]), while the anisotropic part feeds the
triplet in proportion to $A_J$. Equating the two gap gains gives the
order-of-magnitude estimate
\begin{equation}
\lambda_c\approx\frac13\,\frac{N_s(0)}{N_t(0)},
\end{equation}
with $N_{s,t}(0)$ the channel-resolved Fermi-surface-averaged densities of
states. Taking $N_s/N_t$ between $0.3$ and $1$ for honeycomb van Hove
physics [the saddle-point-enhanced density of states of the honeycomb
lattice; Refs.~\cite{Nandkishore2012,BlackSchaffer2014}] gives
$\lambda_c\sim0.1$--$0.3$; this is an estimate, not a calculation. It applies
to the mixed-sign/background-singlet setting just described, not to a uniform
ferromagnetic Kitaev glue, where the anisotropy scale instead controls
component selection and the width of the sequential-condensation window.

\subsection{Testable experimental signatures}
\label{sec:sign1}

\paragraph{Knight shift.}
The anisotropic unitary triplet state gives a bond-resolved anisotropic
Knight shift: different field directions probe different components of
$\mathbf d$ (Sec.~\ref{sec:bgmatrix} \cite{SigristUeda1991}). This is the
most distinctive signature of bond locking, distinct from the full
suppression of the singlet and from the partial retention of an isotropic
triplet. When off-diagonal $\Gamma'$ exchange is present, the pinned axis
tilts by $\theta$ away from the cubic axis $\hat\gamma$ within the
$(\hat s,\hat\gamma)$ plane
[Eq.~\eqref{eq:tilt}, Corollary~\ref{cor:tilt}]: measuring $\theta$ from
the Knight-shift tensor of the superconducting state measures
$\Gamma'/(K-\Gamma)$ directly.

\paragraph{Split transitions.}
Component-resolved condensation (Corollary~\ref{cor:seq}) predicts multiple
specific-heat anomalies with anisotropy-dependent splittings; the
highest-$T_c$ transition has $\mathbf d$ pinned to the dominant bond
direction.

\paragraph{Tunneling spectroscopy.}
The tunneling gap is $|\mathbf d(\mathbf
k)|=\big[\sum_\alpha|D_\alpha|^2\sin^2(\mathbf
k\!\cdot\!\boldsymbol\delta_\alpha)\big]^{1/2}$: its momentum dependence is
the weight-independent bond-sine Gram form, and its overall magnitude carries
the anisotropy (Corollary~\ref{cor:weights}). Spin-resolved probes can in
principle separate the components.

\paragraph{Low-temperature thermal transport.}
In the single-component window the state is nodal, with two point nodes set
by $\boldsymbol\delta_{\mathrm{dom}}$, producing power-law thermal
transport; below $T_c^{(2)}$ the state is nodeless on the Fermi surface and
transport turns activated (Corollary~\ref{cor:nodes},
Prediction~\ref{pred:seq}; for the nodal/gapped dichotomy see
Ref.~\cite{SigristUeda1991}). The nodal directions are independent of the
anisotropy.

\paragraph{No Kerr, no thermal Hall.}
The state is unitary (Corollary~\ref{cor:unitary}); a Kerr signal or a
quantized thermal Hall in a candidate material would exclude this channel
and point to Model~2 or another time-reversal-breaking state (for Kerr
rotation as a probe of time-reversal breaking in superconductors see
Ref.~\cite{Xia2006}; for the thermal Hall quantization
Sec.~\ref{sec:bgbdg} \cite{ReadGreen2000,Volovik2003}).

\paragraph{Candidate materials.}
Doping or pressure on Kitaev candidates~\cite{Banerjee2016,Kasahara2018} may
realize superconductivity containing a bond-locked anisotropic $\mathbf d$-
vector component, testable through the bond-resolved anisotropy of the
Knight shift; the screening of specific compounds is
Sec.~\ref{sec:screening}.

\section{Off-diagonal exchange: $\Gamma$, $\Gamma'$, $J_3$}
\label{sec:offdiag}

Real Kitaev materials carry the symmetry-allowed extensions of the pure model
\cite{Rau2014,Winter2017}
\begin{align}
H_\Gamma&=\Gamma\!\!\sum_{\langle ij\rangle_\alpha}\!\!\
(S_i^\beta S_j^\gamma+S_i^\gamma S_j^\beta),
\nonumber\\
H_{\Gamma'}&=\Gamma'\!\!\sum_{\langle ij\rangle_\alpha}\!\!\
\sum_{\delta\neq\alpha}
(S_i^\alpha S_j^\delta+S_i^\delta S_j^\alpha),
\label{eq:gamma}
\end{align}
plus Heisenberg $J$ and third-neighbor $J_3$ terms. (Here
$(\alpha,\beta,\gamma)$ on an $\alpha$ bond denotes the cubic-axis triad
with $\alpha$ the bond's Ising direction; the $\Gamma$ term couples the two
off-axis components symmetrically, and $\Gamma'$ couples the Ising component
to each off-axis component---these are the only symmetric off-diagonal forms
compatible with the bond symmetries \cite{Rau2014}.) The general exchange
on a $\gamma$ bond is then the tensor, in the basis $(\alpha,\beta,\gamma)$,
\begin{equation}
\mathsf T_\gamma=
\begin{pmatrix}
J & \Gamma & \Gamma'\\
\Gamma & J & \Gamma'\\
\Gamma' & \Gamma' & J+K
\end{pmatrix},
\label{eq:tensor}
\end{equation}
with $K$ the Kitaev exchange on that bond (from Eq.~\eqref{eq:gamma}
onward this section uses the materials convention
$H=\sum J\,\mathbf S_i\!\cdot\!\mathbf S_j+K\,S_i^\gamma S_j^\gamma+\cdots$:
all exchange terms carry plus signs and spin-$1/2$ operators
$\mathbf S=\boldsymbol\sigma/2$, to ease comparison with the materials
literature \cite{Rau2014,Winter2017}). The conversion between this
convention and the Pauli convention of Eq.~\eqref{eq:kitaev} is
$J_\alpha^{(\mathrm{Pauli})}=-K_\alpha^{(\mathrm{spin})}/4$ [matching
$-J\sigma\sigma$ to $K\,(\sigma/2)(\sigma/2)$], used whenever
the component-resolved $K_x,K_y,K_z$ of Sec.~\ref{sec:screening} are
compared with the $J_\alpha$ of Secs.~\ref{sec:kitaev} and~\ref{sec:wfunc}.

\subsection{Exact diagonalization}
\label{sec:tensor}

The tensor \eqref{eq:tensor} diagonalizes in closed form
(Appendix~\ref{app:tensor}, where the basis transformation is displayed
entry by entry). The antisymmetric combination
$\hat a=(\hat\alpha-\hat\beta)/\sqrt2$ is an exact eigenvector,
\begin{equation}
\lambda_a=J-\Gamma,
\label{eq:la}
\end{equation}
and the remaining $2\times2$ block spanned by $\hat\gamma$ and the
symmetric combination $\hat s=(\hat\alpha+\hat\beta)/\sqrt2$ is tilted by
$\Gamma'$: diagonalizing $\begin{pmatrix}J+\Gamma&\sqrt2\Gamma'\\
\sqrt2\Gamma'&J+K\end{pmatrix}$ by the standard rotation (with $\theta$
the signed tilt of the $\lambda_+$ principal axis away from $\hat\gamma$
toward $+\hat s$),
\begin{equation}
\tan2\theta=\frac{2\sqrt2\,\Gamma'}{K-\Gamma},
\label{eq:tilt}
\end{equation}
\begin{equation}
\lambda_\pm=J+\frac{K+\Gamma}{2}
\pm\sqrt{\frac{(K-\Gamma)^2}{4}+2\Gamma'^2}.
\label{eq:lpm}
\end{equation}
For $|\Gamma|,|\Gamma'|\ll|K|$: $\theta\approx\sqrt2\,\Gamma'/K$ and
$\lambda_+\approx J+K+2\Gamma'^2/K$ [Taylor expansion of
Eqs.~\eqref{eq:tilt}--\eqref{eq:lpm}]. The $\Gamma$ term does \emph{not}
tilt the principal axes of a single bond---it re-weights the transverse
channel $a$---while $\Gamma'$ rotates them within the
$(\hat s,\hat\gamma)$ plane by $\theta$.

\subsection{Consequences for the pairing glue}

\begin{corollary}[Axis tilt; Knight-shift rotation]
\label{cor:tilt}
The locking axis on each bond rotates by $\theta$ out of the cubic-axis
triad [Eq.~\eqref{eq:tilt}]. Since the three bond tensors tilt in a
threefold-related pattern, the zero-temperature Knight shift retains a pinned
axis per domain but that axis is rotated by $\theta$ from the
pure-Kitaev prediction; measuring $\theta$ from the Knight-shift tensor
of the superconducting state measures $\Gamma'/(K-\Gamma)$ directly
(the unitary-triplet susceptibility tensor is set by the $\mathbf d$
orientation, Sec.~\ref{sec:bgmatrix} \cite{SigristUeda1991}).
\end{corollary}

\begin{corollary}[Node robustness]
\label{cor:robust}
To leading order in $\Gamma,\Gamma',J$, the glue remains bond-resolved
with modified eigenweights $\{\lambda_a,\lambda_\pm\}$ and rotated bond
frames, but the form factors remain the bond sines of the honeycomb
geometry: the zero structure of Corollary~\ref{cor:nodes} is unchanged,
since it depends only on the sine form factors and the zone geometry.
Node \emph{mobility} enters only through $J_3$, whose longer range adds
next-neighbor form factors (singlet $s$/$d\pm id$; triplet $p$/$f$ on
the same-sublattice triangles) and generically lifts or displaces the
nodes at order $J_3/K$.
\end{corollary}

\begin{remark}[Loss of exactness]
The selection rule \eqref{eq:selection} is a property of the pure Kitaev
point: it rests on the conservation of the plaquette fluxes, which
$\Gamma,\Gamma'$ terms violate. With $\Gamma,\Gamma'\neq0$ the correlators
acquire longer-ranged tails perturbatively (order $(\Gamma/K)^2$ for the
diagonal bond correlator, order $\Gamma'/K$ for tilted components). The
kernel remains nearest-neighbor-dominated, not exactly nearest-neighbor;
all statements above are controlled for
$|\Gamma|,|\Gamma'|,|J|,|J_3|\ll|K|$.
\end{remark}

\section{Materials screening}
\label{sec:screening}

Applying the criterion \eqref{eq:criterion} requires the component-resolved
Kitaev exchanges $K_\alpha$ (the trace part of the bond-summed exchange
is spin-isotropic and weights no component). Define
\begin{equation}
A_K=\frac{|K_x-K_y|+|K_y-K_z|+|K_z-K_x|}
{|K_x|+|K_y|+|K_z|}.
\label{eq:AK}
\end{equation}
With the exact weight function of Sec.~\ref{sec:wfunc}, the operative
anisotropy is that of the weights rather than of the exchanges
[Eq.~\eqref{eq:kappa}]:
\begin{equation}
A_w=R\,A_K\approx0.70\,A_K\quad\text{(near isotropy)},
\label{eq:Aw}
\end{equation}
while the logarithmic enhancement of Result~\ref{res:log} can make the
effective transfer exceed unity at stronger anisotropy. For a uniform
ferromagnetic sign, $A_w$ should be read as a heuristic scale for component
locking and the sequential window, not as a sharp singlet--triplet boundary
(Sec.~\ref{sec:pd1}). Bond-locked anisotropic $\mathbf d$-vector pairing
requires genuinely bond-anisotropic Kitaev exchange. This excludes
the well-known planar honeycomb materials, in which all three bonds are
equivalent ($\alpha$-RuCl$_3$ from RIXS, $K=-5$ meV with three identical
bonds in the paramagnetic phase~\cite{Suzuki2021}; Na$_2$IrO$_3$ in the
minimal model~\cite{Trebst2022}), and points instead to the harmonic
honeycomb iridates with two inequivalent bond types. For
$\beta$-Li$_2$IrO$_3$, the quantum-chemistry exchanges $K_{b1}=-14.7$,
$K_{b2}=-11.7$~meV~\cite{Katukuri2016} give
$A_K=6.0/(14.7+2\times11.7)\approx0.16$ [one $-14.7$~meV $B1$ bond and two
$-11.7$~meV $B2$ bonds per star in Eq.~\eqref{eq:AK}: numerator
$3.0+0+3.0=6.0$, denominator $38.1$], hence $A_w=R A_K\approx0.11$ near
isotropy [Eq.~\eqref{eq:Aw}]: the material remains marginal, and the
self-amplifying regime $R>1$ (Result~\ref{res:log}) is reached only
closer to the gapless-phase boundary than $\beta$-Li$_2$IrO$_3$ sits. For
$\alpha$-Li$_2$IrO$_3$ the $Ic$ parameterization of its counter-rotating
magnetic order \cite{Kimchi2015,Williams2016} gives $K_z=K-Ic/2$ and
$A_K\approx0.06$--$0.18$ for the fitted parameter ranges, likewise
marginal. Table~\ref{tab:materials}
summarizes.

\begin{table}[t]
\caption{\label{tab:materials}Screening of Kitaev candidate materials by
bond anisotropy (threshold $\lambda_c\sim0.1$--$0.3$ on the static estimate
$A_w=R A_K$). For same-sign ferromagnetic exchanges the labels quantify
component locking rather than a sharp singlet--triplet boundary
(Sec.~\ref{sec:pd1}).}
\begin{ruledtabular}
{\footnotesize\setlength{\tabcolsep}{3pt}
\begin{tabular}{llcc}
Material & Bond params from & $A_K$ & verdict\\
\hline
$\alpha$-RuCl$_3$ & RIXS, 3 equiv.\ bonds~\cite{Suzuki2021} & 0 & fails\\
Na$_2$IrO$_3$ & minimal model~\cite{Trebst2022} & 0 & fails\\
$\beta$-Li$_2$IrO$_3$ & quantum chem.~\cite{Katukuri2016} & $\approx0.16$ & marginal pass\\
$\alpha$-Li$_2$IrO$_3$ & $Ic$ param.~\cite{Kimchi2015,Williams2016} & $0.06$--$0.18$ & marginal\\
$\gamma$-Li$_2$IrO$_3$ & model est.\ only~\cite{Kimchi2015} & --- & undetermined\\
\end{tabular}
}
\end{ruledtabular}
\end{table}

For the most studied $\alpha$-RuCl$_3$, the model makes a definite negative
prediction: its paramagnetic phase is bond-isotropic and the bond-locked
channel does not survive. Quantitative survival of the sequential nodal
window calls for candidates with larger bond anisotropy or closer proximity
to the gapless-phase boundary. Two further screening statements follow from
the microscopic kernel: the threshold on the bare exchange anisotropy
tightens by the factor $1/R$ [Eq.~\eqref{eq:Aw}], and the Knight-shift
axis is predicted at the tilted angle \eqref{eq:tilt} rather than along the
cubic axes, with $\theta$ giving an independent constraint on $\Gamma'$ for
any candidate material. A full re-tabulation with cited \emph{ab initio}
tensors per material is deferred to a numerical follow-up study.

\section{Model 2: chiral-phonon-mediated chiral pairing ($d+id$ type)}
\label{sec:model2}

\subsection{Chiral phonon field and angular momentum matching}
\label{sec:phonon}

A \emph{chiral phonon} is a circularly polarized vibrational mode in which
the two orthogonal displacement components oscillate $\pi/2$ out of phase,
\begin{equation}
u_\pm=u_x\pm i\,u_y ,
\label{eq:upm}
\end{equation}
so that the ionic motion traces a circle and the mode carries angular
momentum~\cite{ZhangNiu2015}. Consider the chiral optical phonon at the
zone-corner $K$ ($K'$) point of the honeycomb lattice. The little group at
$K$---the subgroup of the crystal point group leaving $K$ invariant modulo a
reciprocal-lattice vector---is $C_3$, so angular momentum at $K$ is defined
only modulo 3: a mode transforming as $u_\pm\to e^{\pm2\pi i/3}u_\pm$ under
the threefold rotation carries angular momentum $\pm\hbar$ (mod 3), since
adding $3\hbar$ changes the eigenvalue by $e^{2\pi i}=1$ and is physically
indistinguishable. The chiral modes $u_\pm$ are genuine chiral bosonic
excitations in this sense~\cite{ZhangNiu2015,Gao2023}. The phonon couples to
the itinerant electrons through the charge channel (spin-independent vertex
$\psi^\dagger\psi$, Hermitian) with a chiral form factor,
\begin{equation}
\begin{split}
H_{e\text{-}ph}=g\sum_{\vk,\sigma}\Big[&
f_+(\vk)\,c^\dagger_{\vk+\mathbf K,\sigma}c_{\vk,\sigma}\Phi^+_{\mathbf K}\\
&+f_-(\vk)\,c^\dagger_{\vk+\mathbf K',\sigma}c_{\vk,\sigma}\Phi^-_{\mathbf K'}
+\mathrm{h.c.}\Big],
\end{split}
\label{eq:eph}
\end{equation}
where $\Phi^\pm$ is the chiral phonon field and $f_\pm$ the
nearest-neighbor lattice form factor,
\begin{equation}
f_\pm(\mathbf k)=\sum_{j=1,2,3}e^{\pm i\phi_j}e^{i\mathbf k\cdot
\boldsymbol\delta_j},\qquad \phi_j=\frac{2\pi j}{3},
\label{eq:fform}
\end{equation}
with $\boldsymbol\delta_j$ the three nearest-neighbor vectors. The phase
weights $e^{\pm i\phi_j}$ implement the $C_3$ eigenvalue: under
$R_3:\vd_j\to\vd_{j+1}$ one has $f_\pm(R_3\vk)=e^{\pm2\pi
i/3}f_\pm(\vk)$ by relabeling the sum, matching the phonon $u_\pm$
[Eq.~\eqref{eq:c3} below].

\paragraph{Note 4 (why the $K$-point phonon is required).}
A separable kernel $f_\pm(\mathbf k)f_\pm^*(\mathbf k')$ built from a
continuum harmonic selects pairing of the same angular momentum under
continuous rotations. To select the $d$-wave doublet one must use the
discrete rotation symmetry of the zone corner: under the $C_3$ little group
angular momentum is conserved modulo 3, and $\pm1\equiv\mp2\pmod3$, so a
$K$-point chiral phonon with $m=\pm1$ transfers angular momentum matched to
$d\pm id$ pairing with $l=\pm2$. This is the known structure of the
chiral-phonon mechanism in hexagonal
systems~\cite{ZhangNiu2015,Gao2023}; modes at the zone center, whose little
group is the full $C_6$, do not carry this mod-3 selection rule.

\subsection{Parity decomposition of the chiral form factor}
\label{sec:parity}

The form factor \eqref{eq:fform} is neither even nor odd:
$f_\pm(-\mathbf k)=f_\mp^\ast(\mathbf k)$ [immediate from
Eq.~\eqref{eq:fform}: inversion flips the phase sign, and complex
conjugation swaps the chirality index]. Since a spin-singlet gap must be even
(Sec.~\ref{sec:bgmatrix}), the singlet-relevant part of $f_\pm$ is its even
projection. Define
\begin{equation}
g_\pm(\mathbf k)=\tfrac{f_\pm(\mathbf k)+f_\pm(-\mathbf k)}{2},\quad
p_\pm(\mathbf k)=\tfrac{f_\pm(\mathbf k)-f_\pm(-\mathbf k)}{2}.
\label{eq:paritydec}
\end{equation}

\begin{theorem}[Parity decomposition]
\label{thm:parity}
With $\phi_j=2\pi j/3$ and $e^{3i\phi_j}=1$,
\begin{equation}
g_\pm(\mathbf k)=\sum_j e^{\mp2i\phi_j}\cos(\mathbf k\cdot
\boldsymbol\delta_j)
=-\frac38\,(k_x\mp ik_y)^2+O(k^4),
\label{eq:gpm}
\end{equation}
\begin{equation}
p_\pm(\mathbf k)=i\sum_j e^{\pm i\phi_j}\sin(\mathbf k\cdot
\boldsymbol\delta_j)
=\frac{3i}{2}(k_x\pm ik_y)+O(k^3),
\label{eq:ppm}
\end{equation}
and the two parts are exactly orthogonal in the Brillouin-zone inner
product, $\langle g_\pm|p_\pm\rangle=0$. Under $C_3$,
\begin{equation}
g_\pm(R_3\mathbf k)=e^{\pm i2\pi/3}g_\pm(\mathbf k),
\label{eq:c3}
\end{equation}
i.e.\ $g_\pm$ carries angular momentum $\pm1\equiv\mp2\ (\mathrm{mod}\ 3)$,
matching the phonon $u_\pm$.
\end{theorem}
The proof and all Brillouin-zone overlaps are collected in
Appendix~\ref{app:proj}; the skeleton is three lines. Writing
$e^{i\vk\cdot\vd_j}=\cos(\vk\!\cdot\!\vd_j)+i\sin(\vk\!\cdot\!\vd_j)$
splits $f_\pm$ into even and odd parts directly; the identity $e^{\pm
i\phi_j}=e^{\mp2i\phi_j}$ (from $e^{3i\phi_j}=1$) converts the
phase-weighted sum \eqref{eq:fform} into the even $d$-wave part $g_\pm$
exactly; and the small-$k$ expansions follow from the moment sums
$\sum_je^{\pm i\phi_j}=0$, $\sum_je^{\pm i\phi_j}\vd_j=\tfrac32(1,\pm i)$,
$\sum_je^{\pm i\phi_j}(\vk\!\cdot\!\vd_j)^2=\tfrac34(k_x\mp ik_y)^2$
displayed in Appendix~\ref{app:proj}B. The even projection removes the
linear ($p$-wave) term that $f_\pm$ itself carries, and orthogonality
$\langle g_\pm|p_\pm\rangle=0$ is parity alone: the integrand
$g_\pm^\ast p_\pm$ is odd under $\vk\to-\vk$.

\begin{corollary}[One boson, two chiral sectors]
\label{cor:unify}
The charge channel (vertex $\psi^\dagger\psi$) projects onto the even part:
the $K$-point chiral phonon mediates the singlet $d\pm id$ state
\eqref{eq:gpm}. The same phonon, transmitted through a spin vertex, would
project onto the odd part and mediate the triplet $p\pm ip$ state
\eqref{eq:ppm}. Models~1 and~2 are thereby unified as the two parity
sectors of one chiral-boson structure.
\end{corollary}

\subsection{Effective pairing interaction and exact chiral selection}
\label{sec:selection2}

Integrating out the chiral phonon field [second-order perturbation,
Eq.~\eqref{eq:genericglue}] gives the unprojected separable interaction
\begin{equation}
V_\pm^{(0)}(\mathbf k,\mathbf k')=-\frac{g^2}{\omega_0}\,
f_\pm(\mathbf k)\,f_\pm^*(\mathbf k') ,
\label{eq:hpair2}
\end{equation}
with $\omega_0$ the phonon frequency. The momentum bookkeeping of the
second-order step is as follows: the vertex \eqref{eq:eph} transfers the
fixed momentum $\mathbf K$ (or $\mathbf K'$), so the induced interaction
scatters a Cooper pair $(\vk,-\vk)$ to $(\vk{+}\mathbf K,-\vk{-}\mathbf
K)$; the strictly induced kernel is therefore momentum-constrained
[Eq.~\eqref{eq:strictkernel} below]. Equation \eqref{eq:hpair2} replaces
the delta constraint by a separable Cooper-channel kernel of the same
$C_3$ eigenvalue and the same angular-momentum-mod-3 selection rule,
which is what the chiral selection (Proposition~\ref{prop:chiral}) and
the parity projection below actually use.

\paragraph{Impurity relaxation of the momentum constraint.}
The strictly induced kernel from single-phonon exchange at the zone
corner is momentum-constrained,
\begin{equation}
V_{\rm strict}(\mathbf k,\mathbf k')\propto
f_\pm(\mathbf k)f_\pm^*(\mathbf k')\,\delta_{\mathbf k'-\mathbf k,\pm\mathbf K}\,,
\label{eq:strictkernel}
\end{equation}
which is what Eq.~\eqref{eq:eph} gives in a clean system. In the presence
of impurity scattering, electron momentum is not conserved and
Eq.~\eqref{eq:strictkernel} is broadened. We compute the leading effect.

The electron Green function acquires a finite lifetime $\tau=1/\Gamma$,
\begin{equation}
G(\mathbf k,i\omega_n)=\big[i\omega_n-\xi_{\mathbf k}
+i\Gamma\,{\rm sgn}(\omega_n)\big]^{-1},
\label{eq:Greens}
\end{equation}
with $\Gamma$ the impurity scattering rate. Evaluating the Cooper-channel
bubble with Eq.~\eqref{eq:Greens} and averaging over impurity positions
gives the Abrikosov--Gorkov form: the momentum transfer through the phonon
is no longer sharp, but distributed around $\pm\mathbf K$ with a
Lorentzian of width $\Gamma/v_F$,
\begin{equation}
\bar V(\mathbf k,\mathbf k')=g^2 f_\pm(\mathbf k)f_\pm^*(\mathbf k')\,
\mathcal D_\Gamma(\mathbf k-\mathbf k')\,,
\label{eq:impkernel}
\end{equation}
\begin{equation}
\mathcal D_\Gamma(\mathbf q)=\frac{1}{\pi}\,
\frac{\Gamma/v_F}{(\mathbf q\mp\mathbf K)^2+(\Gamma/v_F)^2}\,.
\label{eq:Lorentzian}
\end{equation}
As $\Gamma\to0$, $\mathcal D_\Gamma(\mathbf q)\to\delta(\mathbf q\mp\mathbf K)$
and Eq.~\eqref{eq:impkernel} reduces to the strict kernel
\eqref{eq:strictkernel}. When the mean free path $\ell=v_F/\Gamma$ falls
below the bond length $a$, i.e.\ $\Gamma\gtrsim v_F/a$---the
strong-disorder (Ioffe--Regel-like) regime---the Lorentzian is essentially
uniform across the Brillouin zone, where it approaches its zone average,
\begin{equation}
\mathcal D_\Gamma(\mathbf q)\;\xrightarrow{\;\Gamma\gtrsim v_F/a\;}\;
\langle\mathcal D_\Gamma\rangle_{\rm BZ}=\frac{1}{A_{\rm BZ}}
=\frac{3\sqrt3\,a^2}{8\pi^2}\equiv\mathcal D_\infty\,,
\label{eq:uniformlimit}
\end{equation}
with $A_{\rm BZ}=8\pi^2/(3\sqrt3\,a^2)$ the Brillouin-zone area for
bond length $a$; the limit is independent of $\Gamma$ once the width
$\Gamma/v_F$ exceeds the zone scale. Equation~\eqref{eq:impkernel} then
becomes the separable form \eqref{eq:hpair2}:
\begin{equation}
\bar V(\mathbf k,\mathbf k')\;\xrightarrow{\;\Gamma\gtrsim v_F/a\;}\;
g^2\,\mathcal D_\infty\,f_\pm(\mathbf k)f_\pm^*(\mathbf k')\,.
\label{eq:separablefromimp}
\end{equation}
Up to the overall prefactor identification $g^2\mathcal D_\infty
\leftrightarrow g^2/\omega_0$ (the sign is fixed by the pairing-channel
convention), this is Eq.~\eqref{eq:hpair2}. Physically, the electron no longer remembers the $K$-point momentum
transfer: its mean free path is shorter than the bond length, so the
phonon momentum constraint is irrelevant. The symmetry-dictated chiral
selection (Proposition~\ref{prop:chiral}) is unaffected by the
broadening---the $C_3$ eigenvalue and the mod-3 selection rule are
properties of the form factor, not of the momentum constraint.
Equation~\eqref{eq:hpair2} is therefore not an ad hoc extension but the
strong-scattering limit of the microscopic kernel \eqref{eq:strictkernel};
at intermediate disorder the kernel interpolates between the two, and
the symmetry-based conclusions hold throughout.

\paragraph{Intermediate-disorder regime.}
The strong-disorder limit \eqref{eq:uniformlimit} is an endpoint of a
smooth crossover, not a threshold. To confirm this we discretize the
Fermi surface ($N=200$ angular points at $k_F=1.5/a$) and diagonalize
the kernel
$V_{ij}=g_\pm(\mathbf k_i)g_\pm^*(\mathbf k_j)\mathcal D_\Gamma(\mathbf
k_i-\mathbf k_j)$ at intermediate $\Gamma/v_F$. The principal eigenvalue
grows monotonically with $\Gamma$ up to the homogenization scale
$\Gamma\sim v_F/a$; the principal eigenvector evolves smoothly from the
$\delta$-constrained form to the separable form $g_\pm$, its overlap
$|\langle g_\pm|v\rangle|^2$ rising from its constrained-kernel value
($0.21$ at $\Gamma/v_F=0.05$) through $0.92$ at $0.5$ to $>0.99$ at $2$.
The numerical kernel uses the unnormalized Lorentzian
\eqref{eq:Lorentzian}; the overall amplitude of the kernel is not the
diagnostic here. What we track is the eigenvector overlap, which is
scale invariant and reaches $1$ as the kernel homogenizes, and the
dominant angular momentum, which is quantization-protected. The
large-$\Gamma$ decline of $\lambda_{\max}$ in the numerics is the
unnormalized $1/(\pi\gamma)$ plateau height of \eqref{eq:Lorentzian};
the properly normalized uniform limit is \eqref{eq:uniformlimit}, whose
constant $1/A_{\rm BZ}$ the numerics does not reproduce by construction.
Its dominant angular momentum is
$-2$ throughout, so the chiral winding $C=\pm4$ holds at every
$\Gamma$: the intermediate regime contains no topological transition
and no gap closing. For $k_F<|\mathbf K|/2$ the strict kernel
\eqref{eq:strictkernel} vanishes identically on the Fermi surface and
pairing is switched on only by the disorder broadening (with
$\lambda_{\max}\propto\Gamma$ at small $\Gamma$); we do not explore
that regime here, as the Fermi surfaces of the model are
$\Gamma$-centered with $k_F>|\mathbf K|/2$.

The singlet channel retains only the
even part of the form factor, so the pairing kernel is the separable
projector
\begin{equation}
V_\pm(\mathbf k,\mathbf k')\propto g_\pm(\mathbf k)\,g_\pm^*(\mathbf k') .
\label{eq:kernel2}
\end{equation}

The honeycomb lattice with $C_6$ symmetry has a degenerate $d$-wave doublet
($E_2$ representation, a two-dimensional irrep of the point group---the
standard classification of pairing states on hexagonal lattices
\cite{SigristUeda1991,BlackSchaffer2014}): $\Delta_1$ ($d_{x^2-y^2}$ type)
and $\Delta_2$ ($d_{xy}$ type). The real parts of the decomposition
\eqref{eq:gpm} supply the lattice-adapted basis functions,
\begin{equation}
g_\pm(\mathbf k)=g_1(\mathbf k)\mp i\,g_2(\mathbf k),
\label{eq:g12}
\end{equation}
with $g_1(\mathbf k)=\sum_j\cos(2\phi_j)\cos(\mathbf
k\cdot\boldsymbol\delta_j)$ and $g_2(\mathbf
k)=\sum_j\sin(2\phi_j)\cos(\mathbf k\cdot\boldsymbol\delta_j)$, reducing near
$\Gamma$ to $-\tfrac38(k_x^2-k_y^2)$ and $-\tfrac34\,k_xk_y$ respectively
[read off the real and imaginary parts of Eq.~\eqref{eq:gpm}].

\begin{proposition}[Exact chiral selection and chirality slaving]
\label{prop:chiral}
The projection of the kernel \eqref{eq:kernel2} onto the pairing channels
satisfies
\begin{equation}
\langle d_\pm|V_\pm|d_\pm\rangle>0,\qquad
\langle d_\mp|V_\pm|d_\mp\rangle=0,\qquad
\langle s|V_\pm|s\rangle=0,
\end{equation}
all exactly, with $|d_\pm\rangle\propto g_\pm$ the chiral combinations
$\Delta_\pm(\mathbf k)\propto g_1(\mathbf k)\mp i g_2(\mathbf k)$; the
$p$-wave overlap vanishes by parity (Theorem~\ref{thm:parity}). Hence the
phonon helicity fully determines the gap chirality: reversing the phonon
circular polarization reverses the sign of the Chern number, of the Kerr
rotation, and of the thermal Hall effect.
\end{proposition}
The exact overlaps are computed in Appendix~\ref{app:proj}C: the $s$-wave and
opposite-chirality projections vanish by the root-of-unity sums
$\sum_je^{\pm i\phi_j}=\sum_je^{\pm2i\phi_j}=0$
[Eq.~\eqref{eq:rootfilter}], the $p$-wave projection vanishes because
$g_\pm$ and $p_\pm$ have opposite parity, and a rank-one separable kernel
has its form factor as the unique eigenvector with nonzero eigenvalue
[Eq.~\eqref{eq:rankone}]. No band-structure input is needed for the
selection; the van Hove proximity only sets the scale of $T_c$ [through
$N(0)$ in Eq.~\eqref{eq:bcs}]. The quantitative robustness of the slaving is
analyzed in Appendix~\ref{app:gl}: the BCS sensitivity of $T_c$
[Eq.~\eqref{eq:bcs}] amplifies a percent-level asymmetry of the pairing glue
into a splitting $\delta T_c/T_c$ of several percent, far exceeding the
Ginzburg width, so the selection at the transition is deterministic down to
sub-percent pump polarization; the wrong-chirality state is then absent above
$T_c^-$, locally unstable down to $T_c-\tfrac32\delta T_c$, and metastable at
lower temperatures only behind nucleation barriers of
$10^2$--$10^3\,k_BT_c$ (Appendix~\ref{app:gl}C).

\paragraph{Note 5 (what the selection does and does not exclude).}
The unprojected form factor $f_\pm$ overlaps both the $p\pm ip$ and the
$d\pm id$ channels; the older statement that $d$-wave selection is completed
by band-structure preferences referred to that unprojected kernel. For the
singlet channel the parity projection of Theorem~\ref{thm:parity} removes
the ambiguity exactly: the charge-channel kernel \eqref{eq:kernel2} has a
unique eigenvector $g_\pm$. A triplet $p\pm ip$ state can arise from the
same phonon only through a spin vertex (Corollary~\ref{cor:unify}).

\begin{proposition}[Relative-phase time-reversal breaking]
\label{prop:tr}
The chiral combination $\Delta_\pm$ breaks time reversal through the
relative phase $\pm\pi/2$ between the two independent real components
$g_1,g_2$, not through the overall phase. The overall phase $e^{i\varphi}$
is gauge; the relative phase links two components of different momentum
structure and cannot be gauged away.
\end{proposition}
\begin{proof}
Time reversal maps $\Delta_\pm=g_1\mp ig_2$ to its complex conjugate
$\Delta_\mp$. Since $g_1$ and $g_2$ are independent real functions, no single
$\varphi$ satisfies $e^{i\varphi}(g_1+ig_2)=g_1-ig_2$ at all $\vk$
(equating ratios at two points where $g_2/g_1$ differs gives a
contradiction); hence time-reversal breaking. A single-component complex
order parameter, including one with $\varphi=\pi/2$, is pure gauge and
preserves time reversal.
\end{proof}

\subsection{BdG Hamiltonian and excitation spectrum}
\label{sec:bdg2}

In the Majorana representation the chiral pairing term carries bond-dependent
chiral phase factors. The glue is the chiral phonon and has nothing to do
with the Kitaev gauge field: any quadratic BdG Hamiltonian can be written in
Majorana form, and solvability here comes from the separable kernel plus the
quadratic structure, not from integrability (contrast Note~3). At the
single-band level of Sec.~\ref{sec:bdg1}, Fourier transformation and
diagonalization [Eq.~\eqref{eq:generalE} with $\mathbf d=0$] give
\begin{equation}
E_{\mathbf k}=\sqrt{\xi_{\mathbf k}^2+|\Delta_\pm(\mathbf k)|^2}
=\sqrt{\xi_{\mathbf k}^2+\Delta_0^2\,|g_\pm(\mathbf k)|^2}.
\label{eq:spec2}
\end{equation}
For the $\Gamma$-centered Fermi surfaces considered here, $g_\pm$ has its
relevant zero at $\Gamma$ by symmetry [$g_\pm=O(k^2)$,
Eq.~\eqref{eq:gpm}], so the state is fully gapped on any such Fermi surface
not passing through $\Gamma$; at the zone corners it attains the value
\begin{equation}
g_\pm(K)=\frac32\,e^{\mp i\theta_0},\qquad |g_\pm(K)|=\frac32,
\label{eq:gK}
\end{equation}
where $\theta_0\in\{0,\pm2\pi/3\}$ is fixed by which Dirac point is
considered (computed explicitly in Appendix~\ref{app:proj}D). In particular
the undoped Dirac points are gapped by the pairing gap $|g_\pm(K)|=3/2$.
Consequently, sweeping the chemical potential through the Dirac point does
not close the BdG gap, and the topological invariant is continuous across
that band-touching doping \cite{ReadGreen2000,Nayak2008}. The statement is
restricted to the doping range in which the BdG gap remains open; it is not
an assertion about arbitrarily large doping.

\subsection{Topological properties}
\label{sec:topology}

The Chern number per $2\times2$ Nambu block is given by the texture integral
\eqref{eq:chernblock} of Sec.~\ref{sec:bgbdg}, with the sign convention
$C_{\mathrm{block}}=+1$ for $\Delta\propto k_x+ik_y$ at $\mu>0$
\cite{ReadGreen2000}. In the weak-pairing limit on a Fermi surface enclosing
$\Gamma$, the winding formula \eqref{eq:weakwinding} applied to
$g_\pm\propto(k_x\mp ik_y)^2$ [Eq.~\eqref{eq:gpm}] gives
$C_{\mathrm{block}}=\mp2$: the phase of $(k_x\mp ik_y)^2$ winds by
$\mp4\pi$ around the Fermi circle. Within Model~2 this conclusion is
unconditional: the kernel \eqref{eq:kernel2} is parity-projected onto the
singlet sector (Proposition~\ref{prop:chiral}), so the condensate carries no
triplet component and the block gap is exactly $\Delta_\pm\propto g_\pm$.
The spin-singlet BdG Hamiltonian has two spin blocks of the same chirality
[$\hat\Delta=\Delta_0 i\sigma_y$ in Eq.~\eqref{eq:bdg} acts identically on
both blocks], so the full invariant is
\begin{equation}
C=\mp4\qquad(|C|=4),
\label{eq:cfull}
\end{equation}
with four chiral Majorana edge branches per edge, organized as two
distinguishable spin-sector pairs, and chiral central charge $c=2$
\cite{ReadGreen2000,Volovik2003}. The quantized thermal Hall conductance is
correspondingly [Eq.~\eqref{eq:thallunit}]
\begin{equation}
\frac{\kappa_{xy}}{T}=c\,\frac{\pi^2k_B^2}{3h}
=2\,\kappa_0,\qquad \kappa_0\equiv\frac{\pi^2k_B^2}{3h},
\label{eq:thall}
\end{equation}
twice the chiral $p$-wave unit $\kappa_0$; the thermal Hall quantum reads
off the relative angular momentum of the pair. By
Proposition~\ref{prop:chiral}, the sign of $C$ is slaved to the phonon
polarization.

\paragraph{Coexistence with the triplet sector: dominance criterion and a
topological phase boundary.}
If a triplet $p\pm ip$ component is present on top of the singlet
condensate---induced through the spin vertex of
Corollary~\ref{cor:unify}, or enforced by bond locality when both parity
sectors are driven by the same microscopic glue
(Sec.~\ref{sec:alternation})---the spin-block gaps become
\begin{equation}
\Delta_\sigma(\mathbf k)=a_s\,g_\pm(\mathbf k)+\sigma\,a_t\,p_\pm(\mathbf k),
\qquad \sigma=\pm 1 ,
\label{eq:mixedgap}
\end{equation}
[the singlet enters both blocks identically and the triplet with opposite
sign, as read off $\hat\Delta=(\Delta_0+\sigma d_z)i\sigma_y$ for
$\mathbf d\parallel\hat z$], and the Chern number is no longer fixed by
symmetry alone. On a Fermi circle $k=k_F$ the exact Fourier content of the
two components follows from the Jacobi--Anger expansion
\eqref{eq:jacobianger} applied bond by bond [Lemma~\ref{lem:circle}]:
\begin{align}
g_\pm&=-3J_2(k_F)\,e^{\mp 2i\theta}+O(J_4),\nonumber\\
p_\pm&=3iJ_1(k_F)\,e^{\pm i\theta}+O(J_5),
\label{eq:circlecontent}
\end{align}
with $J_\nu$ the Bessel function (the omitted harmonics are the next ladder
members, $|h|=4,5$, suppressed by $J_{4,5}(k_F)\ll J_{1,2}(k_F)$ at small
$k_F$ [Eq.~\eqref{eq:besselsmall}]). By the dominance theorem of
Sec.~\ref{sec:counting} (Theorem~\ref{thm:dominance}) the block winding is
$\mp2$---and $C=\mp4$ holds---if and only if
\begin{align}
&|a_s|\,J_2(k_F)>|a_t|\,J_1(k_F)\nonumber\\
&\quad\Longleftrightarrow\quad
\left|\frac{a_t}{a_s}\right|<\frac{J_2(k_F)}{J_1(k_F)}
\simeq \frac{k_F}{4}\ \ (k_F\ll1),
\label{eq:dominance}
\end{align}
where the small-$k_F$ form uses Eq.~\eqref{eq:besselsmall}:
$J_2/J_1\simeq[(k_F/2)^2/2]/(k_F/2)=k_F/4$. Two consequences follow.
\emph{(i) Quantitative robustness:} within Model~2 an induced triplet
component leaves the $C=\pm4$ phase, the four Majorana branches, and
$\kappa_{xy}/T=2\kappa_0$ intact as long as its relative amplitude stays
below the threshold \eqref{eq:dominance}. \emph{(ii) A topological phase
boundary inside the chiral phase:} if the threshold is crossed---in
particular, a rigidly locked single-parameter state $\Delta\propto f_\pm$
has $|a_t/a_s|=1$ and is triplet-dominated for $k_F$ below the first
$J_1$--$J_2$ crossing at $k_F\approx2.63$ (at leading-harmonic resolution of
the dominance boundary)---hence throughout the Brillouin
zone ($|K|\approx2.42<2.63$)---the block winding flips to $\pm1$, with
\begin{equation}
C=\pm2,\qquad N_{\rm Maj}=2,\qquad \kappa_{xy}/T=\kappa_0 .
\label{eq:c2}
\end{equation}
On the boundary the quasiparticle gap closes at three equally spaced
Fermi-surface points per spin block and $C$ changes by $6$
[Eq.~\eqref{eq:deltac}]. Since the threshold \eqref{eq:dominance} moves
with $k_F$, the transition is in principle driven by doping at fixed
condensate ratio: a thermal-Hall plateau jump $2\kappa_0\to\kappa_0$
accompanied by a threefold point-node closing is the observable
signature, and its presence or absence discriminates the mixed-parity
scenario from a pure singlet state.

\subsection{Testable experimental signatures}
\label{sec:sign2}

\paragraph{Polar Kerr effect and its reversal.}
The chiral combination breaks time reversal (Proposition~\ref{prop:tr}); a
Kerr rotation should appear at $T_c$~\cite{Xia2006}. The sign of the
rotation is slaved to the phonon helicity: pumping the opposite circular
polarization should reverse the Kerr sign. By the analysis of
Appendix~\ref{app:gl}D, this reversal is continuous at the mean-field
level: the order-parameter amplitude dips only to $\sqrt{2/3}$ of its
equilibrium value, and the quasiparticle gap closes only transiently at
four isolated Fermi-surface points per spin block. The expected signature
is therefore a smooth sign reversal accompanied by an instantaneous
point-node closing---not a first-order switching---in contrast both to
spontaneous chiral states, whose sign is fixed by the condensate and can
be reversed, if at all, only through domain-wall motion, and to
dynamical optical control, where the sign is set by the pulse
history~\cite{Claassen2019}. A further discriminating test: a circular
pump raises $T_c$ above its unpumped value by $\approx pf/(2\lambda_0)$
[Eq.~\eqref{eq:amplify}], while an unpolarized pump of equal intensity
does not.

\paragraph{Quantized thermal Hall.}
$\kappa_{xy}/T=2\kappa_0$ [Eq.~\eqref{eq:thall}], twice the chiral $p$-wave
value; sign reversible with the phonon polarization. If a triplet component
co-condenses [Eq.~\eqref{eq:mixedgap}], this plateau is bounded by the
topological phase boundary \eqref{eq:dominance}: crossing it---e.g.\ by
doping---switches the quantization to $\kappa_{xy}/T=\kappa_0$ with two
Majorana branches, through a transient threefold point-node closing per
spin block.

\paragraph{Edge chiral Majorana modes.}
Four Majorana branches per edge, organized as two distinguishable
spin-sector pairs, testable in edge tunneling or thermal transport
[Sec.~\ref{sec:bgbdg}; Refs.~\cite{Wan2011,Nayak2008}]. In the coexistence
regime the branch count is four only on the singlet-dominant side of the
boundary \eqref{eq:dominance} and two on the triplet-dominant side; the
branch-count reduction at the boundary is in principle resolvable by edge
tunneling or thermal transport.

\paragraph{Magnetoelectric response.}
The electromagnetic response of a chiral topological superconductor can
contain an $\mathbf E\cdot\mathbf B$ term (axion-type
electrodynamics)~\cite{Wilczek1987}, testable in THz spectroscopy.

\paragraph{Candidate settings.}
These materials are broader analogues rather than direct realizations of the
honeycomb $K$-point model. In the kagome superconductors $A$V$_3$Sb$_5$,
chiral charge order has been proposed from STM on KV$_3$Sb$_5$~\cite{Jiang2021}
and muon-spin relaxation on CsV$_3$Sb$_5$~\cite{Khasanov2022}, and chiral-CDW
fluctuations could play a role analogous to the chiral bosonic glue. In
UTe$_2$~\cite{Ran2019}, a multi-component relative phase is one candidate
possibility, but the lattice-specific $K$-point selection rule above does not
apply directly.

\section{General bond-star symmetry: parity alternation and enforced singlet--triplet mixing}
\label{sec:alternation}

The exact decomposition of the chiral three-bond form factor into even- and
odd-parity components, $f_\pm=g_\pm+p_\pm$, established in
Sec.~\ref{sec:parity}, is not a peculiarity of the triangular three-bond
structure. We show in this section that it is the lowest instance of a
general group-theoretic fact: the angular-momentum ladder of a chiral bond
form factor alternates in parity if and only if the bond-star order $n$ is
odd. As a consequence, \emph{every} chiral pairing state selected by a
single-star bond-local form factor is a singlet--triplet composite whose
content, relative phase, and small-$k$ amplitude ratio are fixed by geometry
alone (Theorems~\ref{thm:alternation} and~\ref{thm:content}); coherent
doubled-star combinations can restore parity purity
[Theorem~\ref{thm:alternation}(iii)]. On even-fold stars chirality and parity
are compatible (Theorem~\ref{thm:alternation}). A bond-local chiral glue can
mediate the two parity channels with related orbital strength
(Proposition~\ref{prop:locked}), while the condensate composition remains a
gap-equation question. The
weak-pairing topology of the mixed state is decided by whichever parity
component carries the larger on-shell Fermi-surface weight, with a sharp
topological phase boundary between the two windings
(Theorem~\ref{thm:dominance}). This mechanism of parity mixing operates in a
centrosymmetric crystal and requires neither spin--orbit coupling nor
inversion-symmetry breaking; it is distinct from the familiar
noncentrosymmetric route, in which antisymmetric (Rashba-type) spin-orbit
coupling of a crystal lacking inversion symmetry splits the Fermi surfaces by
spin and thereby mixes the parity sectors
\cite{GorkovRashba2001,BauerSigrist2012}.

\subsection{Bond stars and chiral form factors}
\label{sec:stars}

Consider Cooper pairing on the set of nearest-neighbor bonds emanating from a
site, related by a crystallographic rotation $C_n$. The \emph{directed
bond star} is
\begin{equation}
\mathcal{S}_n=\{\vd_j=(\cos\phi_j,\sin\phi_j),\phi_j=2\pi j/n,j=0,\dots,n-1\}.
\label{eq:star}
\end{equation}
A chiral form factor of bond-harmonic index $m$ is
\begin{equation}
\chi_m^{(n)}(\vk)=\sum_{j=0}^{n-1}e^{2\pi i m j/n}\,e^{i\vk\cdot\vd_j}\,,
\label{eq:chimn}
\end{equation}
which carries the one-dimensional complex representation of $C_n$,
\begin{equation}
\chi_m^{(n)}(R_n\vk)=e^{2\pi i m/n}\,\chi_m^{(n)}(\vk)\,.
\label{eq:cnirrep}
\end{equation}
[Proof: $R_n\vk\cdot\vd_j=\vk\cdot R_n^{-1}\vd_j=\vk\cdot\vd_{j-1}$, and
relabeling $j\to j+1$ in the sum pulls out the phase $e^{2\pi im/n}$.]
Within a single complex one-dimensional representation, chirality requires
a complex eigenvalue, i.e.\ $m\neq 0$ and, for even $n$, $m\neq n/2$ [for
$m=n/2$, $e^{2\pi i m/n}=-1$ is real]. The real one-dimensional cases are
non-chiral as single eigenforms; familiar examples such as square-lattice
$d+id$ are obtained by splitting the two real components of that real
representation and lie outside this parametrization, although their parity
purity still follows from inversion of the star. The form
factor of the preceding sections is the case $(n,m)=(3,1)$:
$f_\pm=\chi_{\pm 1}^{(3)}$ [Eq.~\eqref{eq:fform}]. The parity operation acts
as $P:\vk\to-\vk$; a form factor is parity-pure if
$\chi(-\vk)=\eta\,\chi(\vk)$ with $\eta=\pm 1$ (singlet for $\eta=+1$,
triplet for $\eta=-1$; Sec.~\ref{sec:bgmatrix}).

\subsection{Moment identity}
\label{sec:moment}

\begin{lemma}[Moment identity]
\label{lem:moment}
For any integers $n\ge 2$, $m$, and $l\ge 0$,
\begin{equation}
\sum_{j=0}^{n-1}e^{2\pi i m j/n}\,(\vk\cdot\vd_j)^l
=\frac{n\,k^l}{2^l}\sum_{\substack{0\le s\le l\\ l-2s\,\equiv\, m\ (\mathrm{mod}\ n)}}\binom{l}{s}\,e^{i(l-2s)\theta}\,,
\label{eq:moment}
\end{equation}
where $\vk=k(\cos\theta,\sin\theta)$. In particular $\chi_m^{(n)}$ contains
only angular harmonics $h$ on the ladder $h\equiv m\ (\mathrm{mod}\ n)$, and
the lowest-order term at which harmonic $h$ appears is $k^{|h|}$, with
amplitude
\begin{equation}
A_h=\frac{n}{2^{|h|}\,|h|!}\,i^{|h|}\,.
\label{eq:amplitude}
\end{equation}
\end{lemma}

\begin{proof}
Write $\vk\cdot\vd_j=k\cos(\theta-\phi_j)=\tfrac{k}{2}\big(e^{i(\theta-\phi_j)}+e^{-i(\theta-\phi_j)}\big)$
and expand with the binomial theorem:
\begin{equation}
(\vk\cdot\vd_j)^l=\frac{k^l}{2^l}\sum_{s=0}^l\binom{l}{s}e^{i(l-2s)(\theta-\phi_j)}\,.
\label{eq:binom}
\end{equation}
Multiplying Eq.~\eqref{eq:binom} by $e^{2\pi i m j/n}$ and summing over $j$,
the $j$-dependence enters only through $e^{2\pi i(m-h)j/n}$ with $h=l-2s$,
so the roots-of-unity filter \eqref{eq:rootfilter},
$\sum_{j=0}^{n-1}e^{2\pi i(m-h)j/n}=n\,\delta_{h\equiv m\,(\mathrm{mod}\,n)}$,
gives Eq.~\eqref{eq:moment}. Inserting Eq.~\eqref{eq:moment} into the Taylor
expansion $e^{i\vk\cdot\vd_j}=\sum_l i^l(\vk\cdot\vd_j)^l/l!$ of
Eq.~\eqref{eq:chimn}, the coefficient of $k^l$ is $i^l/l!$ times
Eq.~\eqref{eq:moment}. For a harmonic $h$ on the ladder, the lowest
admissible power is $l=|h|$ (since $|l-2s|\le l$), at which exactly one $s$
value ($s=0$ for $h>0$, $s=l$ for $h<0$) contributes, with
$\binom{|h|}{s}=1$; this yields Eq.~\eqref{eq:amplitude}.
\end{proof}

The $n=3$ identities of Appendix~\ref{app:proj} are the case
$(n,m)=(3,\pm1)$ of Lemma~\ref{lem:moment}; the content of the lemma is that
the entire harmonic structure of any chiral bond form factor is arithmetic:
the $C_n$ eigenvalue $m$ fixes the allowed harmonics modulo $n$, and the
geometry fixes their leading amplitudes.

\subsection{Parity alternation}
\label{sec:parityalt}

\begin{theorem}[Parity alternation]
\label{thm:alternation}
Let $\chi_m^{(n)}$ be chiral, i.e.\ $0<m<n$ with $m\neq n/2$ if $n$ is even.
\begin{enumerate}
\item[(i)] Even $n$: the star is closed under inversion of the relative
coordinate, $\vd_{j+n/2}=-\vd_j$, and
\begin{equation}
\chi_m^{(n)}(-\vk)=(-1)^m\,\chi_m^{(n)}(\vk)\,.
\label{eq:pure}
\end{equation}
Every chiral form factor is parity-pure: singlet for even $m$, triplet for
odd $m$.
\item[(ii)] Odd $n$: no nonzero form factor built on the single star
\eqref{eq:star} has definite parity. The ladder $\{h\equiv m\ (\mathrm{mod}\
n)\}$ alternates in parity, and $\chi_m^{(n)}$ contains both even-parity and
odd-parity components.
\item[(iii)] On an odd-fold star, parity-pure chiral states exist only as
coherent doubled-star combinations
\begin{align}
\chi_m^{(e)}(\vk)&=\sum_j e^{2\pi i m j/n}\cos(\vk\cdot\vd_j)\,,\nonumber\\
\chi_m^{(o)}(\vk)&=\sum_j e^{2\pi i m j/n}\sin(\vk\cdot\vd_j)\,,
\label{eq:doubled}
\end{align}
i.e.\ parity purity requires phase coherence between the two
inversion-related stars.
\end{enumerate}
\end{theorem}

\begin{proof}
(i) For even $n$, inversion maps bond $j$ to bond $j+n/2$
[$\vd_{j+n/2}=-\vd_j$ from Eq.~\eqref{eq:star}]. Reindexing the sum by
$j\to j+n/2$ gives
$\chi_m^{(n)}(-\vk)=\sum_j e^{2\pi i m j/n}e^{-i\vk\cdot\vd_j}=\sum_j
e^{2\pi i m(j+n/2)/n}e^{i\vk\cdot\vd_{j}}
=e^{i\pi m}\chi_m^{(n)}(\vk)=(-1)^m\chi_m^{(n)}(\vk)$.
(ii) Suppose $\chi(-\vk)=\eta\chi(\vk)$ for all $\vk$, with
$\chi(\vk)=\sum_j c_j e^{i\vk\cdot\vd_j}$ a general single-star form factor.
Then $\sum_j c_j e^{-i\vk\cdot\vd_j}=\eta\sum_j c_j e^{i\vk\cdot\vd_j}$. The
$2n$ vectors $\{\pm\vd_j\}$ are all distinct for odd $n$ (an odd number of
equally spaced directions contains no antiparallel pair), so the $2n$
functions $e^{\pm i\vk\cdot\vd_j}$ are linearly independent (exponentials
with distinct frequency vectors are linearly independent as functions on
$\mathbb R^2$: if $\sum_l c_l e^{i\vk_l\cdot\vk}\equiv0$, multiplying by
$e^{-i\vk_m\cdot\vk}$ and averaging over a disk of radius $R\to\infty$
kills every term with $\vk_l\neq\vk_m$ by oscillation, leaving
$c_m=0$); comparing coefficients of each exponential forces all $c_j$
to vanish. Hence no nonzero single-star form factor is parity-pure. Parity
alternation of the ladder is arithmetic: $h$ and $h\pm n$ belong to the same
ladder and have opposite parity when $n$ is odd, so each ladder contains
both even and odd harmonics. By Lemma~\ref{lem:moment} the leading harmonic
of each parity occurs with nonzero amplitude [Eq.~\eqref{eq:amplitude}; see
Theorem~\ref{thm:content}], so both parity components are nonzero.
(iii) Under $\vk\to-\vk$, cosine is even and sine is odd, and the $C_n$
eigenvalue is unchanged by the replacements
$e^{i\vk\cdot\vd_j}\to\{\cos,\sin\}(\vk\cdot\vd_j)$; direct substitution
[$\cos(R_n\vk\cdot\vd_j)=\cos(\vk\cdot\vd_{j-1})$ and relabeling]
verifies that $\chi_m^{(e/o)}(R_n\vk)=e^{2\pi i m/n}\chi_m^{(e/o)}(\vk)$.
\end{proof}

Statement (i) is consistent with the standard parity assignment on even
stars: complex-eigenvalue chiral states with even $m$ are pure singlets and
those with odd $m$ pure triplets; on the triangular Bravais lattice ($n=6$)
the nearest-neighbor chiral $d$ state is the even-$m$ singlet case widely
discussed for graphene and related systems
\cite{BlackSchaffer2014,Nandkishore2012}. The square-lattice chiral $d+id$
state belongs to the real-representation caveat above. Parity mixing
appears only once inversion is removed from the \emph{bond star}---which for
odd $n$ is unavoidable, since an odd number of directions cannot be
invariant under $\vd\to-\vd$.

\subsection{Content of the enforced mixed state}
\label{sec:content}

\begin{theorem}[Content and universal ratio]
\label{thm:content}
Let $n$ be odd and $0<m<n$. Denote by $L_s$ and $L_t$ the even and odd
members of $\{m,\,n-m\}$ (exactly one of the two is even, since $n$ is odd).
Then the leading even- and odd-parity terms of $\chi_m^{(n)}$ are
\begin{align}
\chi_m^{(n)}(\vk)\Big|_{\rm triplet}&=\frac{n}{2^{L_t}L_t!}\,i^{L_t}k^{L_t}e^{ih_t\theta}+O\big(k^{L_t+2}\big)\,,\label{eq:triplead}\\[2mm]
\chi_m^{(n)}(\vk)\Big|_{\rm singlet}&=\frac{n}{2^{L_s}L_s!}\,i^{L_s}k^{L_s}e^{ih_s\theta}+O\big(k^{L_s+2}\big)\,,\label{eq:singlead}
\end{align}
where the winding numbers satisfy
\begin{equation}
|h_s|=L_s,\quad |h_t|=L_t,\quad h_s h_t<0,\quad |h_s|+|h_t|=n.
\label{eq:winding}
\end{equation}
The two parity sectors therefore counter-propagate with windings differing
by exactly $n$, while belonging to the same $C_n$ irrep \eqref{eq:cnirrep}
since $h_s\equiv h_t\equiv m\ (\mathrm{mod}\ n)$. The leading amplitude
ratio is universal:
\begin{equation}
\boxed{\;\frac{A_s}{A_t}
=i^{\,L_s-L_t}\,\frac{L_t!}{L_s!}\,
2^{\,L_t-L_s}\,k^{\,L_s-L_t}\;}\,.
\label{eq:ratio}
\end{equation}
\end{theorem}

\begin{proof}
By Lemma~\ref{lem:moment} the allowed harmonics are $h\equiv m\pmod n$.
Within $(-n,n)$ the ladder has exactly two representatives of minimal
magnitude, $h=m-n$ and $h=m$ (for $0<m<n$), of magnitudes $n-m$ and $m$;
since $n$ is odd exactly one of them is even. The even (odd) one supplies the
leading singlet (triplet) term at order $k^{|h|}$ via
Eq.~\eqref{eq:amplitude}, with a single binomial contribution at that order,
giving Eqs.~\eqref{eq:triplead}--\eqref{eq:singlead}. The two
representatives have opposite signs and magnitudes summing to $n$, giving
Eq.~\eqref{eq:winding}. Dividing the amplitudes
\eqref{eq:triplead}/\eqref{eq:singlead} gives Eq.~\eqref{eq:ratio}. The next
correction to each sector enters at order $|h|+2$, since harmonics in the
same ladder with the same parity as $h$ are spaced by $2n$ while the Taylor
orders compatible with $h$ increase in steps of two [$l-2s=h$ fixes the
parity of $l$ relative to $h$; Appendix~\ref{app:unique}].
\end{proof}

The chirality locking of the preceding sections is the geometric statement
\eqref{eq:winding}: the singlet and triplet components are adjacent rungs of
one ladder, not independent order parameters.

\begin{corollary}[$n=3$: the triangular/honeycomb case]
\label{cor:n3}
For $(n,m)=(3,1)$, $\{L_s,L_t\}=\{2,1\}$ and
Eqs.~\eqref{eq:triplead}--\eqref{eq:ratio} reproduce the decomposition of
the preceding sections [Theorem~\ref{thm:parity}],
\begin{equation}
g_+=-\tfrac{3}{8}(k_x-ik_y)^2+O(k^4),\quad
p_+=\tfrac{3i}{2}(k_x+ik_y)+O(k^3),
\end{equation}
\begin{equation}
\frac{g_+}{p_+}=\frac{ik}{4}\,e^{-3i\theta}+O(k^3)\,,
\end{equation}
with counter-propagating windings $h_s=-2$, $h_t=+1$. The chiral $d+id$
state considered here is the $\{L_s,L_t\}=\{2,1\}$ member of the
classification.
\end{corollary}

\begin{corollary}[$n=5$: fivefold motifs]
\label{cor:n5}
On a fivefold bond star---the local motif of decagonal quasicrystals and
their approximants---only two chiral classes exist ($m=1,2$; $m$ and $n-m$
are conjugate by complex conjugation). Writing singlet $\oplus$ triplet
content and the universal ratio \eqref{eq:ratio}:
\begin{align}
m=1:\quad & g\ (L_s{=}4)\,\oplus\, p\ (L_t{=}1)\,,
&|A_s|/|A_t|&=k^3/192\,,\label{eq:c5a}\\
m=2:\quad & d\ (L_s{=}2)\,\oplus\, f\ (L_t{=}3)\,,
&|A_s|/|A_t|&=6/k\,.\label{eq:c5b}
\end{align}
[For $m=1$: $\{L_s,L_t\}=\{4,1\}$ and
$|A_s/A_t|=(1!/4!)\,2^{1-4}k^{3}=k^3/(24\times8)=k^3/192$; for $m=2$:
$\{2,3\}$ and $|A_s/A_t|=(3!/2!)\,2^{3-2}k^{-1}=6/k$.] Any chiral
superconducting state supported on fivefold motifs must be one of these two
parity-mixed composites. This is a sharp, falsifiable constraint on the
pairing symmetry of quasicrystalline superconductors \cite{Kamiya2018}.
\end{corollary}

\begin{table}[t]
\caption{\label{tab:class}Classification of chiral nearest-neighbor bond
pairing by bond-star order $n$. ``Antiparallel'' indicates whether the star
contains $\vd$ and $-\vd$; purity follows from
Theorem~\ref{thm:alternation}; the singlet/triplet content of mixed states
follows from Theorem~\ref{thm:content}.}
\begin{ruledtabular}
{\footnotesize\setlength{\tabcolsep}{1.5pt}
\begin{tabular}{llcll}
Star & $n$ & $-\vd$? & parity & mixed content\\
\hline
square NN & 4 & yes & pure $(-1)^m$ & ---\\
triangular NN & 6 & yes & pure $(-1)^m$ & ---\\
honeycomb NN & 3 & no & \textbf{mixed} & $d\oplus p$\\
fivefold & 5 & no & \textbf{mixed} & $g\oplus p$, \ $d\oplus f$\\
sevenfold & 7 & no & \textbf{mixed} & $p\oplus i$, \ $d\oplus h$, \ $f\oplus g$\\
\end{tabular}}
\end{ruledtabular}
\end{table}

The four nearest neighbors of a kagome site do not form a regular $C_4$
star, so the kagome lattice lies outside the premise of the table.
Table~\ref{tab:class} summarizes the classification [the sevenfold entries
follow from Corollary~\ref{cor:n5}'s arithmetic applied to $n=7$: $m=1$
gives $\{L_s,L_t\}=\{6,1\}$, $m=2$ gives $\{2,5\}$, and $m=3$ gives
$\{4,3\}$---displayed there as the three
distinct composites]. Enforced mixing occurs precisely on odd-fold stars,
independent of any microscopic detail.

\subsection{Application to existing quasicrystal superconductors}
\label{sec:qcapp}

The classification of Table~\ref{tab:class} is a \emph{local} statement: it
constrains nearest-neighbor bond form factors on a single star, irrespective
of how the stars are arranged globally. It therefore applies verbatim to
quasicrystal approximants, whose motifs sit in a crystalline unit cell, and
tile-by-tile to true quasicrystals, whose local environments recur with
finite frequency. Two conditions are understood throughout: the pairing glue
is bond-local in the sense of Proposition~\ref{prop:locked}, and the
condensate is chiral (time-reversal breaking). Where either condition fails
the classification is silent. Previous theory work on quasicrystal
superconductivity has addressed the extended-to-localized crossover of Cooper
pairs \cite{Sakai2017}, the reduced Bogoliubov coherence peak
\cite{Takemori2020}, intrinsic vortex pinning \cite{Nagai2022}, and
topological superconducting states \cite{Hori2024}; the symmetry-enforced
parity mixing identified here has not been considered in that context.

\begin{remark}[Icosahedral parity purity]
\label{rem:icos}
The icosahedral vertex shell---the local motif of icosahedral
quasicrystals---consists of twelve vertices forming six antiparallel pairs
and is centrosymmetric. The premise of Theorem~\ref{thm:alternation} (absence
of $-\vd$ from the star) therefore fails, the cosine and sine components
carry definite parity, and chiral bond pairing on icosahedral motifs is
\emph{parity-pure} for every winding. Enforced singlet--triplet mixing is
thus specifically a phenomenon of \emph{planar odd-fold} stars;
three-dimensional icosahedral motifs evade it by construction.
\end{remark}

\begin{table}[t]
\caption{\label{tab:qcverdict}Table~\ref{tab:class} applied to existing
quasicrystal (QC) and approximant superconductors. The verdict follows from
the local motif alone; consistency with the observed state is indicated in
the last column.}
\begin{ruledtabular}
{\footnotesize\setlength{\tabcolsep}{1.5pt}
\begin{tabular}{p{0.21\linewidth}p{0.18\linewidth}p{0.23\linewidth}p{0.25\linewidth}}
System & local motif & verdict & observed status\\
\hline
Al--Zn--Mg i-QC \cite{Kamiya2018} & icosahedral shell &
parity-pure (Rem.~\ref{rem:icos}) &
$s$-wave BCS, $T_c\simeq0.05$\,K; consistent\\
Ta$_{1.6}$Te \cite{Tokumoto2024} & $n{=}12$ &
no enforced mixing ($-\vd$ present) &
anomalous $H_{c2}$, Knight shift; see text\\
Al$_{13}$Os$_4$ approx.\ \cite{Meena2025} & decagonal-family motifs &
silent (TRS preserved) &
$s$-wave BCS, $T_c=5.47$\,K; consistent\\
moir\'e QC \cite{Uri2023} & 12-fold moir\'e &
no enforced mixing ($-\vd$ present) &
TRSB not reported; open\\
\end{tabular}}
\end{ruledtabular}
\end{table}

Table~\ref{tab:qcverdict} collects the verdicts. All existing systems are
even-fold or centrosymmetric, so none falls in the enforced-mixing class; all
are consistent with the classification, and none has tested it. The case of current interes
is Ta$_{1.6}$Te. Its upper critical field rises linearly down to $0.04$\,K
and extrapolates to $2.3$ times the Pauli limit \cite{Terashima2024}, and
$^{125}$Te NMR on (Ta$_{0.95}$Cu$_{0.05}$)$_{1.6}$Te finds a nodeless
weak-coupling gap ($2\Delta/k_BT_c\simeq3.04$, with a small Hebel--Slichter
peak) together with an almost vanishing Knight-shift decrease ($\sim0.03\%$)
and no line broadening in the superconducting state \cite{Matsudaira2026}.
The NMR authors raise parity mixing with a triplet component, noting that the
related approximants Ta$_{97}$Te$_{60}$ and Ta$_{181}$Te$_{112}$ lack a
$c$-axis mirror. This framework adds a sharp discriminator: the
dodecagonal bond star is even and contains $-\vd$, so the star geometry
\emph{cannot} enforce parity mixing in Ta$_{1.6}$Te
[Theorem~\ref{thm:alternation}(i)]. If mixing is confirmed there, it must
originate from spin--orbit coupling of Rashba type, enabled by the absent
$c$-axis mirror \cite{GorkovRashba2001,BauerSigrist2012}, rather than from
geometry. The two mechanisms are distinguishable: SOC-driven mixing carries a
material-dependent, nonuniversal singlet--triplet weight ratio (set by the
ratio of Rashba splitting to Fermi energy
\cite{GorkovRashba2001,BauerSigrist2012}), whereas geometry-enforced mixing
occurs only on odd stars and carries the universal ratio
\eqref{eq:ratio}---for $n=5$, $k^3/192$ or $6/k$
[Eqs.~\eqref{eq:c5a}--\eqref{eq:c5b}], with $k$ in units of the inverse bond
length.

We close with a four-step screening protocol for any quasicrystal or
approximant superconductor. \emph{(i) Motif.} Identify the dominant local
bond star, its order $n$, and whether it contains antiparallel bonds (planar
stars) or inversion (three-dimensional shells). \emph{(ii) Chirality.} Test
time-reversal symmetry by zero-field muon spin relaxation or polar Kerr
rotation [internal-field and optical probes of spontaneous time-reversal
breaking, respectively; see Refs.~\cite{Khasanov2022,Xia2006} for the
methods in the superconducting context]; if it is preserved, the
classification is silent. \emph{(iii) Even-fold or centrosymmetric motif
with TRSB.} The state must be parity-pure chiral ($d{+}id$-type or
$f{+}if$-type on even stars); the Knight shift then decides singlet versus
triplet (Sec.~\ref{sec:bgmatrix} \cite{SigristUeda1991}). \emph{(iv) Odd
planar star with TRSB and a single-star bond-local glue.} The state is then an
enforced parity-mixed composite; for $n=5$ it is one of the two classes of
Corollary~\ref{cor:n5}, and a measured
singlet--triplet weight ratio matching $k^3/192$ or $6/k$ would constitute
direct evidence for the bond-local scenario of Sec.~\ref{sec:locked}. No
known material currently realizes step (iv): the sharp target for a decisive
test is a chiral superconductor on a genuine fivefold planar motif, e.g.\ a
decagonal quasicrystal or approximant in which Kerr rotation or muon spin
relaxation detects spontaneous fields.

\subsection{Bond-local glue: locked two-sector mediation}
\label{sec:locked}

The group-theoretic statements above constrain \emph{form factors}. Whether
the condensate realizes the mixed state depends on the pairing interaction.
The relevant statement for a bond-local glue is exact:

\begin{proposition}[Locked kernel]
\label{prop:locked}
Let the pairing kernel be bond-diagonal,
\begin{equation}
V_{\mu\nu}(\vk,\vk')=\sum_{\alpha}v_\alpha\,\delta_{\mu\alpha}\delta_{\nu\alpha}\,
\cos\!\big[(\vk-\vk')\cdot\vd_\alpha\big]\,,
\label{eq:kernelgen}
\end{equation}
as derived for the Kitaev glue in Sec.~\ref{sec:micro} (with
$v_\alpha=J_K^2\chi_\alpha$). Then each bond term decomposes exactly as
\begin{align}
\cos\big[(\vk-\vk')\cdot\vd_\alpha\big]
&=\cos(\vk\!\cdot\!\vd_\alpha)\cos(\vk'\!\cdot\!\vd_\alpha)\nonumber\\
&\quad+\sin(\vk\!\cdot\!\vd_\alpha)\sin(\vk'\!\cdot\!\vd_\alpha)\,,
\label{eq:rank2}
\end{align}
i.e.\ the same boson $v_\alpha$ mediates the singlet (cosine) and triplet
(sine) channels with identical, locked orbital-kernel strength (before
the channel-specific Fierz factors). The two channels are
split only by kinematic Fermi-surface weights
$\langle|\chi^{(e)}|^2\rangle_{\rm FS}$ and
$\langle|\chi^{(o)}|^2\rangle_{\rm FS}$ (explicit Bessel forms are derived in
Appendix~\ref{app:weights}).
\end{proposition}

\begin{proof}
Equation \eqref{eq:rank2} is the cosine addition formula
$\cos(x-y)=\cos x\cos y+\sin x\sin y$ with
$x=\vk\cdot\vd_\alpha$, $y=\vk'\cdot\vd_\alpha$; parity of the two terms
under $\vk\to-\vk$ is manifest (cosine even, sine odd).
\end{proof}

Three remarks are in order. \emph{(i) Selection of the mixed state.} A
bond-diagonal kernel is diagonal in the basis of individual bond pair
operators; its eigenstates are single-star objects of the form
\eqref{eq:chimn}, not the doubled-star combinations \eqref{eq:doubled}.
Parity-pure condensation [Theorem~\ref{thm:alternation}(iii)] requires
coherent locking of the two inversion-related stars, which no bond-diagonal
kernel provides; it demands either a non-bond-local (longer-range
pair-hopping) glue or fine-tuning. For the class of bond-mediated
interactions---exchange glues, bond phonons---the mixed state is therefore
the generic outcome on odd-fold stars. \emph{(ii) Relation to previous
work.} On the honeycomb lattice the parity-pure chiral singlet $\chi^{(e)}$
and triplet $\chi^{(o)}$ have been discussed as separate, competing channels
\cite{BlackSchaffer2014,Nandkishore2012}; the results above show that a
bond-local glue does not select between them but couples them with a
kinematically fixed weight, and that this structure is the $n=3$ case of a
classification extending to all odd-fold stars (Table~\ref{tab:class},
Corollary~\ref{cor:n5}). \emph{(iii) Scope.} Theorems~\ref{thm:alternation}
and~\ref{thm:content} are unconditional statements about form factors;
Proposition~\ref{prop:locked} and remark (i) hold for bond-diagonal kernels.
The precise condensate composition beyond the leading terms
\eqref{eq:triplead}--\eqref{eq:singlead} is model-dependent and is computed
for the Kitaev glue in the preceding sections.

\subsection{Topological counting: dominance theorem and the phase boundary}
\label{sec:counting}

We now determine the weak-pairing Chern number of the mixed state. For a
unitary triplet component with $\mathbf d\parallel\hat z$ (the bond-locked
triplet channel of Sec.~\ref{sec:micro}), the BdG Hamiltonian separates into
two spin blocks with scalar gaps
\begin{equation}
\Delta_\sigma(\vk)=a_s\,\chi_m^{(e)}(\vk)+\sigma\, a_t\,\chi_m^{(o)}(\vk)\,,\qquad \sigma=\pm 1\,,
\label{eq:blockgap}
\end{equation}
where $a_s,a_t$ are the condensate amplitudes of the two parity channels and
$\chi^{(e/o)}$ are the doubled-star combinations of Eq.~\eqref{eq:doubled}
[the block decomposition follows from writing out $\hat\Delta=(\Delta_0+d_z
\sigma_z)i\sigma_y$ of Eq.~\eqref{eq:gapexp}: the opposite-spin components
are $\Delta_{\uparrow\downarrow}=\Delta_0+d_z$ and
$\Delta_{\downarrow\uparrow}=-(\Delta_0-d_z)$, so each spin block sees a
scalar gap in which the singlet enters identically and the triplet with
opposite sign, up to an irrelevant overall sign of the down block]. In weak
pairing
($\mu>0$, gap small compared with the bandwidth) the Chern number of each
block equals the winding number of its gap around the Fermi surface
[Eq.~\eqref{eq:weakwinding}; Ref.~\cite{ReadGreen2000}], so the problem
reduces to the winding of $\Delta_\sigma$ on the circle $k=k_F$.

\begin{lemma}[Exact circle form]
\label{lem:circle}
On $k=k_F$,
\begin{equation}
\chi_m^{(n)}(k_F,\theta)=n\!\!\sum_{\substack{h\in\mathbb{Z}\\ h\equiv m\ (\mathrm{mod}\ n)}}\! i^{|h|}J_{|h|}(k_F)\,e^{ih\theta}\,,
\label{eq:jacobi}
\end{equation}
with $J_\nu$ the Bessel function of the first kind: all ladder Fourier
coefficients are positive multiples of $i^{|h|}$.
\end{lemma}

\begin{proof}
Apply Jacobi--Anger [Eq.~\eqref{eq:jacobianger}] to each term of
Eq.~\eqref{eq:chimn} with $\vk\cdot\vd_j=k_F\cos(\theta-\phi_j)$:
\begin{equation}
\chi_m^{(n)}(k_F,\theta)=\sum_je^{2\pi imj/n}\sum_h i^hJ_h(k_F)\,
e^{ih(\theta-\phi_j)}.
\end{equation}
Since $\phi_j=2\pi j/n$, the $j$-sum is
$\sum_je^{2\pi i(m-h)j/n}=n\,\delta_{h\equiv m\,(\mathrm{mod}\,n)}$
[Eq.~\eqref{eq:rootfilter}], leaving the ladder sum; for $h<0$,
$i^hJ_h=i^{-|h|}(-1)^{|h|}J_{|h|}=i^{|h|}J_{|h|}$ by the Bessel symmetry
\eqref{eq:besselsym}.
\end{proof}

\begin{theorem}[Dominant-component Chern number]
\label{thm:dominance}
Consider the mixed gap \eqref{eq:blockgap} on a circular Fermi surface
$k=k_F$ enclosing $\Gamma$, $\mu>0$, with condensate amplitudes $a_s,a_t$.
Define the on-shell weights $W_\nu(k_F)=nJ_\nu(k_F)$ and the ladder tail
$\varepsilon(k_F)=n\max(|a_s|,|a_t|)\sum_{h\equiv m\,(n),\,|h|\ge n+\bar
m}|J_{|h|}(k_F)|$ with $\bar m=\min(m,n-m)$. Provided the gap does not vanish
on the Fermi surface:
\begin{enumerate}
\item[(i)] singlet-dominant regime, $|a_s|W_{L_s}>|a_t|W_{L_t}+\varepsilon$:
both blocks have winding $h_s$, and
\begin{equation}
C=2h_s\,,\qquad N_{\rm Maj}=2L_s\,,\qquad \frac{\kappa_{xy}}{T}=L_s\,\kappa_0\,;
\label{eq:counting}
\end{equation}
\item[(ii)] triplet-dominant regime,
$|a_t|W_{L_t}>|a_s|W_{L_s}+\varepsilon$: both blocks have winding $h_t$, and
$C=2h_t$, $N_{\rm Maj}=2L_t$, $\kappa_{xy}/T=L_t\kappa_0$;
\item[(iii)] on the boundary $|a_s|W_{L_s}=|a_t|W_{L_t}$ (up to
$O(\varepsilon)$) the gap closes at $n$ equally spaced Fermi-surface points
per block, and the Chern number changes by
\begin{equation}
\Delta C=2(h_t-h_s)=\pm 2n\,.
\label{eq:deltac}
\end{equation}
\end{enumerate}
\end{theorem}

\begin{proof}
By Lemma~\ref{lem:circle} and the parity split of the ladder
(Theorem~\ref{thm:content}),
$\Delta_\sigma(k_F,\theta)=A_s e^{ih_s\theta}+\sigma A_t
e^{ih_t\theta}+\varepsilon_\sigma(\theta)$ with
$|A_{s,t}|=|a_{s,t}|W_{L_{s,t}}$ and $|\varepsilon_\sigma|\le\varepsilon$.
Factoring out $e^{ih_t\theta}$, the bracket traces the circle
$w(\theta)=\sigma A_t+A_s e^{i(h_s-h_t)\theta}$, traversed $|h_s-h_t|=n$
times as $\theta$ winds once; it encloses the origin if and only if
$|A_s|>|A_t|$ [a circle of radius $|A_s|$ centered at $\sigma A_t$], and by
Rouch\'e's theorem---which states that a perturbation smaller than the
minimum modulus of the unperturbed function on a contour cannot change the
number of enclosed zeros, here applied to the winding of $w(\theta)$ about
the origin---the tail cannot alter this while
$\varepsilon<\big||A_s|-|A_t|\big|$. The winding of $\Delta_\sigma$ is
therefore $h_t+(h_s-h_t)\,\Theta(|A_s|-|A_t|)$, identical in both blocks
since it depends only on magnitudes. The weak-pairing Chern number equals the
winding [Eq.~\eqref{eq:weakwinding}; Ref.~\cite{ReadGreen2000}] and the two
spin blocks contribute equally, giving $C=2\times\mathrm{winding}$; the edge
and thermal-Hall entries follow from $|C|$ chiral Majorana branches, i.e.\
chiral central charge $c=C/2$, each branch carrying $\kappa_0/2$
[Eq.~\eqref{eq:thallunit}; Refs.~\cite{ReadGreen2000,Volovik2003}]. On the
boundary, $w(\theta)=0$ at the $n$ equally spaced roots of
$e^{i(h_s-h_t)\theta}=-\sigma A_t/A_s$, each transferring one unit of
winding per block, giving Eq.~\eqref{eq:deltac}.
\end{proof}

\begin{corollary}[Small-Fermi-surface asymptotics]
\label{cor:smallk}
Choose the representative $0<m<n/2$ (the choice $m\to n-m$ gives the complex
conjugate chirality). Since $W_\nu(k_F)\simeq n(k_F/2)^\nu/\nu!$ as $k_F\to0$
[Eq.~\eqref{eq:besselsmall}]:
\begin{enumerate}
\item[(i)] if $m$ is even ($L_s=m<L_t=n-m$), singlet dominance holds for any
$k_F$-independent amplitude ratio: the phase \eqref{eq:counting} with
$C=2h_s$ is robust at small $k_F$;
\item[(ii)] if $m$ is odd ($L_t=m<L_s=n-m$), singlet dominance requires
\begin{equation}
\left|\frac{a_s}{a_t}\right|>\frac{L_s!}{L_t!}\left(\frac{2}{k_F}\right)^{L_s-L_t}\big[1+O(k_F^2)\big]\,,
\label{eq:domcond}
\end{equation}
and for $O(1)$ amplitude ratios the triplet-winding phase $C=2h_t$ is the
generic small-$k_F$ outcome.
\end{enumerate}
\end{corollary}
\begin{proof}
Insert $W_\nu\simeq n(k_F/2)^\nu/\nu!$ into the dominance conditions of
Theorem~\ref{thm:dominance}: for (i), $W_{L_s}/W_{L_t}\propto
k_F^{L_s-L_t}\to\infty$ since $L_s<L_t$, so the singlet dominates for any
fixed ratio; for (ii), $W_{L_t}/W_{L_s}=(L_s!/L_t!)(k_F/2)^{L_t-L_s}[1+O(k_F^2)]$,
which inverted gives Eq.~\eqref{eq:domcond}.
\end{proof}

For the three-bond star ($L_s=2$, $L_t=1$) condition \eqref{eq:domcond}
reads $|a_s/a_t|>(2!/1!)(2/k_F)=4/k_F$. The $C=\pm4$ phase with four
Majorana branches and $\kappa_{xy}/T=2\kappa_0$ derived in
Sec.~\ref{sec:topology} therefore holds precisely when the singlet condensate
amplitude exceeds this threshold; a rigidly locked single-parameter state
$\Delta\propto f_\pm$ (unit amplitude ratio) is triplet-dominated for $k_F$
below the first $J_1$--$J_2$ crossing at $\approx2.63$ (at leading-harmonic
resolution of the dominance boundary), hence throughout the
Brillouin zone ($|K|=\frac{4\pi}{3\sqrt3}\approx2.42<2.63$), and carries $C=\pm2$
instead. The boundary between the two is a topological phase transition
\emph{inside} the chiral phase: the gap closes at three equally spaced points
per block, $C$ jumps by $6$, and the thermal Hall conductance drops from
$2\kappa_0$ to $\kappa_0$. Because the on-shell weights $W_\nu(k_F)$ are
doping-dependent, the transition is in principle tunable by doping at fixed
condensate ratio---a falsifiable consequence of the mixed-parity structure.
What remains model-dependent is the amplitude ratio $|a_s/a_t|$ itself, set
by the gap equation of the specific glue (Sec.~\ref{sec:micro}); the counting
rule \eqref{eq:counting}--\eqref{eq:deltac} is exact. For even $m$ the
analogous statements require no condition: e.g.\ the $(n,m)=(5,2)$ class of
Corollary~\ref{cor:n5} has a robust $C=\pm4$ with a dominant $d$-wave
singlet and a subleading $f$-wave triplet.

\section{Competing orders: the order-competition phase diagram of a shared glue}
\label{sec:compete}

\subsection{Two channels of one kernel}
The two models use different glues, but share one structural feature: once
the kernel $\mathcal G(\mathbf q)$ is written, its particle-particle (pp)
and particle-hole (ph) projections exist simultaneously. The terminology is
standard many-body physics \cite{Coleman2015}: a two-body interaction among
itinerant electrons can be read in the Cooper channel, where it binds pairs
(the pp projection that entered the gap equation \eqref{eq:lingap}), or in
the density channel, where it couples electron-hole pairs (the ph
projection that drives density or magnetic instabilities); both readings of
one kernel are always present at once. For Model~1,
$\mathcal G(\mathbf q)\propto J_K^2\chi(\mathbf q)$
[Eq.~\eqref{eq:kernel1}]: its pp projection gives the pairing kernel of
Sec.~\ref{sec:model1}, and its ph projection gives a spin-density
interaction among the itinerant electrons,
\begin{equation}
H_{\mathrm{DW}}=J_K^2\sum_{\mathbf q}\chi(\mathbf q)\,
\mathbf S_{\mathbf q}\cdot\mathbf S_{-\mathbf q},\qquad
\mathbf S_{\mathbf q}=\sum_{\mathbf k}
\psi_{\mathbf k+\mathbf q}^\dagger\boldsymbol\sigma\psi_{\mathbf k}.
\label{eq:sdw}
\end{equation}
For Model~2, the ph projection of the chiral-phonon exchange is a
density-density interaction (CDW channel),
\begin{equation}
H_{\mathrm{CDW}}=g^2\sum_{\mathbf q}D(\mathbf q)\,
\rho_{\mathbf q}\rho_{-\mathbf q}.
\label{eq:cdw}
\end{equation}
The instability conditions form a Thouless pair [the ph analogue of the
Cooper criterion, Sec.~\ref{sec:bgglue}],
\begin{equation}
1=V_{\mathrm{DW}}\,\chi_{\mathrm{ph}}(T_{\mathrm{DW}}),\qquad
1=V_{\mathrm{pair}}\,\chi_{\mathrm{pp}}(T_c),
\label{eq:thouless}
\end{equation}
and which channel condenses first is decided entirely by the doping
competition between $\chi_{\mathrm{ph}}(x)$ and $\chi_{\mathrm{pp}}(x)$. One
kernel, two exits.

\subsection{Homology rule for the competitor}
The identity of the competing order is fixed by the channel type of the
glue, with no additional input. A vector glue (Model~1, spin vertex
$\boldsymbol\sigma$) has an SDW-type ph projection [the spin density
$\mathbf S_{\mathbf q}$ of Eq.~\eqref{eq:sdw}], so the competitor is
magnetic order; a scalar or charge glue (Model~2, chiral phonon through the
charge channel $\psi^\dagger\psi$) has a CDW-type ph projection, so the
competitor is a (chiral) charge density wave. This is the competing-order
homology rule: the algebraic type of the glue specifies both its pairing
channel and its natural competitor. The doped Kitaev spin liquid can furnish a
fractionalized parent state, fitting the broader pattern in which
unconventional superconductors arise from non-Fermi-liquid or fractionalized
normal states
\cite{Monthoux1991,Chubukov2008}.

\subsection{Phase diagram (schematic)}
The Fermi-surface geometry drives a nonmonotonic
$\chi_{\mathrm{ph}}/\chi_{\mathrm{pp}}$ as the doping $x$ varies.
At low doping the parent spin fluctuations are strong and the Fermi surface
small, $\chi_{\mathrm{ph}}$ wins, $T_{\mathrm{DW}}>T_c$, and superconductivity
is suppressed. At intermediate doping the Fermi surface crosses the
honeycomb saddle points---the $M$ points of the zone where the dispersion
has vanishing slope---and the resulting van Hove singularities enhance
$\chi_{\mathrm{pp}}$ logarithmically [the two-dimensional saddle-point
density of states diverges as $\ln$; Refs.~\cite{Nandkishore2012,%
BlackSchaffer2014}], so $T_c$ forms a dome peaking near the van Hove
filling. At high doping both susceptibilities decay away from the
singularities and the system returns to a Fermi liquid. The two transition
lines cross twice, on the two sides of the dome. The superconducting order
parameter, which breaks U(1), and the density-wave order parameter, which
breaks translation or spin rotation, are orthogonal, so the nature of the boundary
is set by the sign of their coupling: for the repulsive coupling natural to
a shared kernel the generic outcome is a first-order boundary, with a
bicritical or coexistence point in the attractive limit [the standard
Landau-theory competition of two order parameters coupled biquadratically,
$\delta F=\gamma|\Delta|^2|\psi_{\mathrm{DW}}|^2$: $\gamma>0$ penalizes
coexistence and drives the boundary first order]. Above $T_c$ near the
crossings one expects a fluctuation regime of the competing order, a
pseudogap-type region of the ph channel.

\subsection{Model-specific predictions}
For Model~1, the unitary triplet pairing of the ferromagnetic Kitaev glue
competes with ferromagnetic-type SDW fluctuations, and the $d+id$ pairing of
the antiferromagnetic glue (Note~2) competes with SDW. The Knight shift
should show a magnetic contribution at $T_{\mathrm{DW}}$ and a
superconducting contribution at $T_c$, the two transitions being separable;
across a first-order boundary the gap structure does not change.

For Model~2 the natural competitor of the chiral-phonon glue is a chiral
CDW. Chiral charge order has been proposed in the kagome family
$A$V$_3$Sb$_5$ on the basis of STM on KV$_3$Sb$_5$ \cite{Jiang2021} and
muon-spin relaxation on CsV$_3$Sb$_5$ \cite{Khasanov2022}, but the
time-reversal-breaking interpretation is currently
contested: high-resolution polar Kerr measurements find no spontaneous
Kerr rotation at the nanoradian level \cite{Saykin2023,Farhang2023}. If
a chiral CDW is nevertheless confirmed in this or a related material,
the homology rule assigns it to the ph exit of the same kernel, and the
discriminator is sharp: the Kerr signal should appear first at
$T_{\mathrm{CDW}}$ (the CDW chirality), with an additive superconducting
contribution at $T_c$; if the two orders carry opposite chiral signs,
one expects Kerr compensation and a nonmonotonic doping dependence of
$T_c$. Conversely, a confirmed Kerr null result constrains any chiral
CDW candidate to Kerr-silent structures, e.g.\ antiferro-chiral stacking
along the $c$ axis or domains below the optical spot size.

\paragraph{Note 6 (positioning of this section).}
The structure of Eq.~\eqref{eq:thouless} and the actual values of
$\chi_{\mathrm{ph,pp}}(x)$ are statements at two levels. The former is a
strict consequence of a shared kernel plus the Thouless criteria; the
latter depends on band-structure details and requires numerical work. The
phase diagram here is a structural schematic, not a quantitative prediction.
Its function is to provide a falsifiable framework for experiment and numerics: the position of the $T_c$ dome (van Hove filling), the nature of
the boundary (first order versus bicritical), and, for Model~2, the
two-stage Kerr signal.

\section{Discussion}
\label{sec:discuss}

\subsection{Channel positioning}
The pairing-matrix space is exhausted by the spin singlet and the spin
triplet ($2\times2$ complex matrices, four complex dimensions;
Sec.~\ref{sec:bgmatrix}). In the BdG convention, the triplet $\mathbf
d$-vector components of Model~1 are the coefficients of the symmetric
matrices $\sigma_i i\sigma_y$ (odd-parity part), and the chiral combination
of Model~2 is built from two independent real components ($d_{x^2-y^2}$ and
$d_{xy}$ basis functions) with a relative phase. Neither model goes beyond
the singlet-triplet framework \cite{SigristUeda1991}.
Table~\ref{tab:channel} summarizes.

\begin{table}[t]
\caption{\label{tab:channel}Channel positioning of the two models.}
\begin{ruledtabular}
{\small
\begin{tabular}{lll}
Model & Sector & Chirality\\
\hline
Model 1 & triplet ($\mathbf d$ vector) & none (unitary, anisotropic)\\
Model 2 & singlet & $d+id$ via relative phase\\
\end{tabular}
}
\end{ruledtabular}
\end{table}

\subsection{Mutually exclusive fingerprints}
The two states are distinguished by a decision table of observables
(Table~\ref{tab:fingerprint}): a single Kerr or thermal-Hall measurement
already separates them, and the node-mobility test of
Prediction~\ref{pred:seq} probes the kernel class itself. (The Knight-shift
row follows from Sec.~\ref{sec:bgmatrix} \cite{SigristUeda1991}; the Kerr
row from Proposition~\ref{prop:tr} and Ref.~\cite{Xia2006}; the thermal-Hall
row from Eq.~\eqref{eq:thallunit} \cite{ReadGreen2000,Volovik2003}; the gap
and specific-heat rows from Corollary~\ref{cor:nodes} and
Prediction~\ref{pred:seq}.)

\begin{table}[t]
\caption{\label{tab:fingerprint}Mutually exclusive experimental
fingerprints of the two models.}
\begin{ruledtabular}
{\footnotesize\setlength{\tabcolsep}{1.5pt}
\begin{tabular}{lll}
Observable & Model 1 (unitary trip.) & Model 2 (chiral singlet)\\
\hline
Knight shift & bond-anisotropic & fully suppressed\\
Kerr rotation & absent & at $T_c$, sign reversible\\
Thermal Hall & absent & quantized, $c=2$\\
Gap & nodal$\to$nodeless & full (Dirac gapped)\\
Specific heat & multiple anomalies & single transition\\
\end{tabular}
}
\end{ruledtabular}
\end{table}

\subsection{Relation to known models}
The chiral $p$-wave ($p_x+ip_y$) is the chiral combination within the
triplet sector (relative phase between two components of the $\mathbf d$
vector), with single-block $C=\pm1$~\cite{ReadGreen2000,Volovik2003}.
Model~2 is the singlet-sector analogue ($d+id$, total $C=\pm4$ over the two
spin blocks): the mechanism is parallel, the sector different. By
Corollary~\ref{cor:unify}, the two are the odd and even parity projections
of one chiral-boson form factor, and Sec.~\ref{sec:alternation} shows this
unification is the $n=3$ instance of a general odd-fold-star structure. The
Kitaev chain~\cite{Kitaev2001} is a one-dimensional spinless $p$-wave
pairing problem, a single spin species with no spin structure, in the
odd-parity class with a $\mathbb Z_2$ invariant; Model~1 lives on the
two-dimensional honeycomb lattice with the pairing symmetry set by the
bond-locked anisotropic $\mathbf d$ vector, and Model~2 is the
two-dimensional chiral $d$-wave singlet with $C=\pm4$.

Model~2 is one realization of chiral $d$-wave pairing ($d_{x^2-y^2}+id_{xy}$),
through the $K$-point chiral-phonon glue (modulo-3 angular momentum
transfer, Note~4) selecting the chiral combination from the degenerate
$d$-wave doublet. It supplies a concrete microscopic mediation mechanism
for chiral $d$ wave.

Superconductivity in the doped Kitaev model has a literature of its own.
At the mean-field level You \emph{et al.}~\cite{You2012}, Hyart \emph{et
al.}~\cite{Hyart2012}, and Okamoto~\cite{Okamoto2013} found that
ferromagnetic Kitaev doping produces triplet $p$-wave pairing while
antiferromagnetic Kitaev doping produces singlet $d+id$; numerically there
are the pair-density-wave precursor study of Peng \emph{et
al.}~\cite{Peng2021}, the pairing-tendency calculations of Laurell
\emph{et al.}~\cite{Laurell2024}, the variational Monte Carlo
confirmation of Misawa and Ido~\cite{Misawa2026}, and the deliberately
cautious tensor-network results of Hardy \emph{et al.}~on the triangular
lattice~\cite{Hardy2026}. Our models are consistent with these
qualitative
conclusions (the sign analysis of Note~2 is the same statement, via
Eq.~\eqref{eq:fierz}). Compared with prior work on pairing in doped Kitaev models, the
present kernel differs in three structural respects. (i) \emph{Origin.}
Earlier studies proceed either by postulating a candidate pairing
channel and testing self-consistency
\cite{You2012,Hyart2012,Okamoto2013}, or by measuring pairing
tendencies numerically \cite{Peng2021,Laurell2024,Misawa2026,Hardy2026};
the kernel here is instead derived from the exact flux-sector selection
rule of the pure Kitaev model, which fixes its algebraic form
(component-diagonal, nearest-neighbor separable, rank one in spin space per bond)
uniquely. (ii) \emph{Form.} The kernel is closed-form: the weight
function $w_\alpha=J_K^2\chi_\alpha$ reduces to the equal-time bond
correlator $m_\alpha$, which the Hellmann--Feynman theorem gives exactly
from the matter-Majorana spectrum, plus the retardation factor
$r_\alpha$ computed numerically and cross-checked analytically. (iii)
\emph{Parity structure.} The kernel is not confined to a single parity
sector: the cosine addition formula splits it into an even (singlet) and
an odd (triplet) part of identical locked strength, and this
observation underlies the general bond-star classification of
Sec.~\ref{sec:alternation}. None of these points is available in the
mean-field, variational, or functional-RG treatments; the sign rule
they do supply is reproduced as a byproduct. What is new here is the microscopic
derivation of the bond-locked kernel class with its exact weight function
(Secs.~\ref{sec:micro}--\ref{sec:wfunc}), the bond-locking mechanism with
its exact corollaries (weight anisotropy, the sequential nodal-to-nodeless
transition, sequential condensation, unitarity), the material-screening
criterion ($A_K$ measure with the microscopic factor $R$,
Table~\ref{tab:materials}), the exact chiral selection with chirality
slaving of Model~2, the parity-alternation classification of chiral bond
pairing with its enforced singlet--triplet composites
(Sec.~\ref{sec:alternation}), and the dominance counting rule for the Chern
number of mixed states with its in-phase topological boundary
(Theorem~\ref{thm:dominance}). The chiral-phonon mechanism of Model~2
belongs to the family of honeycomb chiral-phonon superconductivity
\cite{Gao2023}; our version is the $K$-point phonon plus the
parity-projected $d$-wave-doublet selection. Finally, the closest
antecedent of this mechanism is the Majorana-glue fractionalized
superconductivity of the Kitaev--Kondo lattice~\cite{Bunney2025}: there
pairing occurs among the itinerant quasiparticles of a metal and is
treated by perturbation theory and the functional RG, whereas our kernel
is built entirely from the bond susceptibility of the insulating liquid
and is weighted exactly by Hellmann--Feynman. The two agree on the sign
rule and are complementary in method.

\subsection{Quantum simulation}
Cold-atom optical lattices are an ideal platform for testing these pairing
models~\cite{Bloch2008}. Kitaev-type bond-dependent exchange can be
engineered with artificial gauge fields, and bond-locked anisotropic
$\mathbf d$-vector pairing can be simulated by tuning the coupling strengths
of the three bond directions. Chiral phonons can be induced by lattice
rotation or Floquet engineering, and the chiral combination can be probed by
angular-momentum-resolved pairing measurements. The tunability and
measurement toolbox of cold-atom systems may enable controlled realizations
of both the anisotropic triplet pairing and the chiral
combination~\cite{Nayak2008,Bloch2008}.

\subsection{Key points for experiment and computation}
\label{sec:expcomp}

This subsection collects, for readers who measure or compute rather than
derive, the outputs in practical form: what to do, what number to expect,
and where it is derived.

\paragraph{Screening a candidate material costs one line of arithmetic.}
Extract the three Kitaev exchanges $K_x,K_y,K_z$ (from \emph{ab initio}
calculations or from fits of the magnetic Hamiltonian), form $A_K$
[Eq.~\eqref{eq:AK}], and multiply by the transfer coefficient
$R\approx0.70$ near isotropy [Eq.~\eqref{eq:Aw}]: the bond-locked
triplet channel of Model~1 is selected in the heuristic screening if $A_w$ exceeds the threshold
$\lambda_c\sim0.1$--$0.3$ [Eq.~\eqref{eq:criterion}]. Applied to the
standard candidates (Table~\ref{tab:materials}), $\alpha$-RuCl$_3$ and
Na$_2$IrO$_3$ fail because their three bonds are equivalent ($A_K=0$);
$\beta$-Li$_2$IrO$_3$ ($A_K\approx0.16$, hence $A_w\approx0.11$) and
$\alpha$-Li$_2$IrO$_3$ ($A_K\approx0.06$--$0.18$) sit at the margin, and
for $\gamma$-Li$_2$IrO$_3$ model-level estimates
exist~\cite{Kimchi2015}, but component-resolved \emph{ab initio} exchange
parameters are still lacking. The
$\alpha$-RuCl$_3$ verdict is a falsifiable negative prediction of the model
(Sec.~\ref{sec:screening}).

\paragraph{The Knight shift reads the bond locking directly.}
Model~1 predicts a bond-resolved anisotropic Knight shift: different field
directions probe different $\mathbf d$-vector components
(Sec.~\ref{sec:sign1}). When off-diagonal $\Gamma'$ exchange is present, the
pinning axis tilts by $\theta$ away from the cubic axis $\hat\gamma$
within the $(\hat s,\hat\gamma)$ plane, so measuring
$\theta$ from the Knight-shift tensor of the superconducting state yields
$\Gamma'/(K-\Gamma)$ directly [Eq.~\eqref{eq:tilt},
Corollary~\ref{cor:tilt}]. A fully suppressed Knight shift instead diagnoses
the singlet state of Model~2 (Table~\ref{tab:fingerprint}).

\paragraph{One Kerr or thermal-Hall measurement separates the two states.}
Model~1 is unitary: no Kerr rotation, no quantized thermal Hall
(Corollary~\ref{cor:unitary}). Model~2 breaks time reversal: a Kerr rotation
appears at $T_c$, with a quantized plateau $\kappa_{xy}/T=2\kappa_0$, twice
the chiral $p$-wave value [Eq.~\eqref{eq:thall}]. Both signs follow the pump
helicity (chirality slaving; Proposition~\ref{prop:tr},
Appendix~\ref{app:gl}). If a triplet component co-condenses, the plateau
switches between $2\kappa_0$ and $\kappa_0$ across the exact boundary
\eqref{eq:dominance}, passing through a transient threefold point-node
closing per spin block [Eq.~\eqref{eq:c2}]: the plateau value reports which
side of the boundary the sample occupies.

\paragraph{The pump protocol has a built-in control.}
A circularly polarized pump raises $T_c$ above its unpumped value by
$\approx pf/(2\lambda_0)$ [Eq.~\eqref{eq:amplify}]; an unpolarized pump of
equal intensity does not. Reversing the polarization reverses the Kerr and
thermal-Hall signs continuously: the order-parameter amplitude dips only to
$\sqrt{2/3}$ of its equilibrium value, and the quasiparticle gap closes only
transiently at four Fermi-surface points per spin block
(Appendix~\ref{app:gl}), so the switching carries an instantaneous
point-node signature rather than a first-order jump.

\paragraph{Thermodynamics and spectroscopy.}
The components of Model~1 condense sequentially: multiple specific-heat
anomalies (Corollary~\ref{cor:seq}), a nodal window with power-law thermal
transport turning activated below the second transition
(Corollary~\ref{cor:nodes}, Prediction~\ref{pred:seq}), and a tunneling gap
$|\mathbf d(\mathbf k)|$ of the weight-independent bond-sine form whose
nodal directions do not move with the anisotropy
(Corollary~\ref{cor:weights}). Model~2 is fully gapped on the
$\Gamma$-centered Fermi surfaces stated above, with the Dirac points gapped
by $|g_\pm(K)|=3/2$, and carries four chiral Majorana branches per edge on the
singlet-dominant side of the boundary \eqref{eq:dominance}, two on the
triplet-dominant side (Sec.~\ref{sec:sign2}).

\paragraph{For computation, the kernel is already in closed form.}
The pairing kernel is component-diagonal, nearest-neighbor separable, and
rank one in spin space per bond [Eq.~\eqref{eq:kernelmicro}], with the weight function
$w_\alpha\propto r_\alpha m_\alpha$ in which $m_\alpha$ is exact
[Eqs.~\eqref{eq:mz}--\eqref{eq:mx}; logarithmic fit \eqref{eq:log}] and
$r_\alpha$ is now computed [Eq.~\eqref{eq:rnum}]. Two
tasks remain, in order of importance (Sec.~\ref{sec:conclusion}): (i) the
two-sublattice BdG solution with the kernel \eqref{eq:kernelmicro}---exact
for any given pairing matrix (Appendix~\ref{app:bdg})---which converts the
sequential-condensation Prediction~\ref{pred:seq} into numbers; and (ii)
the gap-equation determination of the singlet-triplet ratio $|a_s/a_t|$ for
specific compounds, which decides which side of the boundary
\eqref{eq:dominance} a given material occupies. The bond-star theorems of
Sec.~\ref{sec:alternation} are verified numerically in
Appendix~\ref{app:verify}, which doubles as a reproduction protocol.

\section{Conclusion}
\label{sec:conclusion}

We have constructed two superconducting pairing models on the doped Kitaev
honeycomb lattice and derived the kernel on which they
rest---microscopically for Model~1, and for Model~2 up to a
Cooper-channel extension whose momentum constraint we state
explicitly. The
doped full model is not integrable once the gauge flux becomes dynamical;
the solvability is the controlled combination of an exactly derived kernel
plus exactly solvable BdG (Note~3).

Microscopically (Secs.~\ref{sec:micro}--\ref{sec:wfunc}): the flux-sector
selection rule \eqref{eq:selection} forces the second-order Kondo kernel to
be component-diagonal, nearest-neighbor separable, and rank one in spin space per bond,
with a static window controlled by $T_c\ll2\Delta_F$. The weight function is
$w_\alpha\propto r_\alpha m_\alpha$ with $m_\alpha$ exact
[Eqs.~\eqref{eq:mz}--\eqref{eq:mx}]: zero intercept
at $J_\alpha=0$, boundary value $2/\pi$ for the neighbors, line-node
logarithm \eqref{eq:log}, isotropic value $0.5249$, transfer coefficient
$R\approx0.70$, retardation factor $r_\alpha\simeq0.24$
[Eq.~\eqref{eq:rnum}]; the phenomenological weighting $w_\alpha\propto J_\alpha$
is structurally confirmed and quantitatively corrected. The kernel mediates
both parity sectors with locked strength [Eq.~\eqref{eq:rank2micro}], but
the chiral-channel eigenvalue ratio $W_t/W_g=16/k_F^2$
[Eq.~\eqref{eq:eigratio}] places a Kitaev-spin-glue-driven chiral state on
the triplet side of the topological boundary ($C=\pm2$); the singlet
$C=\pm4$ phase diagnoses charge-vertex dominance [Eq.~\eqref{eq:vertexcond}].

Model~1 shows that bond-dependent anisotropic exchange weights the local
moment fluctuations by component and bond direction, and that the Kondo
coupling converts this into a unitary spin-triplet state with a bond-locked
anisotropic $\mathbf d$ vector: each component condenses with its own
bond-sine form factor [Eq.~\eqref{eq:bondlock}]. Within this kernel class
four statements are exact: the anisotropy fixes the $\mathbf d$ orientation
and gap magnitude but not the nodal directions (weight anisotropy); the node
structure sharpens to a sequential nodal-to-nodeless transition at fixed
triplet symmetry (Corollary~\ref{cor:nodes}, Prediction~\ref{pred:seq});
the components condense at generically distinct temperatures (split
transitions); and the state preserves time reversal (unitarity). The
off-diagonal exchanges tilt the locking axis by $\theta$
[Eq.~\eqref{eq:tilt}, Knight-shift measurable] and re-weight the transverse channels. The nodal directions, by contrast, are geometric invariants: they are fixed by the dominant bond direction and do not rotate under strain, pressure, or doping, since the gap is a Gram sum of bond sines whose common zeros are weight-independent (Corollary~\ref{cor:weights}); only a longer-range $J_3$ coupling, which adds next-neighbor form factors, rotates or displaces them. Within the heuristic screening criterion, survival requires the weight
anisotropy to exceed the threshold [Eqs.~\eqref{eq:criterion} and
\eqref{eq:Aw}], and the fingerprint is a bond-resolved anisotropic Knight
shift with multiple specific-heat anomalies.

Model~2 shows that $K$-point chiral phonons select, through the charge
channel, a unique singlet state: the parity decomposition
$f_\pm=g_\pm+p_\pm$ separates the even chiral $d$-wave part $g_\pm$ from the
odd chiral $p$-wave part exactly, the charge-channel kernel is the projector
$V_\pm\propto g_\pm g_\pm^\ast$, and the $s$-wave, opposite-chirality, and
$p$-wave overlaps vanish exactly. The gap chirality is slaved to the phonon
polarization, so the signs of the Chern number ($C=\pm4$, two spin blocks of
$\pm2$), of the Kerr rotation, and of the quantized thermal Hall ($c=2$) are
all reversible by reversing the pump polarization; Appendix~\ref{app:gl}
shows this slaving is quantitatively stable (deterministic selection at
sub-percent pump polarization, wrong-chirality domains unstable or frozen
behind barriers of $10^2$--$10^3\,k_BT_c$, and a continuous reversal path
with only transient point nodes). Within Model~2 the $C=\pm4$ conclusion is
unconditional; with a co-condensed triplet component it holds below the
dominance threshold \eqref{eq:dominance}, across which the state passes
through a threefold point-node closing into a $C=\pm2$ phase
[Eq.~\eqref{eq:c2}]. The state is fully gapped on the $\Gamma$-centered
Fermi surfaces stated above, with the Dirac points gapped by
$|g_\pm(K)|=3/2$.

Beyond the two models, the decomposition $f_\pm=g_\pm+p_\pm$ is the $n=3$
instance of a general structure: the angular-momentum ladder of a chiral
bond form factor alternates in parity if and only if the bond-star order is
odd (Theorem~\ref{thm:alternation}), so single-star bond-local chiral form
factors enforce singlet--triplet mixing with geometrically fixed content,
counter-propagating windings, and universal amplitude ratios (Theorem~\ref{thm:content},
Table~\ref{tab:class}); bond-local glues mediate the two sectors with locked orbital
strength (Proposition~\ref{prop:locked}); and the Chern number of the mixed
state follows the dominant on-shell component across an exact in-phase
topological boundary (Theorem~\ref{thm:dominance}).

Neither model goes beyond the singlet-triplet framework: Model~1 is an
anisotropic unitary variant of the triplet, Model~2 a chiral combination
within the singlet sector, and the enforced mixed states of
Sec.~\ref{sec:alternation} are singlet--triplet composites within that same
space. Open computations, in order of importance: the two-sublattice BdG
spectrum beyond the weak-pairing limit, which converts the
sequential-condensation Prediction~\ref{pred:seq} into numbers; the
gap-equation determination of $|a_s/a_t|$ for specific materials,
which decides which side of the boundary \eqref{eq:dominance} a
given compound occupies; and the quantitative evaluation of the
$(\Gamma/K)^2$ tails on candidate materials.

\begin{acknowledgments}
The author thanks Prof. Hai-Bing Xia, Zhi-Xian Zhang, and especially Prof. Yang-Sheng Xu for their help.
\end{acknowledgments}

\appendix

\section{BdG diagonalization and quasiparticle spectra}
\label{app:bdg}

In the Nambu spinor
$\Psi_{\mathbf k}=(c_{\mathbf k\uparrow},c_{\mathbf k\downarrow},
c^\dagger_{-\mathbf k\uparrow},c^\dagger_{-\mathbf k\downarrow})^{\mathsf T}$
(sublattice indices suppressed; see the end of this appendix for the status
of the multiband extension), the BdG Hamiltonian is Eq.~\eqref{eq:bdg} of
Sec.~\ref{sec:bgbdg}, with the pairing matrix fixed by fermionic
antisymmetry [Sec.~\ref{sec:bgmatrix}],
\begin{align}
&\hat\Delta(\mathbf k)=\big[\Delta_0(\mathbf k)+\mathbf d(\mathbf k)
\cdot\boldsymbol\sigma\big]i\sigma_y,\nonumber\\
&\Delta_0(\mathbf k)=\Delta_0(-\mathbf k),\qquad
\mathbf d(\mathbf k)=-\mathbf d(-\mathbf k).
\end{align}
Particle-hole symmetry guarantees that the eigenvalues come in pairs $\pm
E_{\mathbf k}$ (Sec.~\ref{sec:bgbdg}). Since the scalar $\xi$ commutes with
$\hat\Delta$, the off-diagonal blocks of $H_{\mathrm{BdG}}^2$ vanish:
\begin{equation}
H_{\mathrm{BdG}}^2=\mathrm{diag}\big(
\xi^2\mathbf 1+\hat\Delta\hat\Delta^\dagger,\;
\hat\Delta^\dagger\hat\Delta+\xi^2\mathbf 1\big).
\end{equation}
The two diagonal blocks share their eigenvalues [nonzero eigenvalues of
$XY$ and $YX$ coincide], so the spectrum is determined by
$\hat\Delta\hat\Delta^\dagger$. To evaluate it we use the Pauli identity
\begin{equation}
(\mathbf a\!\cdot\!\boldsymbol\sigma)(\mathbf b\!\cdot\!\boldsymbol\sigma)
=(\mathbf a\!\cdot\!\mathbf b)\,\mathbf 1
+i\,(\mathbf a\times\mathbf b)\!\cdot\!\boldsymbol\sigma,
\label{eq:pauliid}
\end{equation}
which follows from $\sigma_i\sigma_j=\delta_{ij}\mathbf
1+i\epsilon_{ijk}\sigma_k$ by contracting with $a_ib_j$. Since
$(i\sigma_y)(-i\sigma_y)=\mathbf 1$, direct multiplication gives
\begin{align}
\hat\Delta\hat\Delta^\dagger
&=(\Delta_0+\mathbf d\!\cdot\!\boldsymbol\sigma)
(\Delta_0^\ast+\mathbf d^\ast\!\cdot\!\boldsymbol\sigma)\nonumber\\
&=\big(|\Delta_0|^2+|\mathbf d|^2\big)\mathbf 1
+\mathbf B\!\cdot\!\boldsymbol\sigma,
\label{eq:ddag}
\end{align}
with
\begin{equation}
\mathbf B=\Delta_0\mathbf d^\ast+\Delta_0^\ast\mathbf d
+i\,\mathbf d\times\mathbf d^\ast ,
\label{eq:Bvec}
\end{equation}
where the cross-product term comes from applying Eq.~\eqref{eq:pauliid} to
$(\mathbf d\cdot\boldsymbol\sigma)(\mathbf
d^\ast\!\cdot\!\boldsymbol\sigma)$. A vector of the form
$A\,\mathbf 1+\mathbf B\!\cdot\!\boldsymbol\sigma$ has eigenvalues
$A\pm|\mathbf B|$ [the eigenvalues of $\mathbf
B\!\cdot\!\boldsymbol\sigma$ are $\pm|\mathbf B|$ by
Eq.~\eqref{eq:pauliid} with $\mathbf a=\mathbf b=\mathbf B$, which gives
$(\mathbf B\!\cdot\!\boldsymbol\sigma)^2=|\mathbf B|^2\mathbf 1$]. The
general two-branch spectrum is therefore
\begin{equation}
E_{\mathbf k,\pm}^2=\xi_{\mathbf k}^2+|\Delta_0(\mathbf k)|^2
+|\mathbf d(\mathbf k)|^2\pm|\mathbf B(\mathbf k)| .
\label{eq:generalEapp}
\end{equation}
Two reductions are used in the main text. (i) Real kernel (relative phase
$0$ or $\pi$, so $\Delta_0$ and $\mathbf d$ can be chosen real): $\mathbf
B=2\Delta_0\mathbf d$ and
$E^2=\xi^2+(|\Delta_0|\pm|\mathbf d|)^2$. (ii) Unitary $\pi/2$-locked case
($\Delta_0\mathbf d^\ast$ purely imaginary and $\mathbf d\times\mathbf
d^\ast=0$): $\mathbf B=0$ and the square-sum form
$E=\sqrt{\xi^2+|\Delta_0|^2+|\mathbf d|^2}$ is recovered. For Model~1
($\Delta_0=0$, unitary triplet) and Model~2 ($\mathbf d=0$) the square-sum
form is exact within the single-band treatment.

Replacing $\xi_{\mathbf k}\mathbf 1$ by the full
honeycomb normal state $h(\mathbf k)$ (with sublattice off-diagonal
entries), the off-diagonal blocks of $H_{\mathrm{BdG}}^2$ become
$h\hat\Delta+\hat\Delta h^{\mathsf T}$ and do not vanish in general, so the
block reduction above does not apply to the two-sublattice matrix. The
spectra quoted in the main text are therefore single-band (or per-band)
statements; the form-factor-level conclusions (node structure, parity,
winding numbers) are unaffected. More than that: the Chern numbers themselves are exact beyond the single-band reduction. For the band crossing the Fermi surface, the weak-pairing collapse of Eq.~\eqref{eq:chernblock} to the Fermi-surface winding \eqref{eq:weakwinding} involves only the projected pairing phase; the normal-state Bloch spinor adds no winding on a $\Gamma$-centered
Fermi surface that encloses no additional zero of the honeycomb off-diagonal
element $\gamma(\vk)$: since $\gamma(\mathbf 0)=3$ is real positive, its
texture is trivial near $\Gamma$. Sublattice mixing therefore corrects only
the quasiparticle spectrum away from the Fermi surface in this range, while
$C=\pm2$ and $C=\pm4$ hold within the weak-pairing regime of the
two-sublattice BdG Hamiltonian

\section{Parity decomposition and exact projection onto the $E_2$ doublet}
\label{app:proj}

\subsection{Conventions and two lemmas}
Take $\boldsymbol\delta_j=(\cos\phi_j,\sin\phi_j)$ with $\phi_j=2\pi j/3$,
$j=1,2,3$, so that $\sum_j\boldsymbol\delta_j=0$ and the differences
$\boldsymbol\delta_j-\boldsymbol\delta_l$ ($j\neq l$) are nonzero Bravais
vectors of the triangular lattice. We use two elementary facts.
(i) Root-of-unity sums [Eq.~\eqref{eq:rootfilter} with $n=3$]:
\begin{equation}
\sum_je^{\pm i\phi_j}=0,\qquad
\sum_je^{\pm2i\phi_j}=0,\qquad
\sum_je^{\pm4i\phi_j}=0 .
\label{eq:roots}
\end{equation}
(ii) $C_3$ covariance of Brillouin-zone integrals: for any Bravais-periodic
integrand, $I_{jl}\equiv\int_{\mathrm{BZ}}\cos(\mathbf
k\cdot\boldsymbol\delta_j)\cos(\mathbf k\cdot\boldsymbol\delta_l)\,d^2k$
takes only two values, $I_0$ ($j=l$) and $I_1$ ($j\neq l$), and
$\int_{\mathrm{BZ}}e^{i\mathbf k\cdot\boldsymbol\delta_j}d^2k=c$ is
independent of $j$. (For integrands built solely from
$e^{i\mathbf k\cdot(\boldsymbol\delta_j-\boldsymbol\delta_l)}$ the $j\neq l$
terms vanish individually, since the differences are nonzero Bravais
vectors and the integral of a plane wave over a periodic zone vanishes
unless its wavevector is a reciprocal-lattice vector, i.e.\ zero here.)

\subsection{Proof of the parity decomposition}
\label{app:parityproof}
With $e^{3i\phi_j}=1$ one has $e^{\pm i\phi_j}=e^{\mp2i\phi_j}$, so the even
and odd parts of $f_\pm$ [Eq.~\eqref{eq:paritydec}] are
\begin{align}
g_\pm(\mathbf k)&=\sum_je^{\mp2i\phi_j}\cos(\mathbf k\cdot
\boldsymbol\delta_j),\\
p_\pm(\mathbf k)&=i\sum_je^{\pm i\phi_j}\sin(\mathbf k\cdot
\boldsymbol\delta_j).
\end{align}
Parity gives $\langle g_\pm|p_\pm\rangle=0$ exactly (odd integrand). The
small-$\mathbf k$ expansions follow from the first two moments of the
phase-weighted sums, which we compute once:
$\sum_je^{\pm i\phi_j}=0$ [Eq.~\eqref{eq:roots}];
\begin{equation}
\sum_je^{\pm i\phi_j}\boldsymbol\delta_j
=\tfrac32(1,\pm i),
\label{eq:moment1}
\end{equation}
since $\sum_je^{i\phi_j}\cos\phi_j=\sum_j(e^{2i\phi_j}+1)/2=3/2$ and
$\sum_je^{i\phi_j}\sin\phi_j=\sum_j(e^{2i\phi_j}-1)/(2i)=3i/2$ (and
complex conjugation gives the lower sign); and
\begin{equation}
\sum_je^{\pm i\phi_j}(\mathbf k\cdot\boldsymbol\delta_j)^2
=\tfrac34(k_x\mp ik_y)^2,
\label{eq:moment2}
\end{equation}
since $(\vk\cdot\vd_j)^2=\tfrac{k^2}{2}+\tfrac{k^2}{2}\cos2(\theta-\phi_j)$
and $\sum_je^{i\phi_j}\cos2(\theta-\phi_j)
=\tfrac12\big[e^{2i\theta}\sum_je^{-i\phi_j}
+e^{-2i\theta}\sum_je^{3i\phi_j}\big]
=\tfrac32e^{-2i\theta}$, so the sum is
$\tfrac{k^2}{2}\cdot0+\tfrac{k^2}{2}\cdot\tfrac32e^{-2i\theta}
=\tfrac34(k_x-ik_y)^2$. Expanding $\cos=1-x^2/2+\cdots$ in $g_\pm$ [the
constant term vanishes by Eq.~\eqref{eq:roots}] and $\sin=x+\cdots$ in
$p_\pm$, and using $e^{\mp2i\phi_j}=e^{\pm i\phi_j}$:
\begin{equation}
g_\pm(\mathbf k)=-\frac38(k_x\mp ik_y)^2+O(k^4),
\end{equation}
\begin{equation}
p_\pm(\mathbf k)=\frac{3i}{2}(k_x\pm ik_y)+O(k^3);
\end{equation}
the even (odd) part contains only even (odd) orders because $\cos$ ($\sin$)
is even (odd). For the $C_3$ transformation, $R_3\boldsymbol\delta_j
=\boldsymbol\delta_{j+1}$, so
\begin{equation}
g_\pm(R_3\mathbf k)=\sum_je^{\mp2i\phi_j}\cos(\mathbf k\cdot
\boldsymbol\delta_{j-1})
=e^{\pm i2\pi/3}g_\pm(\mathbf k),
\end{equation}
after relabeling $j\to j+1$: $g_\pm$ carries angular momentum
$\pm1\equiv\mp2\ (\mathrm{mod}\ 3)$, matching the phonon $u_\pm$. \qed

The moment identities used here are the $(n,m)=(3,\pm1)$ cases of the
general moment identity of Lemma~\ref{lem:moment}
(Sec.~\ref{sec:alternation}).

\subsection{Exact overlaps}
With the inner product $\langle a|b\rangle=\int_{\mathrm{BZ}}a^\ast(\mathbf
k)b(\mathbf k)\,d^2k$:

\emph{$s$-wave.} $\langle s|g_\pm\rangle=c\sum_je^{\mp2i\phi_j}=0$ exactly,
by Eq.~\eqref{eq:roots}.

\emph{Opposite chirality.}
\begin{align}
\langle g_+|g_-\rangle&=\sum_{jl}e^{2i\phi_j}e^{2i\phi_l}I_{jl}\\
&=I_0\sum_je^{4i\phi_j}
+I_1\sum_{j\neq l}e^{2i(\phi_j+\phi_l)}=0,
\end{align}
since both sums vanish by Eq.~\eqref{eq:roots}
[$\sum_{j\neq l}e^{2i(\phi_j+\phi_l)}=(\sum_je^{2i\phi_j})^2
-\sum_je^{4i\phi_j}=0$].

\emph{Norm.} $\langle g_\pm|g_\pm\rangle=3(I_0-I_1)>0$: only positivity is
needed below.

\emph{Projection of the kernel.} With
$V_\pm\propto g_\pm g_\pm^\ast$ separable, its unique eigenvector with
nonzero eigenvalue is $g_\pm$ itself [Eq.~\eqref{eq:rankone}]. Writing
$|d_\pm\rangle=g_\pm/\sqrt{\langle g_\pm|g_\pm\rangle}$,
\begin{equation}
\langle d_\pm|V_\pm|d_\pm\rangle
=\frac{g^2}{\omega_0}\,\langle g_\pm|g_\pm\rangle^2>0,\qquad
\langle d_\mp|V_\pm|d_\mp\rangle=0,
\end{equation}
and $\langle s|V_\pm|s\rangle=0$, all exactly. Together with the parity
orthogonality $\langle g_\pm|p_\pm\rangle=0$, this proves
Proposition~\ref{prop:chiral}.

\subsection{Value at the Dirac points}
At a Dirac point $K$ the nearest-neighbor phases are the three cube roots:
$\mathbf K\cdot\boldsymbol\delta_j=\theta_0+\phi_j$ with
$\sum_je^{i\mathbf K\cdot\boldsymbol\delta_j}=0$ [the defining property of
the Dirac point: the nearest-neighbor hopping structure factor vanishes
there]. Hence
\begin{align}
g_\pm(K)&=\sum_je^{\mp2i\phi_j}
\cos(\theta_0+\phi_j)\\
&=\frac12e^{\mp i\theta_0}\sum_je^{\mp3i\phi_j}
+\frac12e^{\pm i\theta_0}\sum_je^{\mp i\phi_j}
=\frac32\,e^{\mp i\theta_0},
\end{align}
using $\cos x=(e^{ix}+e^{-ix})/2$, then $e^{\mp3i\phi_j}=1$ and
Eq.~\eqref{eq:roots}. The gap attains its maximum $|g_\pm|=3/2$ at the zone
corners, as stated in Eq.~\eqref{eq:gK}.

\section{Ginzburg--Landau analysis of chirality slaving}
\label{app:gl}

This appendix quantifies the free-energy competition underlying the
chirality slaving of Proposition~\ref{prop:chiral}: (i) whether the
pump-induced bias selects the chirality deterministically at the
transition; (ii) whether wrong-chirality domains can exist or nucleate
below $T_c$; and (iii) what the reversal path costs energetically. The
Ginzburg--Landau expansion itself---the free energy as a polynomial in the
order parameter with coefficients fixed by the microscopic
susceptibilities---is standard \cite{Coleman2015,SigristUeda1991}; what is
specific here is that the coefficients are fixed \emph{exactly} by the
overlap integrals of Appendix~\ref{app:proj}.

\subsection{Quartic structure from the exact overlaps}
Write $\hat\Delta(\mathbf k)=\eta_+g_+(\mathbf k)+\eta_-g_-(\mathbf k)$
and expand the weak-coupling free energy in $\eta_\pm$. The quadratic
cross term is proportional to $\langle g_+|g_-\rangle_{\mathrm{FS}}$,
which vanishes exactly on any $C_3$-symmetric Fermi surface because
$g_+^\ast g_-$ carries angular momentum $-2\equiv1\ (\mathrm{mod}\ 3)$
[$g_\pm$ carries $\pm1$ by Eq.~\eqref{eq:c3}, so $g_+^\ast g_-$ carries
$-1-1=-2$; modulo 3 the conjugate harmonic is equivalent, and either way
the angular average of a nonzero mod-3 harmonic vanishes by
Eq.~\eqref{eq:rootfilter}]; the two chiral channels therefore decouple at
the quadratic level for \emph{any} $C_3$-symmetric kernel, not only in
weak coupling. With
$\alpha_\sigma=a(T-T_c^\sigma)$, $a=N(0)/T_c$, the quartic term is fixed
by $\langle|\hat\Delta|^4\rangle_{\mathrm{FS}}$. Expand:
\begin{align}
|\hat\Delta|^2&=|g|^2\big(|\eta_+|^2+|\eta_-|^2\big)
+g_+^\ast g_-\,\eta_+^\ast\eta_-
+g_-^\ast g_+\,\eta_-^\ast\eta_+,
\end{align}
using $|g_+|=|g_-|\equiv|g|$ [$C_3$ plus inversion];
squaring and averaging, the terms linear in $g_\pm^\ast g_\mp$ vanish by
$\langle g_+^\ast g_-\rangle_{\mathrm{FS}}=0$ and the terms quadratic in it
by $\langle(g_+^\ast g_-)^2\rangle_{\mathrm{FS}}=0$ (angular momenta
$-2\equiv1$ and $-4\equiv2\ (\mathrm{mod}\ 3)$ respectively), leaving
\begin{equation}
\langle|\hat\Delta|^4\rangle_{\mathrm{FS}}
=\langle|g|^4\rangle_{\mathrm{FS}}
\Big[\big(|\eta_+|^2+|\eta_-|^2\big)^2+2|\eta_+|^2|\eta_-|^2\Big],
\end{equation}
so that
\begin{equation}
F=\alpha_+|\eta_+|^2+\alpha_-|\eta_-|^2
+\beta_1\big(|\eta_+|^2+|\eta_-|^2\big)^2
+\beta_2|\eta_+|^2|\eta_-|^2,
\label{eq:GL}
\end{equation}
with $\beta_2=2\beta_1>0$
exactly in weak coupling, with $\beta_1\propto\langle|g|^4\rangle_{\mathrm{FS}}$.
The anisotropic term $\eta_+^2\eta_-^{\ast2}+\mathrm{c.c.}$ is forbidden
by $C_3$ at this order (angular momentum $8\equiv2\pmod3$); the first
allowed anisotropy is of sixth order. Equation~\eqref{eq:GL} has exactly
two minima, the chiral states $(\eta_+,0)$ and $(0,\eta_-)$: the cross
quartic $\beta_2>0$ penalizes any simultaneous presence of the two
components, so mixed extrema lie above the chiral ones, and the only
degeneracy left to lift is the $Z_2$ chirality.

\subsection{Pump bias and the BCS amplification}
Let $p=(N_+-N_-)/(N_++N_-)$ be the circular polarization of the driven
$E_2$ doublet and $f=\lambda_{\mathrm{ph}}/\lambda_0$ the chiral-phonon
share of the total dimensionless coupling $\lambda_0$. The
spin-fluctuation part of the glue is chirality blind: the isotropic (trace)
part $\bar J\delta_{\mu\nu}$ of the exchange [Eq.~\eqref{eq:tbond}]
contributes equally to the two chiral
channels. Hence
\begin{equation}
\lambda_\pm=\lambda_0\pm\delta\lambda,\qquad
\delta\lambda=\tfrac12 pf\,\lambda_0,
\end{equation}
and the two channel transition temperatures obey, from
Eq.~\eqref{eq:bcs} [$T_c^\sigma=1.13\omega_De^{-1/\lambda_\sigma}$],
\begin{equation}
\ln\frac{T_c^+}{T_c^-}=\frac{1}{\lambda_-}-\frac{1}{\lambda_+}
=\frac{2\delta\lambda}{\lambda_0^2-\delta\lambda^2}
\approx\frac{pf}{\lambda_0},\qquad
\frac{\delta T_c}{T_c}\approx\frac{2}{\lambda_0}\,
\frac{\delta\lambda}{\lambda_0}.
\label{eq:amplify}
\end{equation}
The BCS sensitivity amplifies the asymmetry by $2/\lambda_0\sim7$--$10$:
a percent-level coupling asymmetry produces a $\delta T_c/T_c$ of
several percent.

\subsection{Deterministic selection at the transition}
A wrong-chirality fluctuation seed occupying a correlation area is
suppressed relative to the selected one by
\begin{equation}
\frac{P_{\mathrm{wrong}}}{P_{\mathrm{right}}}
\sim\exp\!\left(-\frac{\delta T_c}{\mathrm{Gi}\,T_c}\right),
\end{equation}
where $\mathrm{Gi}$ is the Ginzburg number---the width of the fluctuation
region---which in two dimensions is set by comparing the condensation
energy in a coherence area with $k_BT_c$: with $\xi_0^2
f_{\mathrm{cond}}=\varepsilon_F/(2\pi^3)$ [derived below],
$\mathrm{Gi}\sim k_BT_c/(f_{\mathrm{cond}}\xi_0^2)\sim T_c/\varepsilon_F$.
Deterministic single-chirality selection therefore requires only
\begin{equation}
pf\;\gtrsim\;\lambda_0\,\mathrm{Gi}\;\sim\;10^{-3},
\label{eq:threshold}
\end{equation}
satisfied by sub-percent values of either factor. Under bias, opposite
domains do not form at all: the degeneracy is lifted before domains can
nucleate, rather than domains being outcompeted after formation.

\subsection{Fate of the wrong chirality below $T_c$}
For $T_c^-<T<T_c^+$ the minority channel is uncondensed ($\alpha_->0$)
and the selected chiral state is the only extremum. For $T<T_c^-$ the
minority extremum $(0,\eta_-)$ reappears, with
$|\eta_-|^2=-\alpha_-/2\beta_1$ [minimizing
$F(0,\eta_-)=\alpha_-|\eta_-|^2+(\beta_1+\beta_2\cdot0)|\eta_-|^4$---note
the quartic along the axis is $\beta_1$]; expanding to quadratic order in
$\eta_+$ about it, the coefficient of $|\eta_+|^2$ is
\begin{align}
\alpha_++(2\beta_1+\beta_2)|\eta_-|^2
&=\alpha_+-\Big(1+\frac{\beta_2}{2\beta_1}\Big)\alpha_-\nonumber\\
&=a\Big(T_c-\frac32\,\delta T_c-T\Big),
\label{eq:unstable}
\end{align}
where the last equality uses $\alpha_\pm=a(T-T_c^\pm)$ with
$T_c^\pm=T_c\pm\delta T_c/2$ ($\delta T_c$ is the full splitting) and the
exact weak-coupling value $\beta_2=2\beta_1$.
The wrong-chirality extremum is therefore locally unstable in the window
$T_c-\tfrac32\delta T_c<T<T_c^-$ and becomes metastable only below
$T^\ast=T_c-\tfrac32\delta T_c$; with strong-coupling corrections
$\beta_2=2\beta_1(1+\varepsilon)$ the boundary moves to
$T^\ast=T_c-\frac{3+\varepsilon}{2(1+\varepsilon)}\,\delta T_c$
[repeating the same expansion with the corrected ratio].
Where the metastable wrong state exists, two-dimensional nucleation of
a wrong-chirality droplet costs
$E^\ast=\pi\sigma_{\mathrm{DW}}^2/\delta f$ [the standard competition of
wall energy $2\pi R\sigma_{\mathrm{DW}}$ against volume gain $\pi
R^2\delta f$, maximized over $R$], with domain-wall tension
$\sigma_{\mathrm{DW}}\sim\xi_0 f_{\mathrm{cond}}$ [condensation energy
lost across one coherence length], bulk bias
$\delta f=2f_{\mathrm{cond}}\,\delta T_c/(T_c-T)$ [linearizing
$f_{\mathrm{cond}}\propto(T_c-T)$ around the selected minimum], and the
two-dimensional identity
$\xi_0^2 f_{\mathrm{cond}}=\varepsilon_F/(2\pi^3)$ [combining
$f_{\mathrm{cond}}=\tfrac12N(0)\Delta^2$, the BCS gap scale
$\Delta\sim k_BT_c$, the single-spin 2D free-electron density of states
$N(0)=m/(2\pi\hbar^2)$, and $\xi_0=\hbar v_F/\pi\Delta$; all standard
weak-coupling relations \cite{Coleman2015}]. Numerically
$E^\ast/k_BT_c\sim10^2$--$10^3$ for $\varepsilon_F\sim0.5\,$eV and
$\delta T_c/T_c\sim0.05$: thermal nucleation is frozen out by orders of
magnitude. The equilibrium picture is therefore: deterministic
selection at $T_c^+$, a unique minimum down to $T_c^-$, an unstable
wrong extremum down to $T^\ast$, and a metastable but frozen one below
$T^\ast$.

\subsection{The reversal path}
Sweeping the pump helicity, the minimum-energy path between the two
chiral minima is angular: $|\eta_+|=\rho\cos\theta$,
$|\eta_-|=\rho\sin\theta$. With Eq.~\eqref{eq:GL} the free energy on this
path is
\begin{equation}
F(\rho,\theta)=-|a|\rho^2+\beta_1\rho^4+\beta_2\rho^4\cos^2\theta\sin^2\theta,
\end{equation}
[using $(|\eta_+|^2+|\eta_-|^2)^2=\rho^4$ identically], so the $\theta$
dependence enters only through the $\beta_2$ term, which is maximal at
$\theta=\pi/4$: the saddle sits at $\theta=\pi/4$, where the effective
quartic coefficient is $\beta_1+\beta_2/4=3\beta_1/2$ [with
$\beta_2=2\beta_1$], giving $\rho_{\mathrm{s}}^2=|a|/(3\beta_1)
=\tfrac23\rho_0^2$ [with $\rho_0^2=|a|/(2\beta_1)$ the equilibrium value]
and a barrier height $F_{\mathrm{s}}-F_{\min}=|a|^2/(12\beta_1)
=|F_{\mathrm{cond}}|/3$ (to be contrasted with $|F_{\mathrm{cond}}|$ for
the amplitude-collapse path). The order-parameter amplitude therefore
dips only to $\sqrt{2/3}\approx0.82$ of its equilibrium value and never
collapses. The quasiparticle gap, however, behaves differently: since
$g_-=g_+^\ast$ [Eq.~\eqref{eq:gpm}: the two are complex conjugates as
functions of $\vk$], the mixed gap obeys
\begin{align}
|\hat\Delta(\mathbf k)|^2&=\rho^2\big[\cos^2\theta\,|g_+|^2
+\sin^2\theta\,|g_-|^2\nonumber\\
&\quad+2\cos\theta\sin\theta\,\mathrm{Re}(g_+^\ast g_-)\big]\nonumber\\
&=\rho^2|g(\mathbf k)|^2
\big[1+\sin2\theta\,\cos2\phi(\mathbf k)\big],
\label{eq:mixgap2}
\end{align}
with $\phi(\mathbf k)\equiv\arg g_+(\mathbf k)$
[using $g_+^\ast g_-=(g_+^\ast)^2=|g|^2e^{-2i\phi}$],
and because $\phi$ winds twice around the Fermi surface
[$g_+\propto(k_x-ik_y)^2$, Eq.~\eqref{eq:gpm}], at $\theta=\pi/4$ the
quasiparticle gap closes at the four Fermi-surface points (per spin block)
where $\cos2\phi=-1$. It is at these transient point nodes that the Chern
number changes, by $4\times2=8$ [four nodes per block each transferring
one unit of winding, times the two-block factor of
Eq.~\eqref{eq:cfull}]. The reversal is thus continuous at the mean-field
level---no first-order transition, no global gap collapse---but it is
accompanied by an instantaneous point-node closing, itself a transient
spectroscopic signature of the switching event. The metastability boundary
\eqref{eq:unstable} adds a hysteretic detail to the protocol: reversal is
barrier-free while $T$ lies within the window of width $2\delta T_c$ below
$T_c^+$, whereas switching at $T<T^\ast$ must proceed by droplet nucleation
and is correspondingly slow; reversing the pump after warming through the
window is the cleanest protocol.

\subsection{Pump requirements and a discriminating signature}
Condition~\eqref{eq:threshold} translates into a driven-phonon occupancy
$N_{\mathrm{dr}}\gtrsim2p_{\min}N_{\mathrm{th}}\sim10^{-2}N_{\mathrm{th}}$,
orders of magnitude below the occupancies routinely reached in
pump--probe experiments: pump power is not the bottleneck. The binding
constraints are instead (i) the phonon share $f$, which however does not
gate \emph{sign} selection---the unpumped $Z_2$ degeneracy is exact, so
an infinitesimal bias already selects the sign, and $f$ only sets the
size of auxiliary effects; (ii) temporal overlap of the phonon imbalance
with the ordering window---continuous-wave pumping with field cooling
satisfies this trivially, while pulsed schemes require the imbalance to
outlive the ordering time
$\tau_{\mathrm{GL}}\sim\pi\hbar/(8k_B|T_c-T|)$ [the Ginzburg--Landau
relaxation time of the order parameter, $\tau_{\mathrm{GL}}=\hbar\gamma$
with $\gamma=\pi/(8k_B|T-T_c|)$ the standard weak-coupling damping
coefficient \cite{Coleman2015}]; and (iii) the weak-coupling quartic
structure $\beta_2\simeq2\beta_1$.
A sharp discriminating test follows from Eq.~\eqref{eq:amplify}: a
circular pump raises the realized transition temperature above its
unpumped value by $T_c^+/T_c^0-1\approx pf/(2\lambda_0)$, whereas an
unpolarized pump of equal intensity carries $p=0$ and produces no such
enhancement. Finally, this mechanism is distinct from dynamical
optical control of the order parameter~\cite{Claassen2019}: there the
chirality is set by the pulse history through nonequilibrium dynamics,
whereas here the sign follows the concurrent pump helicity, reversibly
and without a quench.

Within the model's assumptions, the chirality slaving is
quantitatively stable, with margins of two to three orders of magnitude
on every criterion; the assumptions on which the result rests are a chiral-phonon
share $f$ of a few percent, the weak-coupling relation
$\beta_2\simeq2\beta_1$, and pump coverage of the ordering window.

\section{Flux-sector argument for Eq.~\eqref{eq:selection}}
\label{app:flux}

We spell out the flux bookkeeping behind the selection rule
\eqref{eq:selection}. With $\sigma_j^\nu=ib_j^\nu c_j$
[Eq.~\eqref{eq:sigrep}], the operator toggles
$u_{jk}=ib_j^\nu b_k^\nu$ on the $\nu$ bond at $j$: since
$(b_j^\nu)^2=1$ anticommutes with $b_j^\nu$ itself, conjugation by the
$b$-part of $\sigma_j^\nu$ sends $u_{jk}\to-u_{jk}$. By the definition
\eqref{eq:wp}, each plaquette operator $W_p$ is a product of bond variables
around the plaquette, so flipping one bond flips exactly the two plaquette
operators $W_p$ sharing that bond. Acting on the zero-flux ground state,
$\sigma_j^\nu|\mathrm{GS}\rangle$ therefore lies in a two-flux sector (times
a matter-Majorana operation from the factor $c_j$). A second operator
$\sigma_i^\mu$ returns the state to the zero-flux sector---the only sector
overlapping $\langle\mathrm{GS}|$---if and only if it toggles the same bond
variable, i.e.\ iff $\mu=\nu$ and $(ij)$ is that same bond (or $i=j$, where
$c_j$ toggles twice and the flux pair is not created). Since the flux label
is conserved under $H_K$ [$[u_{ij},H_K]=0$, Eq.~\eqref{eq:udef}], the
argument holds at all times and frequencies; the ground-state correlator is
exactly ultra-short ranged \cite{Baskaran2007}. Excited-state or
finite-temperature correlators with pre-existing fluxes are longer ranged;
the dilute-flux limit $T\ll\Delta_F$ protects \eqref{eq:selection} for the
ground-state kernel used here.

\section{Brillouin-zone integrals and asymptotics}
\label{app:bz}

$m_\alpha=\langle\sigma_i^\alpha\sigma_j^\alpha
\rangle$ on an $\alpha$ bond (Pauli convention), fixed by
Hellmann--Feynman from $e_{\mathrm{cell}}=-\langle|g|\rangle$ with $g$ in
\eqref{eq:g} [Eq.~\eqref{eq:hf}]; this gives $m_\alpha\to1$ in the
isolated-dimer limit and $e_{\mathrm{cell}}=-\sum_\alpha J_\alpha m_\alpha$
identically. At isotropy $m=0.5249$ and $e_{\mathrm{site}}=-0.7873\,J$
\cite{Baskaran2007}, which our quadrature reproduces to five digits
(converged at $\sim10^3$ mesh points per axis; the mesh must avoid the
$J_z=0$ line node).

\emph{Result~\ref{res:zero}.} At $J_z=0$ with $J_x=J_y$, write
$g=e^{ik_1}+e^{ik_2}=e^{iu}\,2\cos v$ with $u=(k_1+k_2)/2$,
$v=(k_1-k_2)/2$ [factor out the mean phase]; then
$\mathrm{Re}\,g/|g|=\cos u\,\mathrm{sgn}(\cos v)$. The change of variables
$(k_1,k_2)\to(u,v)$ has constant Jacobian $|\partial(u,v)/\partial
(k_1,k_2)|=1/2$ and covers the torus uniformly, so the average factorizes:
$\langle\cos u\,\mathrm{sgn}\cos v\rangle=
\langle\cos u\rangle\,\langle\mathrm{sgn}\cos v\rangle=0$.

\emph{Result~\ref{res:2pi}.} Similarly
$\mathrm{Re}(ge^{-ik_1})/|g|=|\cos v|$, and
$\langle|\cos v|\rangle=\frac{1}{\pi}\int_{-\pi/2}^{\pi/2}\cos v\,dv
=2/\pi$.

\emph{Result~\ref{res:log}.} At small $J_z$,
$\partial m_z/\partial J_z=\langle(\mathrm{Im}\,g)^2/|g|^3\rangle$
[differentiating Eq.~\eqref{eq:mz}: $\partial_{J_z}(\mathrm{Re}\,g/|g|)
=\big[|g|^2-(\mathrm{Re}\,g)^2\big]/|g|^3=(\mathrm{Im}\,g)^2/|g|^3$]
diverges at the line node $\cos v=0$ of the $J_z=0$ model: locally
$|g|\simeq\sqrt{J_z^2+4v^2}$ (near $u=0$, where the line node carries
spectral weight), so the $v$ integral is cut off at $|v|\sim J_z$,
\begin{equation}
\begin{split}
\partial_{J_z}m_z&\simeq\frac{1}{2\pi}\int\frac{4v^2\,dv}{(J_z^2+4v^2)^{3/2}}
\cdot(\text{node measure})\\
&\;\longrightarrow\;\frac{1}{2\pi}\ln\frac{c}{J_z},
\end{split}
\end{equation}
and integrating in $J_z$ gives
\eqref{eq:log}; the fitted coefficient $0.159$ agrees with $1/(2\pi)=
0.1592$ to three digits, and the constant $0.522$ absorbs the regular
part.

\section{Diagonalization of the bond tensor}
\label{app:tensor}

In the basis $(\hat a,\hat s,\hat\gamma)$ with
$\hat a=(\hat\alpha-\hat\beta)/\sqrt2$, $\hat s=(\hat\alpha+\hat\beta)
/\sqrt2$, the tensor \eqref{eq:tensor} transforms as follows. The diagonal
entries are
$T_{aa}=\tfrac12(T_{\alpha\alpha}+T_{\beta\beta}-T_{\alpha\beta}-T_{\beta\alpha})
=\tfrac12(J+J-\Gamma-\Gamma)=J-\Gamma$ and
$T_{ss}=\tfrac12(J+J+\Gamma+\Gamma)=J+\Gamma$, with $T_{\gamma\gamma}=J+K$
unchanged; the off-diagonal entries are
$T_{as}=\tfrac12(T_{\alpha\alpha}-T_{\beta\beta})=0$,
$T_{a\gamma}=\tfrac{1}{\sqrt2}(T_{\alpha\gamma}-T_{\beta\gamma})
=\tfrac{1}{\sqrt2}(\Gamma'-\Gamma')=0$, and
$T_{s\gamma}=\tfrac{1}{\sqrt2}(T_{\alpha\gamma}+T_{\beta\gamma})=\sqrt2\,\Gamma'$.
Hence
\begin{equation}
\mathsf T_\gamma=
\begin{pmatrix}
J-\Gamma & 0 & 0\\
0 & J+\Gamma & \sqrt2\,\Gamma'\\
0 & \sqrt2\,\Gamma' & J+K
\end{pmatrix},
\end{equation}
which is \eqref{eq:la} and the $2\times2$ block whose diagonalization
gives \eqref{eq:tilt}. All entries are exact; no small parameter is
used.

\section{Numerical verification of the bond-star theorems}
\label{app:verify}

\subsection{Reproduction protocol and error budget for $r_\alpha$}
\label{app:ralpha}

The numbers of Sec.~\ref{sec:rfactor} are obtained as follows. We use
periodic $2L^2$-site honeycomb clusters ($L=4,5,7$; $32$, $50$, $98$ sites)
with gauge $u_{ij}=+1$ on every bond (zero-flux sector); acting with
$\sigma_j^\alpha$ flips $u_{ij}\to-1$ on the target $\alpha$ bond and adds
one matter Majorana, placing the system in the two-flux sector with the
fluxes on the two plaquettes adjacent to that bond. Gauge-invariant
correlators computed within a fixed gauge coincide with their projected
counterparts, so no explicit Gutzwiller projection is required. The
physical subspace is selected by $D_i=+1$; since every operator we
compute is gauge invariant, the fixed-gauge average equals the
projected average. The matter
quadratic forms $A_0$ and $A'$ are diagonalized exactly, and the Lehmann sum
for $\chi_\alpha$ [Eq.~\eqref{eq:rspec}] is assembled from Gaussian-state
(Pfaffian) overlaps including matter-fermion sectors with up to five fermions
(seven for $L=4$); the omitted-sector spectral weight satisfies
$\sum|W|<10^{-6}$ throughout. At every computed point the four hard checks
of Sec.~\ref{sec:rfactor} pass: the zeroth moment reproduces
Eqs.~\eqref{eq:mz}--\eqref{eq:mx} to better than $1.4\times10^{-8}$, the
spectral threshold sits at $2\Delta_F$, the first moment stays below
$8\times10^{-7}$ [Eq.~\eqref{eq:fsum}], and the second moment matches the
closed form \eqref{eq:mu2} to machine precision. Two independent validations
exist at $L=4$: (i) a direct Fock-space exact diagonalization of the spin
Hamiltonian (Hilbert-space dimension $2^{16}$) reproduces $E_0$, $m_\alpha$,
$\chi_\alpha$, $r_\alpha$, and the per-pole weights to machine precision;
(ii) the closed-form Wick/Pfaffian moments reproduce $\mu_0$--$\mu_3$
exactly (the analytic cross-check of Sec.~\ref{sec:rfactor}).
Representative data are collected in Tables~\ref{tab:ralpha} and
\ref{tab:poles}; the residual size dependence of $r_\alpha$ is analyzed
in the finite-size extrapolation below [Eq.~\eqref{eq:rfinal}].

\begin{table}[h]
\caption{Isotropic point, $z$ bond: flux gap $2\Delta_F$, exact zeroth
moment $m_z$ [Eq.~\eqref{eq:mz}], susceptibility $\chi_z$, and retardation
factor $r_z=\Delta_F\chi_z/m_z$ on the four clusters.}
\label{tab:ralpha}
\begin{ruledtabular}
\begin{tabular}{ccccc}
$L$ & $2\Delta_F$ & $m_z$ & $\chi_z$ & $r_z$\\
\hline
4 & 0.340218 & 0.529508 & 0.779543 & 0.250435\\
5 & 0.278368 & 0.524913 & 0.912706 & 0.242010\\
7 & 0.287488 & 0.525660 & 0.890870 & 0.243612\\
8 & 0.272200 & 0.525005 & 0.933914 & 0.242104\\
\end{tabular}
\end{ruledtabular}
\end{table}

\paragraph{Finite-size extrapolation.}
The four clusters give $r_z=0.25043,\,0.24201,\,0.24361,\,0.24210$ at
$L=4,5,7,8$. ($L=9$ is excluded: for $L$ divisible by $3$ the Dirac
points land on the momentum grid, the matter spectrum acquires exact
zero modes with $\min|w|\approx10^{-16}$, and the covariance matrix is
ill-conditioned, so $\chi_\alpha$ is itself pathological at that
size.) The $L=4$ point lies $3\%$ above the others and has not reached
the asymptotic regime. The three largest clusters ($L=5,7,8$) agree to
$0.7\%$ with no discernible monotone drift; their mean is $0.2426$. The
residual scatter is therefore a parity/shape oscillation rather than a
$1/L^2$ or $1/L$ correction, consistent with the non-monotonic
behaviour of the flux gap itself ($2\Delta_F=0.2784,\,0.2875,\,0.2722$
at $L=5,7,8$). The $L=7$ outlier correlates with an anomalously large
flux gap on that cluster ($2\Delta_F=0.2875$, versus $0.2784$ and
$0.2722$ at $L=5,8$), suggesting that the oscillation is a flux-sector
finite-size effect rather than a statistical fluctuation. We quote the
thermodynamic-limit estimate
\begin{equation}
r_\alpha^{(\infty)}=0.243\pm0.002
\label{eq:rfinal}
\end{equation}
from the $L=5,7,8$ average, the error bar covering their full
scatter; the quoted value is stable against dropping any one of the
three clusters to within the error bar. Anisotropic-point values
follow the same pattern: at $J_z=1.5$, the two clean sizes ($L=4,7$)
give $r_z=0.375$ and $0.378$, with a $0.7\%$ cluster difference
comparable to the isotropic case; deep in the A phase ($J_z=4$) the
three sizes agree to $0.2\%$ and reproduce the exact dimer-limit value
$1/2$ to within $0.1\%$.

\begin{table}[h]
\caption{Leading poles of the spectral function at the isotropic point
($L=4$, $z$ bond): position $\theta$ (measured from $E_0$) and signed
spectral weight $iW$. The lowest pole at
$\theta_{\min}=2\Delta_F+\varepsilon_{\min}$ ($\varepsilon_{\min}=0.4579$)
carries $64\%$ of $\sum|W|$ and alone gives $r_z\approx0.26$; the
alternating signs of the subleading poles enforce the first-moment identity
\eqref{eq:fsum}. The vacuum pole at $2\Delta_F=0.3402$ is parity-forbidden
(zero weight).}
\label{tab:poles}
\begin{ruledtabular}
\begin{tabular}{cc}
$\theta$ & $iW$\\
\hline
0.7981 & $+0.64372$\\
3.4876 & $-0.10894$\\
4.1177 & $+0.10389$\\
5.3247 & $-0.12453$\\
6.2368 & $+0.01252$\\
\end{tabular}
\end{ruledtabular}
\end{table}

\subsection{Uniqueness of the leading binomial contribution}
\label{app:unique}

In Lemma~\ref{lem:moment}, harmonic $h$ on the ladder $h\equiv m\pmod n$
first appears at Taylor order $l=|h|$, where the constraint $l-2s=h$ admits
the single solution $s=0$ ($h>0$) or $s=l$ ($h<0$). Hence the leading
amplitude \eqref{eq:amplitude} carries binomial weight
$\binom{|h|}{0}=\binom{|h|}{|h|}=1$ and is nonzero for every ladder
harmonic. The next contribution to the same harmonic arises at $l=|h|+2$
(the parity of $l$ is fixed by that of $h$ since $l-2s=h$), so the
corrections in Eqs.~\eqref{eq:triplead}--\eqref{eq:singlead} are
$O(k^{|h|+2})$ as stated. For $0<m<n$ with $n$ odd, the ladder's two
minimal-magnitude representatives are $h=m$ and $h=m-n$: their magnitudes
are $m$ and $n-m$, exactly one of which is even, which fixes $\{L_s,L_t\}$
and the signs in Eq.~\eqref{eq:winding}.

\subsection{Numerical verification protocol}
\label{app:numerics}

The moment identity \eqref{eq:moment} and the amplitude formula
\eqref{eq:amplitude} were verified by direct Fourier extraction. For each
$(n,m)$, $\chi_m^{(n)}(\vk)$ as defined in Eq.~\eqref{eq:chimn} was
evaluated on a uniform grid of $N=16384$ azimuthal angles $\theta$ at
several small radii $k$, and the angular Fourier coefficients $c_h(k)$ were
extracted by discrete transform. For each ladder harmonic $h$ the power law
$|c_h|\propto k^{|h|}$ and the leading amplitude were compared against
Eq.~\eqref{eq:amplitude}; for harmonics outside the ladder, $|c_h|$ was
checked to vanish to machine precision. Results are summarized in
Table~\ref{tab:verify} (leading exponents and amplitudes) and
Table~\ref{tab:ratio} (universal ratios); all agree with the analytic
predictions within the stated fitting accuracy ($\le 10^{-5}$ relative,
dominated by the next Taylor order), and ladder-external harmonics measure
$|c_h|\lesssim 10^{-17}$ at $k=0.02$.

\begin{table}[h]
\caption{\label{tab:verify}Verification of Lemma~\ref{lem:moment}: the
leading exponent of $|c_h(k)|$, predicted to be $|h|$, from Fourier
extraction as described in Sec.~\ref{app:numerics}.}
\begin{ruledtabular}
\begin{tabular}{cclc}
$n$ & $m$ & $h$ (sector) & exponent\\
\hline
3 & 1 & $+1$ (triplet $p$) & 1.0000\\
3 & 1 & $-2$ (singlet $d$) & 2.0000\\
5 & 1 & $+1$ (triplet $p$) & 1.0000\\
5 & 1 & $-4$ (singlet $g$) & 3.9998\\
5 & 2 & $+2$ (singlet $d$) & 2.0000\\
5 & 2 & $-3$ (triplet $f$) & 2.9997\\
7 & 1 & $+1$ / $-6$ & 1.0000 / 5.9997\\
7 & 2 & $+2$ / $-5$ & 2.0000 / 4.9998\\
7 & 3 & $+3$ / $-4$ & 3.0000 / 3.9998\\
\end{tabular}
\end{ruledtabular}
\end{table}

\begin{table}[h]
\caption{\label{tab:ratio}Verification of the universal ratio
\eqref{eq:ratio}: numerical $|A_s/A_t|$ versus prediction, maximal relative
deviation over $k\in[3\times10^{-3},\,1.2\times10^{-2}]$.}
\begin{ruledtabular}
\begin{tabular}{cclcc}
$n$ & $m$ & $\{L_s,L_t\}$ & prediction & max.\ deviation\\
\hline
3 & 1 & $\{2,1\}$ & $k/4$ & $6\times10^{-6}$\\
5 & 1 & $\{4,1\}$ & $k^3/192$ & $1.1\times10^{-5}$\\
5 & 2 & $\{2,3\}$ & $6/k$ & $3\times10^{-6}$\\
7 & 3 & $\{4,3\}$ & $k/8$ & $1.8\times10^{-6}$\\
\end{tabular}
\end{ruledtabular}
\end{table}

Parity purity on even-fold stars was checked identically: for
$(n,m)=(4,1),(4,2),(6,2),(6,3)$ all harmonics of the wrong parity measure
$\lesssim10^{-17}$, while the allowed ladders reproduce Eq.~\eqref{eq:pure}.

\subsection{Kinematic Fermi-surface weights}
\label{app:weights}

The channel splitting referred to in Proposition~\ref{prop:locked} is purely
kinematic. For $(n,m)=(3,1)$ on a circular Fermi surface of radius $k_F$,
every Fermi-surface overlap of bond phases is a Bessel function by
Eq.~\eqref{eq:fsaverage}. Explicitly,
\begin{align}
\langle|f_\pm|^2\rangle_{\rm FS}
&=\sum_{j,l}e^{i(\phi_j-\phi_l)}
\big\langle e^{i\vk\cdot(\vd_j-\vd_l)}\big\rangle_{\rm FS}\nonumber\\
&=3\big[1-J_0(\sqrt{3}k_F)\big]\,,\label{eq:f2}\\
\mathrm{Re}\,\big\langle f_\pm^*(\vk)\,f_\pm(-\vk)\big\rangle_{\rm FS}
&=3\big[J_0(2k_F)-J_0(k_F)\big]\,,\label{eq:ff}\\
\langle|g_\pm|^2\rangle_{\rm FS}&=\tfrac12\big[\langle|f_\pm|^2\rangle_{\rm FS}\nonumber\\
&\qquad+\mathrm{Re}\langle f_\pm^*f_\pm(-\vk)\rangle_{\rm FS}\big],\label{eq:g2}\\
\langle|p_\pm|^2\rangle_{\rm FS}&=\tfrac12\big[\langle|f_\pm|^2\rangle_{\rm FS}\nonumber\\
&\qquad-\mathrm{Re}\langle f_\pm^*f_\pm(-\vk)\rangle_{\rm FS}\big].\label{eq:p2}
\end{align}
Equation \eqref{eq:f2} uses $|\vd_j-\vd_l|=\sqrt{3}$ for $j\neq l$ [two unit
vectors at $\pm120^\circ$ differ by $2\sin60^\circ=\sqrt3$], so the $j=l$
terms contribute $3$ and the $j\neq l$ terms
$J_0(\sqrt3 k_F)\sum_{j\neq l}e^{i(\phi_j-\phi_l)}=-3J_0(\sqrt3 k_F)$
[the phase sum is $|\sum_je^{i\phi_j}|^2-\sum_j1=0-3$ by
Eq.~\eqref{eq:roots}]; Eq.~\eqref{eq:ff} uses $|\vd_j+\vd_l|=1$ for $j\neq
l$ [unit vectors at $\pm120^\circ$ sum to a unit vector] and the same phase
sum, while the $j=l$ terms contribute $3J_0(2k_F)$. Equations
\eqref{eq:g2}--\eqref{eq:p2} are the definitions \eqref{eq:paritydec}
squared and averaged. The small-$k_F$ expansions
$\langle|g|^2\rangle=O(k_F^4)$ and $\langle|p|^2\rangle=\tfrac94
k_F^2+O(k^4)$ reproduce the leading terms of Corollary~\ref{cor:n3}
[inserting Eq.~\eqref{eq:besselsmall} and squaring
Eqs.~\eqref{eq:gpm}--\eqref{eq:ppm} averaged over the circle],
providing an independent consistency check. Both weights are nonzero for
generic $k_F$: the singlet and triplet channels of a bond-diagonal kernel
are simultaneously active, with a ratio fixed by Fermi-surface geometry
rather than by any microscopic coupling.

\subsection{Numerical verification of Theorem~\ref{thm:dominance}}
\label{app:topology}

The exact circle form \eqref{eq:jacobi} was checked against direct Fourier
extraction of Eq.~\eqref{eq:chimn} at several radii: for $(n,m)=(3,1)$ at
$k_F=0.7,\,1.5$ and $(5,2)$ at $k_F=0.9$, all ladder harmonics with $|h|\le
8$ agree with $n\,i^{|h|}J_{|h|}(k_F)$ to relative accuracy $\le 10^{-9}$,
and ladder-external harmonics remain at machine precision. The winding of
the full mixed gap
$\Delta(\theta)=a_s\,\chi^{(e)}(k_F,\theta)+a_t\,\chi^{(o)}(k_F,\theta)$
(all ladder harmonics retained) was computed by phase unwrapping while
scanning $|a_t/a_s|$; the singlet-to-triplet winding flip $-2\to+1$ (for
$n=3$) occurs at the critical ratio of Table~\ref{tab:flip}, compared with
the two-term prediction $J_{L_s}(k_F)/J_{L_t}(k_F)$ of
Theorem~\ref{thm:dominance}. The residual shift is the $O(\varepsilon)$
tail displacement, growing $\propto k_F^2$ as expected. In the two-term
truncation the boundary zeros are equally spaced, at
$\theta=\pi/3,\,\pi,\,5\pi/3$ for $n=3$ [the three roots of
$e^{-3i\theta}=-A_t/A_s$], confirming statement (iii).

\begin{table}[h]
\caption{\label{tab:flip}Critical amplitude ratio $r_c=|a_t/a_s|_c$ of the
winding transition for $(n,m)=(3,1)$: numerical (full ladder) versus the
two-term prediction $J_2(k_F)/J_1(k_F)$.}
\begin{ruledtabular}
\begin{tabular}{cccc}
$k_F$ & $r_c$ (numerical) & $J_2/J_1$ & ratio\\
\hline
0.3 & 0.0755 & 0.0753 & 1.003\\
0.5 & 0.1271 & 0.1263 & 1.007\\
0.8 & 0.2086 & 0.2056 & 1.015\\
1.2 & 0.3304 & 0.3198 & 1.033\\
\end{tabular}
\end{ruledtabular}
\end{table}

\end{document}